\documentclass[a4paper, 10pt]{article}

\usepackage[utf8]{inputenc}    
\usepackage[T1]{fontenc}
\usepackage{lmodern}
\usepackage{graphicx}
\usepackage{textcomp}
\usepackage{amsmath, amsfonts, amssymb, amsthm}
\usepackage{mathtools}
\usepackage{array}
\usepackage{dsfont} 
\usepackage{enumitem} 
\usepackage{geometry} 
\usepackage{braket} 
\usepackage{esint}
\usepackage{titletoc}
\usepackage{titling}
\usepackage{titlesec}

\usepackage{hyperref} 

\hypersetup{colorlinks = true}

\newcommand{\intervalle}[4]{\mathopen{#1}#2
                            \mathclose{}\mathpunct{},#3
                            \mathclose{#4}}
\newcommand{\intff}[2]{\intervalle{[}{#1}{#2}{]}}
\newcommand{\intof}[2]{\intervalle{(}{#1}{#2}{]}}
\newcommand{\intfo}[2]{\intervalle{[}{#1}{#2}{)}}
\newcommand{\intoo}[2]{\intervalle{(}{#1}{#2}{)}}

\newcommand{\petito}[1]{o\mathopen{}\left(#1\right)}
\newcommand{\grandO}[1]{O\mathopen{}\left(#1\right)}

\newcommand{\abs}[1]{\left\lvert#1\right\rvert}
\newcommand{\norme}[1]{\left\lVert#1\right\rVert}
\newcommand{\pdtsc}[2]{\left\langle#1,#2\right\rangle}

\newcommand{\enstq}[2]{\left\{#1\mathrel{}\middle|\mathrel{}#2\right\}}

\newcommand{\et}{\quad \text{and} \quad}

\newcommand{\ensemblenombre}[1]{\mathbb{#1}}
\newcommand{\N}{\ensemblenombre{N}}

\newcommand{\R}{\ensemblenombre{R}}

\newcommand{\pt}{\, .}

\newcommand{\limit}[1]{\quad\underset{#1}{\longrightarrow}\quad}

\newcommand{\diff}{\mathop{}\mathopen{}\mathrm{d}}

\DeclareMathOperator{\spn}{span}

\DeclareMathOperator{\dist}{d}
\DeclareMathOperator{\supp}{supp}

\theoremstyle{plain}
\newtheorem{theo}{Theorem}
\newtheorem{prop}[theo]{Proposition}
\newtheorem{cor}[theo]{Corollary}
\newtheorem{lemma}[theo]{Lemma}

\theoremstyle{remark}

\newtheorem{rem}[theo]{Remark}

\theoremstyle{definition}

\renewcommand\thesubsection{\arabic{section}.\arabic{subsection}}

\titleformat{\subsection}[runin]{\normalfont\bfseries}{\thesubsection}{.5em}{}[.]

\title{The diatomic Hartree model at dissociation}
\author{Jean Cazalis\thanks{CEREMADE, CNRS, UMR 7534, Universit\'e Paris-Dauphine, PSL University, 75016 Paris, France; email: \href{mailto:cazalis@ceremade.dauphine.fr}{cazalis@ceremade.dauphine.fr}}}
\date{\today}

\begin{document}
	
\maketitle

\begin{abstract}
	We study the Hartree model for two electrons with spin, living in the two-dimensional or three-dimensional space with Coulomb interactions and submitted to the potential induced by two nuclei of charge +1. In the limit where the nuclei move away from each other, we show that the two lowest eigenfunctions of the mean-field hamiltonian are asymptotically given by an even, respectively odd, superposition of the minimizer associated with the corresponding Hartree one nucleus model. We then give upper and lower bounds on the exponentially small gap between the first two eigenvalues, due to (nonlinear) quantum tunneling.
\end{abstract}	

\tableofcontents

\section{Introduction and statement of the main theorem}

\subsection{Diatomic Hartree model}
We consider a neutral diatomic system formed of two electrons and two point nuclei of charge $+1$ located at distance $L$ in $\R^d$ where $d \in \{2,3\}$. We assume the charges interact through the three-dimensional Coulomb potential $\frac{1}{\abs{x}}$. Our aim is to study the behavior of the electrons in the (restricted) Hartree approximation in the limit $L\to\infty$. Without loss of generality, we assume that the two nuclei are located at $\pm\mathbf{x}_L$ where $\mathbf{x}_L= (\frac{L}{2},0)$ if $d=2$ and $\mathbf{x}_L = (\frac{L}{2},0,0)$ if $d=3$. The potential generated by these nuclei is denoted by
\begin{align}
\label{eq: definition_diatomic_potential}
V_L(x) \coloneqq  -\left(\frac{1}{\abs{x - \mathbf{x}_L}} + \frac{1}{\abs{x + \mathbf{x}_L}}\right) \pt
\end{align}
The nucleus located at $-\mathbf{x}_L$ (resp. $+\mathbf{x}_L$) will be called the \emph{left} (resp. \emph{right}) nucleus. When $d=3$ we recover the usual Coulomb potential for two pointlike charges. In dimension $d=2$, this potential corresponds to the one generated by a system of nuclei confined to a plane.

The (restricted) Hartree energy of the electrons in the state $v \in H^1(\R^d)$ is given by \cite{benguria1981thomasfermi, lieb1981thomasfermi}
\begin{align}
\label{eq: restricted_hartree_functional}
\mathcal{E}_L(v) \coloneqq \int_{\R^d} \abs{\nabla v(x)}^2\diff x + \int_{\R^d} V_L(x) \abs{v(x)}^2\diff x + \frac{1}{2}\iint_{\R^d \times \R^d} \frac{\abs{v(x)}^2\abs{v(y)}^2}{\abs{x-y}} \diff x \diff y + \frac{1}{L} \pt
\end{align}
We have chosen a system of units such that $\hbar = e = \frac{1}{4\pi \epsilon_0} = 1$ and $m = \frac{1}{2}$ where $m$ and $e$ are respectively the mass and charge of an electron, $\epsilon_0$ is the dielectric permittivity of the vacuum and $\hbar$ is the reduced Planck constant. The first term in \eqref{eq: restricted_hartree_functional} is the kinetic energy of the two electrons in the state $v$, the second term its interaction energy with respect to the potential $V_L$ induced by the nuclei at $\pm \mathbf{x}_L$, the third and fourth terms are respectively the self-interaction energy of the charge distribution $\abs{v}^2$ and of the nuclei. The reason we have only one function $v$ as variable is because we look at a singlet state where the anti-symmetry is in the spin. In our setting, the Hartree model consists in minimizing the energy functional $\mathcal{E}_L$ on the set of admissible states
\begin{align*}
\mathcal{P} \coloneqq \enstq{v \in H^1(\R^d)}{\int_{\R^d} \abs{v}^2=2} \pt
\end{align*}
The normalization $\int_{\R^d} \abs{v}^2=2$ is because we have two electrons and it ensures the neutrality of the system. It is well known that the functional $v \mapsto \mathcal{E}_L(v)$ is bounded from below and strongly continuous on $\mathcal{P}$ and weakly lower semi-continuous on $\enstq{v \in H^1(\R^d)}{\int_{\R^d} \abs{v}^2\leq2}$ (see \cite[Section~VII]{lieb1981thomasfermi}). The variational problem corresponding to the Hartree model in consideration reads
\begin{align}
\label{eq: restricted_hartree}
\boxed{E_L \coloneqq \frac{1}{2} \inf_{v \in \mathcal{P}} \mathcal{E}_L(v)  \pt}
\end{align}
The factor $\frac{1}{2}$ is because we compute the energy per electron.
It is known (see \cite[Section VII]{lieb1981thomasfermi}) that the variational problem \eqref{eq: restricted_hartree} has a unique positive minimizer $u_L^+$ which satisfies the Euler-Lagrange equation
\begin{align}
\label{eq: EL_equation}
h_L u_L^+ = \left(-\Delta + V_L +  \lvert u_L^+ \rvert^2 \ast |\cdot|^{-1}\right) u_L^+ = \mu_L^+u_L^+ \, ,
\end{align}
where $\mu_L^+$ is a Lagrange multiplier and the first eigenvalue of the mean-field hamiltonian $h_L$. The operator $h_L$ is self-adjoint when defined on the domain $\mathcal{D}(h_L) = H^2(\R^3)$ if $d=3$ and on
\begin{align}
\label{eq: domain_1}
\mathcal{D}(h_L) 
&= \enstq{v \in H^1(\R^2)}{(-\Delta + V_L + \lvert u_L^+ \rvert^2 \ast |\cdot|^{-1})v \in L^2(\R^2)}\\
\label{eq: domain_2}
&= \enstq{v \in H^1(\R^2)}{(-\Delta + V_L)v \in L^2(\R^2)}\, ,
\end{align}
when $d=2$. The equality above comes from the fact that $\lvert v \rvert^2 \ast |\cdot|^{-1}$ belongs to $L^\infty(\R^2)$ whenever $v\in H^1(\R^2)$. By standard arguments \cite{kato1982short, reed1978methodsIV}, the essential spectrum of $h_L$ is $\intfo{0}{\infty}$, $\mu_L^+ $ is negative and non-degenerated and, up to a change of phase, we have $u_L^+ >0$ everywhere on $\R^d$.

Since the nuclei have the same charge, the model is invariant under reflection with respect to the hyperplan $\{x_1=0\}$. To mathematically translate this property, we introduce the reflection operator $\mathcal{R}$, which is the unitary operator defined for all $v \in L^2(\R^d)$ by $\mathcal{R}[v] (x) = v(R^* x)$ where
\begin{align}
\label{eq: def_reflection}
R = \begin{pmatrix}
-1 & 0 \\ 
0 & 1
\end{pmatrix} \quad \text{if}\quad d=2 \et
R = \begin{pmatrix}
-1 & 0 & 0 \\ 
0 & 1 & 0 \\
0 & 0 & 1
\end{pmatrix} \quad \text{if}\quad d=3 \pt
\end{align}
Since $[h_L,\mathcal{R}] = 0$, we can decompose $\mathcal{D}(h_L)$ into an orthogonal sum of two isotypic subspaces, each one carrying irreducible representations of $\mathcal{R}$. As a consequence, we can always choose the eigenfunctions of $h_L$ to be even or odd with respect to the $x_1$ variable. Since the ground state energy of $h_L$ is non degenerate, $u_L^+$ is $\mathcal{R}$-invariant, that is,
\begin{align*}
\mathcal{R}[u_L^+] = u_L^+ \pt
\end{align*}
When $L$ is large enough, we will show that $h_L$ admits a second negative eigenvalue $\mu_L^- > \mu_L^+$ and we denote the corresponding eigenstate by $u_L^-$. We can show that $u_L^-$ is invariant under rotations along the first axis. Up to a change of phase, we can also choose it positive on the half space $\{x_1>0\}$ and negative on $\{x_1<0\}$.

In the limit where $L\to \infty$, $u_L^\pm$ will split into two bubbles, each of them minimizing the one nucleus Hartree model to leading order. Results of this kind have been shown in the large literature about double well type potentials, see for instance \cite{harrell1980double} in the linear case or \cite{olgiati2020hartree} for a nonlinear model. Our goal will be to prove this in the Coulomb case and determine the behavior of the Lagrange multipliers $\mu_L^\pm$ more precisely.

\subsection{Monoatomic Hartree model}
Now, we detail the monoatomic Hartree model which will constitute our elementary brick to describe the behavior of $u_L^\pm$ and $\mu_L^\pm$ when $L$ goes to $\infty$.
We introduce the energy functional
\begin{gather*}
\mathcal{E}(v) \coloneqq \int_{\R^d} \abs{\nabla v(x)}^2\diff x - \int_{\R^d} \frac{\abs{v(x)}^2}{\abs{x}}\diff x +  \frac{1}{2} \iint_{\R^d \times \R^d} \frac{\abs{v(x)}^2 \abs{v(y)}^2}{\abs{x-y}} \diff x \diff y \, ,
\end{gather*} 
defined for all $v\in H^1(\R^d)$. The associated minimization problem reads
\begin{gather}
\label{eq: min_problem_ow}
\boxed{I \coloneqq \inf \enstq{\mathcal{E}(v)}{v \in H^1(\R^d)\et \int_{\R^d} \abs{v}^2 = 1 } \pt}
\end{gather}
This minimization problem has been extensively studied in the litterature, at least in dimension~3 (see for instance \cite{benguria1981thomasfermi, lieb1981thomasfermi}). In particular, all the minimizing sequences for \eqref{eq: min_problem_ow} are precompact in $H^1(\R^d)$ and \eqref{eq: min_problem_ow} admits a unique minimizer $u$, up to a phase, which is the unique positive ground state of the self-adjoint operator
\begin{align}
\label{eq: mean-field_hamiltonian_ow}
h \coloneqq -\Delta +V + \abs{u}^2\ast|\cdot|^{-1} \, ,
\end{align}
defined on the domain 
\begin{align}
\label{eq: domain_mono}
\mathcal{D}(h) = \enstq{v \in H^1(\R^2)}{(-\Delta + V)v \in L^2(\R^2)} \, ,
\end{align}
if $d=2$ and on $H^2(\R^3)$ if $d=3$. The essential spectrum of $h$ is equal to $\intfo{0}{\infty}$ and we denote by $\mu <0$ its first nondegenerate eigenvalue. Because the energy functional $\mathcal{E}$ is invariant under rotations, $u$ is radial.

\subsection{Main result}
Now, we can state our main result which provides leading order when $L\to\infty$ of the ground state $u_L^+$ and the first excited state $u_L^-$ of the diatomic model \eqref{eq: restricted_hartree}. Those can be expressed as an even or odd superposition of right and left translations of the function $u$, respectively. Moreover, we give upper and lower exponential bounds on the spectral gap $\mu_L^- - \mu_L^+$ which is our main contribution in this article.
\begin{theo}
	\label{theo: main_theorem}
	Let $\alpha = 0$ when $d=3$ and $\alpha >0$ when $d=2$. Let $\epsilon >0$. Then we have
	\begin{gather}
	\label{eq: convergence_lagrange_multipliers}
	\abs{\mu_L^\pm - \mu} =
	\begin{cases}
	\grandO{L^{-3+\epsilon}} & \text{if }d=2 \, , \\
	\grandO{L^{-\infty}} & \text{if }d=3 \, ,
	\end{cases} \\
	\label{eq: convergence_ground_and_first_excited_states}
	\norme{u_L^\pm - (u(\cdot - \mathbf{x}_L) \pm u(\cdot + \mathbf{x}_L))}_{H^{2-\alpha}(\R^d)} =
	\begin{cases}
	\grandO{L^{-3+\epsilon}} & \text{if }d=2 \, , \\
	\grandO{L^{-\infty}} & \text{if }d=3 \pt
	\end{cases}
	\end{gather}
	When $d=2$, the constants appearing in the $O$ depend on $\alpha$ and $\epsilon$. The energy $E_L$ satisfies
	\begin{align}
	\label{eq: energy_difference}
	E_L =  I + 
	\begin{cases}
	\left(\frac{3m_1}{4}\right)^2 \frac{1}{L^5} + \petito{\frac{1}{L^5}} & \text{if }d=2 \, ,\\
	\grandO{L^{-\infty}} & \text{if }d=3 \, ,
	\end{cases}
	\end{align}
	where the one-electron energy $I$ is defined in \eqref{eq: min_problem_ow} and $m_1 \coloneqq \int_{\R^2} \abs{u(x)}^2 \abs{x}^2 \diff x$ is the second moment of $\abs{u}^2$.
	Moreover, there exists $C>0$ such that the following lower and upper exponential bounds on the spectral gap $\mu_L^- - \mu_L^+$ holds
	\begin{align}
	\label{eq: spectral_gap_estimation}
	\boxed{\frac{1}{C} \frac{e^{-\sqrt{\abs{\mu}}L}}{L^d} \leq \mu_L^- - \mu_L^+ \leq C e^{-\sqrt{\abs{\mu}}L}\pt}
	\end{align}
\end{theo}
In this statement, we have used the notation $\grandO{L^{-\infty}}$ to denote a $\grandO{L^{-k}}$ for all $k$, where the $O$ may depend on $k$.

It is well-known that the eigenvalues of a linear Schrödinger operator in a double well tend to group in pairs as the distance between the wells increases \cite{harrell1980double}.
The spectral gap between pairs of eigenvalues comes from the tunneling effect between the two wells and its estimation amounts to finding the asymptotics of the corresponding eigenfunctions at infinity, see \cite{harrell1980double, simon1984semiclassicalII}. We refer for instance to \cite{daumer1996schrodinger} where Daumer has precisely quantified the tunneling effect for linear Schrödinger operators with $-\Delta$-compact potentials vanishing at infinity. In dimension 3, this class of potentials contains in particular Coulomb potentials which are physically relevant when one wants to study molecules at dissociation. When the potential is homogeneous of degree $-1$, the regime where the wells get far apart is equivalent to a semi-classical limit, which has been intensively studied in the 80's. We refer to the series of papers from Helffer and Sjöstrand \cite{helffer1984multipleI,helffer1985puitsII,helffer1985multipleIII,helffer1985puitsIV,helffer1986puitsV,helffer1987puitsVI} and their collaborators \cite{outassourt1984, mohamed1991estimations, daumer1994hartree, daumer1996schrodinger}.

To our knowledge, the literature about the tunneling effect emerging from nonlinear models at dissociation seems rather scarce. In \cite{daumer1994hartree}, using fixed point arguments, Daumer constructed solutions to the Hartree-Fock equations in multiple wells from the ground state of reference monoatomic operators. Nevertheless, the assumptions on the well and interaction potentials are rather restrictive and do not cover Coulomb systems. In \cite{conlon1983semiclassical}, Conlon studied the Hartree-Fock model where the atomic potentials are assumed to be smooth in the vicinity of the nuclei and with Coulombic behavior at infinity. In the semi-classical regime, the author shows that the exchange energy for the ground state of this system converges toward the Dirac exchange energy of the Thomas-Fermi one-body density. More recently, in a series of papers \cite{rougerie2018interacting,olgiati2020hartree,olgiati2021bosons}, Olgiati, Rougerie and Spehner consider bosonic systems trapped in a symmetric double-well potential in the limit where the distance between the wells increases to infinity and the potential barrier is high. These last works are probably the closest to our and we will use several arguments from \cite{olgiati2020hartree} in this paper.

In dimension $d=3$, our model has already been considered by Catto and Lions in \cite{catto1993bindingIII}, but in a slightly different context. More specifically, the authors precisely compute the leading order of the energy difference between the diatomic model at dissociation and the two non-interacting monoatomics models. They also give asymptotic estimates on the corresponding ground states but did not considered the spectral gap. The strategy of proof for Theorem \ref{theo: main_theorem} does not depend on the dimension. However, thanks to the estimates from \cite{catto1993bindingIII}, many difficulties arising in the two-dimensional scenario are avoided when $d=3$. In particular, when $d=3$, all the estimates are exponentially small when $L \to \infty$ which is a manifestation of Newton's theorem (recall the system is neutral), that is, of the fact that the potential $V(x) = -\abs{x}^{-1}$ is, up to a constant factor, the Green function of the Laplacian in $\R^3$. That is why we will mainly focus on the more difficult $d=2$ case and only outline the argument when $d=3$, giving references whenever it is relevant. When $d=2$, the multipolar expansion of the potential generated by a radial charge distribution localized in space admits a non zero quadrupole moment (see Lemma \ref{lemma: potential_expansion} in Appendix \ref{sec: expansion_formula}). As a consequence, we obtain a polynomial convergence rate for the Lagrange multipliers, the ground state, the first excited state and the energy (see \eqref{eq: convergence_lagrange_multipliers}, \eqref{eq: convergence_ground_and_first_excited_states} and \eqref{eq: energy_difference}). The main contribution of this article is the estimation \eqref{eq: spectral_gap_estimation} of the tunneling effect when $d=2$ in spite of the absence of Newton's theorem. To this end, we use at our advantage the fact that tunneling depends to leading order only on the behavior around the origin of the corresponding eigenfunctions (see Lemma \ref{lemma: ground_state_resolution}). We are able to determine the exact exponential decay, but not the polynomial factor. We think one could get rid of the $L^{-d}$ in the left side of \eqref{eq: spectral_gap_estimation}.

\subsection{Strategy of proof}
We sketch the strategy for the proof of Theorem \ref{theo: main_theorem}. 
It is convenient to introduce some notation relative to the translated version of the monoatomic model defined in~\eqref{eq: min_problem_ow}.
We denote by
\begin{gather}
\label{eq: definition_translated_potential}
V^r_L \coloneqq V(\cdot - \mathbf{x}_L) \, ,\quad V^\ell_L \coloneqq V(\cdot + \mathbf{x}_L)\, ,\quad u^r_L \coloneqq u(\cdot - \mathbf{x}_L) \et u^\ell_L \coloneqq u(\cdot + \mathbf{x}_L) \, ,
\end{gather}
the left and right potentials and monoatomic eigenstates. For $\kappa \in \{\ell,r\}$, $u^\kappa_L$ satisfies the Euler-Lagrange equation
\begin{align}
\label{eq: euler_lagrange_equation_one_nucleus}
\boxed{h^\kappa_L u^\kappa_L \coloneqq \left(-\Delta +V^\kappa_L + \abs{u^\kappa_L}^2\ast|\cdot|^{-1}\right) u^\kappa_L = \mu\, u^\kappa_L \pt}
\end{align}

First we state regularity properties (see Proposition \ref{prop: regularity}, Proposition \ref{prop: study of d(h)} and Remark \ref{rem: continuity_around_singularities}) about the eigenfunctions associated with the two mean-field hamiltonians we are interested in. Then, we use the maximum principle (Lemma \ref{lemma: comparison_supersolution} and Lemma \ref{lemma: comparison_hoffmann_ostenhoff}) to establish exponential pointwise and integral bounds on the minimizer $u$ of the monoatomic model (Proposition \ref{prop: u_exponential_falloff}). They allow us to estimate interaction terms (Lemma \ref{lemma: tunneling} and Lemma \ref{lemma: interaction_left_right_d=3}) which will appear later on.

Then we show a first convergence result in $H^1(\R^d)-$norm (Proposition \ref{prop: strong_convergence_H1}) by inserting the trial function
\[
u_\mathrm{trial} = \frac{ \sqrt{2}(u^r_L +  u^\ell_L)}{\norme{ u^r_L +  u^\ell_L}_{L^2(\R^d)}} \, ,
\]
into the energy functional $\mathcal{E}_L$ of the diatomic model. The argument uses the precompactness of the minimizing sequence for the monoatomic model. We get the upper bound $E_L \leq \mathcal{E}(u) + \petito{1}$ from the expansion of the mean-field hamiltonian $h$ (Lemma \ref{lemma: potential_expansion}). To get the lower bound we localize in the vicinity of the two nuclei. This allows us to show convergence of the Lagrange multiplier (Proposition~\ref{lemma: convergence energy})
\[
\abs{\mu_L^+ - \mu} \limit{L\to\infty} 0 \, ,
\]
and convergence in $H^{2-\alpha}(\R^d)-$norm of the eigenfunctions (Proposition \ref{prop: stronger_convergence}). This finally provides pointwise and integral exponential bounds on $u_L^+$ (Proposition \ref{cor: exponential decay u^+_n}).

After that, we will aim at estimating the rates of convergence for both $H^1(\R^d)-$norm errors and for the Lagrange multipliers (Proposition~\ref{prop: rate_convergence} and Proposition \ref{prop: rate_convergence_ground_state_energy}). For this purpose, we state a stability result (Proposition \ref{prop: stability}) and we make use of bootstrap-type arguments. We also get similar results for the first excited state $u_L^-$ and its associated Lagrange multiplier $\mu_L^-$ (Proposition~\ref{prop: rate_convergence_bis}). We also show a uniform lower bound on the gap between $\mu_L^\pm$ and the remaining spectrum of $h_L$ (Proposition \ref{prop: the_remaining_spectrum_is_far_away}).

Then, we give convergence rates in higher Sobolev space (Proposition~\ref{prop: convergence_rate_higher_sobolev_spaces}). In Proposition \ref{prop: sharper_exponential_bounds}, we give exact pointwise exponential bounds on $u_L^\pm$ at finite but large distance from the nuclei. This allows us to study the spectral gap $\mu_L^- - \mu_L^+$ (Theorem \ref{th: spectral_gap_estimation}). By the ground state substitution formula (Lemma \ref{lemma: ground_state_resolution}), this spectral gap depends only on the behavior of $u_L^\pm$ in the vicinity of the nuclei which is covered by our exponential bound.

\subsection{Notations and conventions}

We will denote by
\begin{align}
\label{eq: coulomb_energy}
D(\rho,\mu) \coloneqq \frac{1}{2} \iint_{\R^d \times \R^d} \frac{\rho(x) \mu(y)}{\abs{x-y}} \diff x \diff y \, ,
\end{align}
the $d$-dimensional \emph{Coulomb interaction energy }between two charge densities $\rho$ and $\mu$ whenever it makes sense. This is a non-negative bilinear form \cite[Theorem 9.8]{lieb2001analysis} and, by the Hardy-Littlewood-Sobolev inequality \cite[Theorem 4.3]{lieb2001analysis}, it is continuous on $L^{\frac{2d}{2d-1}}(\R^d) \times L^{\frac{2d}{2d-1}}(\R^d)$.

For $R>0$ and $x \in \R^d$, we denote by $B(x,R)$ the open ball of $\R^d$ centered in $x$ with radius $R$.

For any $\mathbf{x} \in \R^d$, we introduce the translation operator $\tau_\mathbf{x}$ defined for all $v \in L^2(\R^d)$ by
\[
\forall x\in\R^d,~(\tau_\mathbf{x}v)(x) = v(x-\mathbf{x}) \, ,
\]
which is unitary as an operator acting on $L^2(\R^d)$.

In the estimates, the constants can change from line to line. Whenever we consider the exact determination of the constant irrelevant for our purpose, we will drop it, writing $A \lesssim B$ if $A \leq C \cdot B$ for some constant $C>0$ which is independent from the parameters.

\subsection{Organization of the paper}
In Section \ref{sec:properties-of-the-eigenfunctions}, we show some basic properties of the eigenfunctions of the mean-field hamiltonians $h$ and $h_L$. In Section \ref{sec:construction-of-quasi-modes}, we construct quasi-modes for $h_L$ from even or odd superposition of translated versions of $u$. We also give convergence rates and study the spectral gap $\mu_L^- - \mu_L^+$.

\subsection*{Acknowledgments}
The author would like to thank his PhD advisor M. Lewin for valuable discussions. This project has received funding from the European Research Council (ERC) under the European Union’s Horizon 2020 research and innovation programme (grant agreement MDFT No 725528 of M. Lewin).

\section{Properties of eigenfunctions}\label{sec:properties-of-the-eigenfunctions}

\subsection{Regularity away from singularities}

In this section, we study the regularity of the eigenfunctions of the mean-field hamiltonians $h$ and $h_L$. Our first result is that these functions are smooth away from the nuclei. In this section, we assume the existence of the first excited state $u_L^-$ of $h_L$ for $L$ large enough. This will be shown later on (see Proposition \ref{prop: rate_convergence_ground_state_energy}).
\begin{prop}[Regularity of eigenfunctions away from the singularities]
	\label{prop: regularity}
	Let $d\in \{2,3\}$. The eigenfunction $u$ belongs to $\mathcal{C}^\infty(\R^d\setminus \{0\})$. The eigenfunctions $u_L^+$ and $u_L^-$ belong to $\mathcal{C}^\infty(\R^d\setminus \{-\mathbf{x}_L,+\mathbf{x}_L\})$.
\end{prop}

In dimension $d=3$, this result is fairly standard and it is obtained by bootstrap-type arguments like in \cite[Lemma 3.2]{olgiati2020hartree}, for instance. In dimension $d=2$, the analysis is more difficult due to the combination of a singular potential $V \notin L^2_\mathrm{loc}(\R^2)$ and the non-local equation
\[
\sqrt{-\Delta}\, V = 2\pi \delta_0 \, ,
\]
satisfied by $V$.
We start with a technical lemma.
\begin{lemma}
	\label{lemma: technical2}
	Let $d \geq 2$.
	For all $u,v \in H^1(\R^d)$ and all $r \in \intof{d}{\infty}$, we have
	\begin{gather}
	\label{lemma: technical2_inequality1}
	\norme{(uv)*|\cdot|^{-1}}_{L^r(\R^d)} \lesssim \norme{u}_{H^1(\R^d)} \norme{v}_{H^1(\R^d)} \pt
	\end{gather}
\end{lemma}
\begin{proof}
	Let $u,v \in H^1(\R^d)$. By Sobolev inequalities, we have the continuous embedding $H^1(\R^d) \subset L^q(\R^d)$ for all $q \in \intfo{2}{\infty}$ if $d=2$ and for all $q \in \intff{2}{\frac{2d}{d-2}}$ if $d\geq 3$. Let $r \in \intoo{d}{\infty}$ and denote $p = \frac{dr}{d(r+1)-r} \in \intoo{1}{\frac{d}{d-2}}$. Then, by the Hardy-Littlewood-Sobolev inequality and the Cauchy-Schwarz inequality , we have
	\[
	\norme{(uv) * |\cdot|^{-1}}_{L^r(\R^d)} \lesssim \norme{uv}_{L^p(\R^d)} \lesssim \norme{u}_{L^{2p}(\R^d)}\norme{v}_{L^{2p}(\R^d)} \lesssim \norme{u}_{H^1(\R^d)} \norme{v}_{H^1(\R^d)} \pt
	\]
	This shows inequality \eqref{lemma: technical2_inequality1} except for the case $r = \infty$ which cannot be treated this way. To bypass this issue, we write
	\[
	|\cdot|^{-1} = |\cdot|^{-1} \mathds{1}_{B_1} + |\cdot|^{-1} \mathds{1}_{B_1^c}\in L^q(\R^d) + L^\infty(\R^d) \, ,
	\]
	where $q \in \intfo{1}{d}$ and $B_1$ denotes the unit ball of $\R^d$. We choose $q$ such that $q' \coloneqq \frac{q}{q-1} \in \intoo{1}{\frac{d}{d-2}}$. Then, we have, using Young's inequality and the Cauchy-Schwarz inequality
	\begin{gather*}
	\norme{(uv)*(|\cdot|^{-1} \mathds{1}_{B_1}(x))}_{L^\infty(\R^d)} \lesssim  \norme{uv}_{L^{q'}(\R^d)} \lesssim \norme{u}_{L^{2q'}(\R^d)} \norme{v}_{L^{2q'}(\R^d)}  \lesssim \norme{u}_{H^1(\R^d)} \norme{v}_{H^1(\R^d)} \, , \\
	\norme{(uv)*(|\cdot|^{-1} \mathds{1}_{B_1^c}(x))}_{L^\infty(\R^d)} \lesssim \norme{uv}_{L^1(\R^d)} \lesssim \norme{u}_{L^2(\R^d)} \norme{v}_{L^2(\R^d)} \lesssim \norme{u}_{H^1(\R^d)} \norme{v}_{H^1(\R^d)}\pt
	\end{gather*}
	This shows inequality \eqref{lemma: technical2_inequality1} in the case $r=\infty$.
\end{proof}
\begin{rem}
	We cannot have $(uv)*|\cdot|^{-1} \in L^d(\R^d)$ for all $u,v \in H^1(\R^d)$. Indeed, the first order of the large $\abs{x}$ behavior of the convolution $\abs{u}^2 \ast |\cdot|^{-1}$ is given by $\norme{u}^2_{L^2(\R^d)} \abs{x}^{-1}$ which is not in $L^d(\R^d)$ whenever $u \neq 0$.
\end{rem}
\begin{lemma}[Local fractional elliptic regularity]
	\label{prop: fractional regularity}
	Let $u \in H^r_\mathrm{loc}(\R^2 \setminus \{0\}) \cap L^2(\R^2)$ for some integer $r \geq 1$. Then $\abs{u}^2 \ast |\cdot|^{-1} \in H^{r-1}_\mathrm{loc}(\R^2 \setminus \{0\})$.
\end{lemma}
\begin{proof}
	Let $\Omega$ be an open and relatively compact set of $\R^2\setminus \{0\}$. Let $\alpha \in \N^2$ such that $\abs{\alpha} \leq r-1$. We will show that $\partial^\alpha(\abs{u}^2 \ast |\cdot|^{-1}) \in L^2(\Omega)$. Let $\delta >0$ such that $0 \notin \Omega + B(0,\delta)$.
	Let $\chi \in \mathcal{C}_c^\infty(\R^2)$ be such that
	\[
	0\leq \chi \leq 1,\quad \chi\equiv 1 \text{ on } \Omega + B(0,\delta/2) \et \chi \equiv 0  \text{ on } \Omega + B(0,\delta) \pt
	\]
	Let $\eta$ be such that $\chi^2 + \eta^2 = 1$. We have $\partial^\alpha (\abs{\chi u}^2 \ast |\cdot|^{-1}) = (\partial^\alpha \abs{\chi u}^2) \ast |\cdot|^{-1}$ in $\mathcal{D}'(\R^2)$. Moreover, because $u \in H^r_\mathrm{loc}(\R^2 \setminus \{0\})$, we see that $\partial^\alpha \abs{\chi u}^2$ is a finite sum of product of two $H^1(\R^2)$ functions. By Lemma \ref{lemma: technical2}, this implies that $\partial^\alpha (\abs{\chi u}^2 \ast |\cdot|^{-1}) \in L^\infty(\R^2) \subset L^2(\Omega)$. Because the function $\partial^\alpha \abs{x-y}^{-1}$ is bounded when $x \in \Omega$ and $y \in \supp \eta$, we have
	$\partial^\alpha (\abs{\eta u}^2 \ast |\cdot|^{-1}) = \abs{\eta u}^2 \ast \left(\partial^\alpha|\cdot|^{-1}\right)$ in $\mathcal{D}'(\Omega)$. Finally, because the support of $\eta$ is localized away from $\Omega$ and because $u \in L^2(\R^2)$, we have $\abs{\eta u}^2 \ast \left(\partial^\alpha|\cdot|^{-1}\right) \in L^\infty(\Omega) \subset L^2(\Omega)$.
\end{proof}
\begin{rem}
	The same result holds with the same proof when $\R^2 \setminus\{0\}$ is replaced by any open set of $\R^2$. For instance, we can take $\R^2 \setminus \{\pm \mathbf{x}_L\}$.
\end{rem}
\begin{proof}[Proof of Proposition \ref{prop: regularity}]
	We treat the two-dimensional case for $u$ and only mention the changes for $d=3$. For $u_L^\pm$, the proof is the same so we do not write it .
	Let $\Omega \subset \R^2 \setminus \{0\}$ be an open and relatively compact set of $\R^2 \setminus \{0\}$. We denote
	\[
	W(u) = \abs{u}^2 *|\cdot|^{-1} \et \tilde{V}(x) = - \abs{x}^{-1} - \mu \, ,
	\]
	which are locally essentially bounded in $\Omega$. This is obvious for $\tilde{V}$ and a consequence of Lemma \ref{lemma: technical2} for $W(u)$.
	We recall that $u$ is solution, in the sense of distribution, of the strictly elliptic partial differential equation
	\begin{align}
	\label{eq: lemma_regularity}
	-\Delta u + W(u) u + \tilde{V} u = 0 \pt
	\end{align}
	As a consequence, $u$ is also a weak solution in $H^1(\Omega)$ (in the sense of \cite[Chapter 8]{gilbarg2015elliptic}) of the same equation. By \cite[Theorem 8.8]{gilbarg2015elliptic}, this implies that $u \in H^2(\Omega)$. Thus, we have $u \in H^2_\mathrm{loc}(\R^2 \setminus \{0\})$.
	
	To complete the proof, we use a bootstrap type argument.
	Assume that $u \in H^r_\mathrm{loc}(\R^2 \setminus \{0\})$ for some integer $r\geq 2$. By Lemma \ref{prop: fractional regularity}, this implies that $W(u) \in H^{r-1}_\mathrm{loc}(\R^2 \setminus \{0\})$.
	By Theorem~1 of \cite[Section 4.3]{runst1996sobolev}, for all open and bounded set $\Omega \subset  \R^2 \setminus \{0\}$ with smooth boundary, for all $(f,g) \in H^{s_1}(\Omega) \times H^{s_2}(\Omega)$, we have
	\begin{align}
	\label{eq: lemma_regularity_ineq}
	\norme{fg}_{H^s(\Omega)} \lesssim \norme{f}_{H^{s_1}(\Omega)} \norme{g}_{H^{s_2}(\Omega)} \, ,
	\end{align}
	provided that $s_i \geq s \geq 0$ for all $i \in\{1,2\}$ and $s_1+s_2 > s + 1$.
	In particular, $u \in H^r_\mathrm{loc}(\R^2 \setminus \{0\})$ and $W(u) \in H^{r-1}_\mathrm{loc}(\R^2 \setminus \{0\})$ imply $W(u)u \in H^{r-1}_\mathrm{loc}(\R^2 \setminus \{0\})$. Let $\chi \in \mathcal{C}_c^\infty(\R^2)$ be a localization function such that
	\[
	\supp \chi \subset \R^2\setminus\{0\} \et \chi \equiv 1 \quad \text{on} \quad K \, ,
	\]
	where $K \subset \R^2\setminus\{0\}$ is a compact set. Multiplying \eqref{eq: lemma_regularity} by $\chi$, we find
	\[
	-\Delta (\chi u) =  -(\Delta \chi + \chi \tilde{V})u - \chi W(u)u -2 \nabla \chi \cdot \nabla u \pt
	\]
	The first term on the right hand side belongs to $H^r(\R^2)$ and the last two terms are in $H^{r-1}(\R^2)$. This shows that $\chi u \in H^{r+1}(\R^2)$. As a consequence $u \in H^{r+1}_\mathrm{loc}(\R^2\setminus \{0\})$. By induction, $u \in H^r_\mathrm{loc}(\R^2\setminus\{0\})$ for all $r \geq 0$ and, by the Sobolev embeddings, this shows that $u \in \mathcal{C}^\infty(\R^2\setminus\{0\})$.
	
	The easier case $d=3$ is treated similarly, replacing Lemma \ref{prop: fractional regularity} by the usual elliptic regularity for solutions of the Laplace equation $-\Delta u = f$ in $\R^3$.
\end{proof}

\subsection{Regularity around singularities}

In the sequel, $\mathcal{C}^{\ell,\theta}_0(\R^d)$ denotes the space of $\mathcal{C}^\ell(\R^d)$ functions which vanish at infinity as well as their first $\ell$ derivatives and such that the derivatives of order $\ell$ are Hölder continuous with exponent $\theta$.
In dimension 3, Sobolev embeddings give that $ u \in H^2(\R^3) \subset \mathcal{C}_0^{0,1/2}(\R^3)$. Hence the eigenfunctions $u_L^\pm$ and $u$ are continuous (even Hölder) in the vicinity of the singularities.
In dimension 2, because $|\cdot|^{-1} \notin L^2_\mathrm{loc}(\R^2)$, the mean-field hamiltonians $h$ and $h_L$ are not self-adjoint on $H^2(\R^2)$ but on the domains $\mathcal{D}(h)$ or $\mathcal{D}(h_L)$ (see \eqref{eq: domain_mono} and \eqref{eq: domain_1}). In the next proposition, we study the structure of $\mathcal{D}(h)$.
\begin{prop}[Regularity around singularities]
	\label{prop: study of d(h)}
	Let $\alpha \in \intoo{0}{2}$ if $d=2$ and $\alpha=0$ if $d=3$. Let $\epsilon>0$. Then there exists constants $C(\epsilon)$ and $C(\epsilon,\alpha )$ such that for all $v \in \mathcal{D}(h)$ we have
	\begin{align*}
	\norme{v}_{H^{2-\alpha}(\R^d)} \leq
	\begin{cases}
	\epsilon \norme{h v}_{L^2(\R^d)} +C(\epsilon,\alpha) \norme{v}_{L^2(\R^d)} & \text{if }d=2 \, ,\\
	(1+\epsilon)\norme{hv}_{L^2(\R^d)} + C(\epsilon) \norme{v}_{L^2(\R^d)}  & \text{if }d=3 \pt
	\end{cases}	 
	\end{align*}
	In particular, we have $\mathcal{D}(h)\subset H^{2-\alpha}(\R^d)$ with continuous embedding.
\end{prop}
\begin{rem}\label{rem: continuity_around_singularities}
	\begin{enumerate}[noitemsep, label=(\roman*)]
		\item The conclusions of Proposition \ref{prop: study of d(h)} also hold for $h_L$ but, in that case, the proof needs uniform bounds on $\norme{\rvert u_L^+ \lvert^2 \ast |\cdot|^{-1}}_{L^\infty(\R^d)}$ (see Remark \ref{rem: regularity_h_L}). 
		\item Sobolev inequalities give the continuous inclusion $H^{2-\alpha}(\R^2) \subset \mathcal{C}^{0,1-\alpha}_0(\R^2)$ for $\alpha <1$. Hence the eigenfunctions of the mean-field hamiltonians are Hölder with parameter $1-\alpha$
		for any $\alpha \in \intoo{0}{1}$ in the vicinity of the nuclei.
		\item Because $|\cdot|^{-1} \notin L^2_\mathrm{loc}(\R^2)$, we cannot have $\mathcal{D}(h) \subset H^2(\R^2)$. Otherwise, for all $v \in \mathcal{D}(h)$, we would have $v/|\cdot| \in L^2(\R^2)$. Similar remarks hold for $h^\ell_L$, $h^r_L$ and $h_L$.
	\end{enumerate}	
\end{rem}
Now, we write the proof of Proposition \ref{prop: study of d(h)}. We start with a lemma.
\begin{lemma}
	\label{lemma: fractional_hardy_bis}
	Let $\alpha \in \intoo{0}{\frac{1}{2}}$. The operator $(-\Delta)^{-\alpha} \abs{x}^{-1} (-\Delta)^{-\frac{1}{2}+\alpha}$ acting on $L^2(\R^2)$ is bounded:
	\begin{align}
	\label{eq: fractional_hardy_bis}
	\norme{(-\Delta)^{-\alpha} \abs{x}^{-1} (-\Delta)^{-\frac{1}{2}+\alpha}} \leq C_\alpha \, ,
	\end{align}
	for some constant $C_\alpha >0$.
\end{lemma}
\begin{proof}
	By \cite[Theorem 6.2.2]{simon2015comprehensiveIII} (Stein-Weiss inequality) and the remark which follows, the operators $(-\Delta)^{-\beta/2} \abs{x}^{-\beta}$ and $\abs{x}^{-\beta} (-\Delta)^{-\beta/2}$ acting on $L^2(\R^2)$ are bounded for all $\beta \in \intoo{0}{1}$. It remains to write
	\[
	(-\Delta)^{-\alpha} \abs{x}^{-1} (-\Delta)^{-\frac{1}{2}+\alpha} = (-\Delta)^{-\alpha} \abs{x}^{-2\alpha} \abs{x}^{-1 +2 \alpha} (-\Delta)^{-\frac{1}{2}+\alpha} \, ,
	\]
	to conclude the proof of Lemma \ref{lemma: fractional_hardy_bis}.
\end{proof}
\begin{proof}[Proof of Proposition \ref{prop: study of d(h)}]
	When $d=3$, Proposition \ref{prop: study of d(h)} is a consequence of the fact that for all $\epsilon>0$ there exists $C_\epsilon>0$ such that for all $v \in H^2(\R^3)$ we have (see for instance \cite[Theorem X.19]{reed1975methodsII})
	\[
	\norme{\frac{v}{|\cdot|}}_{L^2(\R^3)} \leq \epsilon \norme{v}_{H^2(\R^3)}  + C_\epsilon\norme{v}_{L^2(\R^3)} \pt
	\]
	
	Now, we turn our attention to the two-dimensional case. We write the proof for $h$ in detail and we point out the changes necessary to handle $h_L$. First, we express the inverse of $-\Delta -\abs{x}^{-1}$ in terms of Neumann series when $\nu >0$ large enough. We write
	\begin{align*}
	(-\Delta -\abs{x}^{-1} +\nu) = (-\Delta+\nu)^{1/2} \left(1 - (-\Delta+\nu)^{-1/2} \abs{x}^{-1} (-\Delta+\nu)^{-1/2}\right) (-\Delta+\nu)^{1/2} \pt
	\end{align*}
	Using Lemma \ref{lemma: fractional_hardy_bis} with the estimates $\norme{\left(\frac{-\Delta}{-\Delta+\nu}\right)^{1/4}} \leq 1$ and $\norme{\left(\frac{1}{-\Delta+\nu}\right)^{1/4}} \leq \nu^{-1/4}$ we obtain
	\begin{align*}
	&\norme{(-\Delta+\nu)^{-1/2} \abs{x}^{-1} (-\Delta+\nu)^{-1/2}} \\
	&= \norme{(-\Delta+\nu)^{-1/4} \left(\frac{-\Delta}{-\Delta+\nu}\right)^{1/4} (-\Delta)^{-1/4} \abs{x}^{-1} (-\Delta)^{-1/4} \left(\frac{-\Delta}{-\Delta+\nu}\right)^{1/4}  (-\Delta+\nu)^{-1/4}} \\
	&\leq C_{1/4}\nu^{-1/2} \pt
	\end{align*}
	If we choose $\nu > C_{1/4}^2$ then the operator $1- (-\Delta+\nu)^{-1/2} \abs{x}^{-1} (-\Delta+\nu)^{-1/2}$ is invertible and its inverse is given by a Neumann series. Thus we have
	\begin{align*}
	(-\Delta - \abs{x}^{-1} +\nu)^{-1} = (-\Delta+\nu)^{-1/2} \sum_{n \geq 0} \left[(-\Delta+\nu)^{-1/2} \abs{x}^{-1} (-\Delta+\nu)^{-1/2}\right]^n (-\Delta+\nu)^{-1/2} \pt
	\end{align*}
	Let $\alpha \in \intoo{0}{\frac{1}{2}}$. We have
	\begin{multline}
	\label{eq: that}
	(-\Delta+1)^{1-\alpha}(-\Delta - \abs{x}^{-1} +\nu)^{-1} \\
	= \left(\frac{-\Delta+1}{-\Delta+\nu} \right)^{1-\alpha}  (-\Delta+\nu)^{1/2 - \alpha} \sum_{n \geq 0} \left[(-\Delta+\nu)^{-1/2} \abs{x}^{-1} (-\Delta+\nu)^{-1/2}\right]^n (-\Delta+\nu)^{-1/2} \pt
	\end{multline}
	By functional calculus, we have $\norme{\left((-\Delta+1)(-\Delta+\nu)^{-1} \right)^{1-\alpha}} \leq 1$ if we take $\nu > \max(1,C_{1/4}^2)$.
	The term of order $n$ in \eqref{eq: that} is estimated by writing
	\begin{multline*}
	(-\Delta+\nu)^{1/2 - \alpha} \left[(-\Delta+\nu)^{-1/2} \abs{x}^{-1} (-\Delta+\nu)^{-1/2}\right]^n (-\Delta+\nu)^{-1/2}\\
	= \left[(-\Delta+\nu)^{- \alpha}\abs{x}^{-1} (-\Delta+\nu)^{-1/2 + \alpha} (-\Delta+\nu)^{-1/2} \right]^n (-\Delta+\nu)^{- \alpha} \pt
	\end{multline*}
	Then, using $n$ times Lemma \ref{lemma: fractional_hardy_bis} and the estimate $\norme{(-\Delta+\nu)^{-\beta}} \leq \nu^{-\beta}$ for any $\beta>0$, we have
	\begin{align*}
	\norme{\left[(-\Delta+\nu)^{- \alpha}\abs{x}^{-1} (-\Delta+\nu)^{-1/2 + \alpha} (-\Delta+\nu)^{-1/2} \right]^n (-\Delta+\nu)^{-1/2 - \alpha}} \leq \nu^{-\alpha} \left(\frac{C_\alpha}{\sqrt{\nu}}\right)^n \, ,
	\end{align*}
	where $C_\alpha$ is the constant appearing in Lemma \ref{lemma: fractional_hardy_bis}.
	Inserting this into \eqref{eq: that}, we obtain
	\begin{align*}
	\norme{(-\Delta+1)^{1-\alpha}(-\Delta - \abs{x}^{-1} +\nu)^{-1}} \leq 2\nu^{-\alpha}  \, ,
	\end{align*}
	for $\nu >  \max(1,C^2_{1/4},4C^2_\alpha)$. If in addition we assume $\nu \geq (2/\epsilon)^{1/\alpha}$ then we have
	\[
	\norme{v}_{H^{2-2\alpha}(\R^2)} \leq \epsilon \norme{hv}_{L^2(\R^2)} + \left(\nu + \norme{\abs{u}^2 \ast |\cdot|^{-1}}_{L^\infty(\R^2)} \right)\norme{v}_{L^2(\R^2)} \, ,
	\]
	for all $v \in \mathcal{D}(h)$.	We conclude the proof by noticing that $\abs{u}^2 \ast |\cdot|^{-1} \in L^\infty(\R^2)$ by Lemma~\ref{lemma: technical2}.
\end{proof}

\subsection{Exponential decay of \texorpdfstring{$u$}{u}}

In this section, we give the precise long-range behavior of the minimizer $u$ of the monoatomic problem \eqref{eq: min_problem_ow}. In particular, $u$ decays exponentially fast at infinity. We will treat both the two and three-dimensional cases. However, in dimension 3, more precise results have already been obtained by Catto and Lions in \cite[Appendix 2 \& 3]{catto1993bindingIII}. The main result of this section is the following proposition.
\begin{prop}[Exponential bounds on $u$]
	\label{prop: u_exponential_falloff}
	Let $d\in\{2,3\}$.
	There exists a constant $C>0$ such that we have for all $x\in \R^d$ the pointwise exponential estimates
	\begin{gather}
	\label{eq: u_exponential_falloff}
	\frac{1}{C} \frac{e^{-\abs{\mu}^{\frac{1}{2}} \abs{x}}}{1+ \abs{x}^{\frac{d-1}{2}}} \leq u(x) \leq C \frac{e^{-\abs{\mu}^{\frac{1}{2}} \abs{x}}}{1+ \abs{x}^{\frac{d-1}{2}}} \quad \text{and} \quad 
	\abs{\nabla u(x)} \leq C \frac{e^{-\abs{\mu}^{\frac{1}{2}} \abs{x}}}{1+ \abs{x}^{\frac{d-1}{2}}} \, ,
	\end{gather}	
	and such that for all $R>0$ we have the integral exponential estimate
	\begin{align}
	\label{eq: a priori exp decay2}
	\int_{\abs{x}>R} \abs{u(x)}^2 \diff x+ \int_{\abs{x}>R} \abs{\nabla u(x)}^2 \diff x \leq C e^{-2\abs{\mu}^{\frac{1}{2}} R} \pt
	\end{align}
	In addition, for all $\epsilon \in \intoo{0}{1}$ there exists a constant $C_\epsilon>0$ such that
	\begin{align}
	\label{eq: a priori exp decay3}
	\norme{u}_{H^{2-\alpha}(B(0,R)^c)} \leq C e^{-(1-\epsilon)\sqrt{\abs{\mu}R} } \pt
	\end{align}
\end{prop}
We could also give non optimal exponential bounds on $u_L^\pm$. 
However, as we do not yet have any information about the eigenvalues of the mean-field hamiltonian $h_L$, we prefer to postpone this analysis to another section.

The first step will be an easy and non optimal estimate on the decay of $u$ which will allow us to obtain some information about the behavior at infinity of the mean-field potential associated with the monoatomic model \eqref{eq: min_problem_ow}, denoted by
\begin{align}
\label{eq: def_mean_field_potential}
V^\mathrm{MF} \coloneqq - |\cdot|^{-1} + \abs{u}^2 \ast |\cdot|^{-1} \pt
\end{align}
Then, a finer analysis will give the optimal exponential decay rate.

We recall two fundamental comparison lemmas (\cite[Corollary 2.8]{agmon1985bounds} and \cite[Theorem 1.1]{hoffmann1980comparison}) which will be useful in the sequel. We will use the first one for the upper bounds in \eqref{eq: u_exponential_falloff} and the second one for the lower bound.
\begin{lemma}[First comparison lemma \cite{agmon1985bounds}]
	\label{lemma: comparison_supersolution}
	Let $d\in \N$ and $p> d/2$. Let $W \in L^p_\mathrm{loc}(\R^d)$ be real-valued and satisfying the condition: there exists $\theta \in \intoo{0}{1}$ and $C>0$ such that
	\[
	\int_{\R^d} W_- \abs{\varphi}^2 \leq \theta \int_{\R^d} \abs{\nabla \varphi}^2 + C \int_{\R^d} \abs{\varphi}^2\, ,\quad\forall \varphi \in \mathcal{C}^\infty_0(\R^d)\, ,
	\]
	where $W_- \geq 0$ denotes the negative part of $W$.
	Let $\psi$ be an eigenfunction of $P = -\Delta + W$ with eigenvalue $\lambda$. Let $\varphi$ be a positive and continuous supersolution of the equation $P - \lambda$ in the region $\{\abs{x} \geq R\}$. Then there exists a constant $C = C(R)$ such that
	\[
	\forall \abs{x} \geq R+1,~ \abs{\psi(x)} \leq C \varphi(x) \pt
	\]
\end{lemma}
\begin{lemma}[Second comparison lemma \cite{hoffmann1980comparison}]
	\label{lemma: comparison_hoffmann_ostenhoff}
	Let $d\in\N$ and $\Omega \subset \R^d$ be an open set. Let $W_1$ and $W_2$ two potentials in $L^1_\mathrm{loc}(\Omega)$ such that $W_2 \geq W_1 \geq 0$ in $\Omega$. Let $\psi$ and $\varphi$ such that :
	\begin{enumerate}[noitemsep, label=(\roman*)]
		\item $\psi,\varphi \in \mathcal{C}^0(\overline{\Omega})$, $\psi,\varphi \geq 0$ a.e in $\Omega$ and $\varphi,\psi \to 0$ as $\abs{x} \to \infty$ if $\Omega$ is unbounded;
		\item $\psi \geq \varphi$ on $\partial \Omega$;
		\item $W_1 \psi \geq \Delta  \psi$ and $W_2\varphi \leq \Delta \varphi$ in the sense of distribution.
	\end{enumerate}
	Then, we have $\psi \geq \varphi$ in $\Omega$.
\end{lemma}
We begin by proving the following \emph{a priori} exponential bound on $u$.
\begin{lemma}[\emph{A priori} pointwise exponential bound on $u$]
	\label{lemma: a priori exp decay}
	Let $d\in \{2,3\}$.	Let $u$ be the minimizer of~\eqref{eq: min_problem_ow}. Then, for all $\epsilon >0$ there exists a constant $C_\epsilon> 0$ such that the following pointwise estimate holds
	\begin{align}
	\label{eq: a priori exp decay}
	\forall x\in \R^d,~0 \leq u(x) \leq C_\epsilon e^{-\sqrt{\abs{\mu} - \epsilon} \abs{x}}  \pt
	\end{align}
\end{lemma}
\begin{proof}
	Recall that $h = -\Delta -|\cdot|^{-1} +\abs{u}^2 \ast |\cdot|^{-1}$ is the mean-field hamiltonian associated with the minimization problem \eqref{eq: min_problem_ow} and that $\mu <0$ is the first eigenvalue of $h$.
	
	Let $\gamma >0$ and $f(x) = \exp(-\gamma \abs{x})$. We shall prove that $f$ is a supersolution associated with the eigenvalue problem $(h - \mu)u = 0$ in the domain $\abs{x}>R$ for $R $ large enough and $\gamma$ well-chosen. We have : $\Delta f(x) = f(x)(\gamma^2 - \gamma (d-1) \abs{x}^{-1})$. Then
	\begin{align*}
	(h - \mu) f(x) 
	= f(x) \left(-\gamma^2 + (\gamma(d-1) - 1) \abs{x}^{-1} + \abs{u}^2 * \abs{\,\cdot\,}^{-1} + \abs{\mu}\right) 
	\geq f(x) \left(-\gamma^2 + \abs{\mu} - \abs{x}^{-1} \right) \pt
	\end{align*}
	If we choose $\gamma = \sqrt{\abs{\mu} - \epsilon}$ for some $0 < \epsilon < \abs{\mu}$ then $f$ is a supersolution for $h - \mu$ on the domain $\abs{x}>\epsilon^{-1}$. By the first comparison Lemma \ref{lemma: comparison_supersolution}, we conclude that for all $\abs{x} > \epsilon^{-1} + 1$, we have the pointwise estimate \eqref{eq: a priori exp decay}. We extend this estimate to $\R^d$ by using the continuity of $u$ on $\R^d$ (see Remark \ref{rem: continuity_around_singularities}).
\end{proof}
Now, we study the behavior of $V^\mathrm{MF}$ at infinity. Notice that since $u$ is radial, this is also the case for $V^\mathrm{MF}$.
\begin{lemma}[Long range behavior of $V^\mathrm{MF}$]
	\label{lemma: long range behavior_MF_potential}
	For any $\epsilon >0$ there exists $c_\epsilon>0$ such that 
	\begin{gather}
	\label{eq: estimation_MF_potential_d=3}
	\forall x \in \R^3\setminus\{0\},~ -c_\epsilon e^{-\sqrt{\abs{\mu} - \epsilon} \abs{x}} \leq V^\mathrm{MF}(x) \leq  0  \quad \text{if}\quad d=3\, , \\
	\label{eq: behavior_infinity}
	V^\mathrm{MF}(x) = \frac{m_1}{4\abs{x}^3} + \frac{9m_2}{64\abs{x}^5} + \grandO{\frac{1}{\abs{x}^{7}}}\quad \text{if}\quad d=2\, ,
	\end{gather}
	where $m_1 = \int_{\R^2} \abs{u(x)}^2 \abs{x}^2 \diff x$ and $m_2 = \int_{\R^2} \abs{u(x)}^2 \abs{x}^4 \diff x$.
\end{lemma}
\begin{proof}
	In dimension 3, Newton's theorem gives the equality (recall that  $\int \abs{u}^2 =1$)
	\begin{align*}
	\forall x \in \R^3,~(\abs{u}^{2} \ast |\cdot|^{-1})(x) 
	= \int_{\R^3} \abs{u(y)}^2 \frac{1}{\max(\abs{x},\abs{y})} \diff y
	= \frac{1}{\abs{x}} - \int_{\abs{y} \geq \abs{x}} \left(\frac{1}{\abs{x}} - \frac{1}{\abs{y}}\right) \abs{u(y)}^2 \diff y \pt
	\end{align*}
	Recalling Lemma \ref{lemma: a priori exp decay}, we get the estimate \eqref{eq: estimation_MF_potential_d=3}.
	
	In the two-dimensional case, Newton's theorem does not apply since $|\cdot|^{-1}$ is not the Green function of the Laplace operator. The multipole expansion of the potential induced by $\abs{u}^2$ produces non-negative odd order terms. In Appendix \ref{sec: expansion_formula}, we prove the expansion
	\begin{align}
	\label{potential_expansion4}
	(\abs{u}^2 \ast |\cdot|^{-1})(x)
	= \frac{1}{\abs{x}} + \frac{m_1}{4\abs{x}^3} + \frac{9m_2}{64\abs{x}^5} + \grandO{\frac{1}{\abs{x}^{7}}} \, ,
	\end{align}
	as $\abs{x} \to \infty$. Expansion \eqref{eq: behavior_infinity} is a direct consequence of \eqref{potential_expansion4} and of the definition of $V^\mathrm{MF}$.
\end{proof}
Before proving Proposition \ref{prop: u_exponential_falloff}, we state a technical lemma whose proof is given in Appendix~\ref{sec:technical-results}. 
\begin{lemma}
	\label{lemma: lemma_technical_exponential_convolution}
	Let $d\in \{2,3\}$.
	Let $\nu,k >0$ and $v : \R^d \to \R$ a radial function satisfying the pointwise exponential estimate
	\[
	\forall x\in \R^d,\quad \abs{v(x)} \leq C \frac{e^{-\nu \abs{x}}}{1+ \abs{x}^k} \, ,
	\]
	for some constant $C>0$. Then there exists a constant  $C'>0$ such that for $x \in \R^d$, we have the pointwise estimates
	\begin{gather}
	\label{eq: lemma_technical_exponential_convolution}
	\abs{(v \ast v)(x)} \leq C'\frac{e^{-\nu\abs{x}}}{1+ \abs{x}^{k + 1-d}} \et	\abs{\left(v \ast \frac{v}{|\cdot|}\right)(x)} \leq C'\frac{e^{-\nu\abs{x}}}{1+ \abs{x}^{k + \frac{3}{2} - d}} \pt
	\end{gather}
\end{lemma}
The $d=3$ case of this lemma will not be employed in the proof of Proposition \ref{prop: u_exponential_falloff} but will be of later use, especially when we will compute interaction terms in Lemma \ref{lemma: tunneling}. Now, we are able to prove Proposition \ref{prop: u_exponential_falloff}.
\begin{proof}[Proof of Proposition \ref{prop: u_exponential_falloff}]
	We focus on the two-dimensional case. The argument is inspired by the proof of \cite[Lemma 19]{gontier2020nonlinear}. We begin with the proof of the first part of \eqref{eq: u_exponential_falloff}. By \eqref{eq: behavior_infinity}, we can choose $R>0$ large enough in order to have for all $\abs{x} \geq R$
	\[
	0 \leq \abs{x}^3 V^\mathrm{MF}(x) \leq c \, ,
	\]
	for some constant $c>0$.
	We recall that the modified Bessel function $K_\alpha$ of the second kind with parameter $\alpha \in \R$ solves the ordinary differential equation
	\[
	r^2 K_\alpha''(r) + rK_\alpha'(r) - (r^2 + \alpha^2) K_\alpha(r) = 0 \, ,
	\]
	and satisfies the asymptotic $K_\alpha(r) \sim \sqrt{\frac{\pi}{2r}} e^{-r}$, see for instance \cite[Eq. 9.7.2]{abramowitz1964handbook}. This exponential decay is independent of the parameter $\alpha$. Let $\epsilon \geq 0$.
	The function $Y_\epsilon(x) \coloneqq K_\epsilon(\abs{\mu}^{1/2} \abs{x})$ satisfies the equation
	\[
	\left(-\Delta + \frac{\epsilon}{\abs{x}^2} - \mu \right)Y_\epsilon = 0 \, ,
	\]
	in the region $\abs{x} \geq R$. By the first comparison Lemma \ref{lemma: comparison_supersolution}, as $V^\mathrm{MF} \geq 0$ in the region $\abs{x}\geq R$, we have $u(x) \leq \overline{C}Y_0(x)$ for all $\abs{x} \geq R$ and some constant $\overline{C}$. For the lower bound, we choose $\epsilon>0$, $R \geq \frac{c_2}{\epsilon}$ and a constant $\underline{C}> 0$ such that $\underline{C}Y_\epsilon(x) \leq u(x)$ for all $\abs{x}=R$. Then we invoke the second comparison Lemma \ref{lemma: comparison_hoffmann_ostenhoff}, remarking that $u(x) \to 0$ as $\abs{x}\to \infty$ by the upper bound, to get
	$\underline{C}Y_\epsilon(x) \leq u(x)$ for all $\abs{x} \geq R$. To obtain the first bound in \eqref{eq: u_exponential_falloff}, it remains to use the asymptotics for the modified Bessel functions and to remark that $u$ is positive and continuous on $\R^2$ (see Remark \ref{rem: continuity_around_singularities}).
	
	To get the pointwise estimates concerning $\abs{\nabla u}$, we introduce the Yukawa potential $\tilde{Y}$ which solves in the sense of distribution the equation
	\[
	(-\Delta - \mu) \tilde{Y} = \delta_0 \, ,
	\]
	where $\delta_0$ denotes the Dirac distribution on 0. Then, the Euler-Lagrange equation satisfied by $u$ is equivalent to
	\[
	u = -(-\Delta - \mu)^{-1} \left(V^\mathrm{MF} u\right) = -\tilde{Y} \ast \left(V^\mathrm{MF}u\right) \pt
	\]
	This implies 
	\[
	\abs{\nabla u} \leq \abs{\nabla \tilde{Y}} \ast \abs{V^\mathrm{MF} u} \pt
	\]
	By \cite[Theorem 6.23]{lieb2001analysis}, we can deduce that $\abs{\nabla \tilde{Y}}$ is integrable in a neighborhood of the origin and behaves like $\tilde{Y}(x) = \grandO{u(x)}$ at infinity. Moreover, by Lemma \ref{lemma: technical2} and the fact $u \in H^1(\R^2)$,  we have $\abs{u}^2 \ast |\cdot|^{-1} \in L^\infty(\R^2)$. Hence, we have
	\[
	\abs{\nabla u} \lesssim u \ast \left(|\cdot |^{-1} u + u \right)
	\]
	Then, the second part of \eqref{eq: u_exponential_falloff} follows immediately from Lemma \ref{lemma: lemma_technical_exponential_convolution}. The integral exponential decay estimates \eqref{eq: a priori exp decay2} follows from the pointwise ones.
	
	Now, we show \eqref{eq: a priori exp decay3}. Let $R>0$ and $\epsilon \in \intoo{0}{R}$. Let $\chi \in \mathcal{C}^\infty(\R^d)$ a cutoff function such that
	\begin{gather*}
	\chi \in W^{1,\infty}(\R^d) \, ,\quad 0 \leq \chi \leq 1\, ,\quad \chi \equiv 1 \quad \text{on} \quad \{\abs{x} \geq R\}\, , \\
	\quad \supp \chi \subset B(0,R-\epsilon) \et \supp \nabla \chi \subset B(0,R) \setminus B(0,R-\epsilon) \pt
	\end{gather*}
	Let $\alpha = 0$ if $d=3$ and $\alpha \in \intoo{0}{1}$ if $d=2$.
	By Proposition \ref{prop: study of d(h)}, we have
	\begin{align*}
	\norme{u}_{H^{2-\alpha}(B(0,R)^c)} 
	&\leq \norme{\chi u}_{H^{2-\alpha}(\R^d)} \leq \norme{h \chi u}_{L^2(\R^d)} + C \norme{\chi u }_{L^2(\R^d)} \\
	&\leq (\abs{\mu} + C) \norme{\chi u}_{L^2(\R^d)} + \norme{\chi}_{W^{1,\infty}(\R^d)} \norme{u}_{H^1(B(0,R-\epsilon)^c)} \\
	&\lesssim \norme{u}_{H^1(B(0,R-\epsilon)^c)} \pt
	\end{align*}
	We use estimate \eqref{eq: a priori exp decay2} and choose $\epsilon$ small enough in order to show estimate \eqref{eq: a priori exp decay3}.
	
	The estimates in dimension 3 are shown as in dimension 2.
	But in fact in \cite[Appendix 2]{catto1993bindingIII} the authors give a more precise statement for \eqref{eq: u_exponential_falloff} in the case $d=3$. Indeed, they show the existence of a constant $a \in \R$ such that 
	\[
	u(x) e^{\abs{\mu}^{1/2} \abs{x}} \abs{x} \limit{\abs{x} \to \infty} a \pt \qedhere
	\]
\end{proof}

\subsection{Some interaction terms}\label{sec:some-tunneling-terms}

We recall that $u_L^r = u(\cdot - \mathbf{x}_L)$ and $u_L^\ell = u(\cdot + \mathbf{x}_L)$ where $u$ is the unique positive minimizer of the monoatomic model \eqref{eq: min_problem_ow} and that $u_L^r $ and $u_L^\ell$ satisfy the Euler-Lagrange equation \eqref{eq: euler_lagrange_equation_one_nucleus}.
As in \cite{olgiati2020hartree}, we introduce a \emph{tunneling term}
\begin{align}
\label{eq: tunneling_term}
\boxed{T_L = \exp \left(-\abs{\mu}^{\frac{1}{2}} L \right) \, ,}
\end{align}
which is the relevant scale to measure the interaction intensity between $u^r_L$ and $u^\ell_L$. In the next proposition, we gather many estimates involving $T_L$.
\begin{lemma}[Interaction terms I]
	\label{lemma: tunneling}
	Let $d\in \{2,3\}$. As $L \to \infty$, we have the estimates
	\begin{gather}
	\label{tunneling1}\int_{\R^d} u^\ell_L u^r_L = \grandO{L^{\frac{d-1}{2}} T_L} \, ,\quad D(u^\ell_L u^r_L,u^\ell_L u^r_L) = \grandO{L T_L^2} \, ,\quad D(\lvert u^\ell_L\rvert^2,u^\ell_L u^r_L) = \grandO{T_L} \, ,\\
	\label{tunneling2}\int_{\R^d} \nabla u^\ell_L \cdot \nabla u^r_L =  \grandO{L^{\frac{d-1}{2}} T_L} \, ,\quad \norme{u^\ell_L u^r_L}_{L^2(\R^d)} = \grandO{ T_L} \, , \quad\int_{\R^d} V_Lu^\ell_L u^r_L = \grandO{L^{\frac{d-2}{2}} T_L} \pt
	\end{gather}
\end{lemma}
\begin{proof}
	We begin with the first estimate of \eqref{tunneling1}. We remark that
	\[
	\int_{\R^d} u^\ell_L u^r_L = (u \ast u)(2\mathbf{x}_L) \, ,
	\]
	where $\abs{2\mathbf{x}_L}=L$.
	Then $\int_{\R^d} u^\ell_L u^r_L = \grandO{L^{\frac{d-1}{2}} T_L}$ is a direct consequence of Lemma \ref{lemma: lemma_technical_exponential_convolution} and of the exponential estimates from Proposition \ref{prop: u_exponential_falloff}. The estimates of \eqref{tunneling2} are obtained in the same way.
	To obtain the second estimate of \eqref{tunneling1}, we use the Hardy-Littlewood-Sobolev inequality. The third estimate of \eqref{tunneling1} comes from $D(u^\ell_L u^r_L,u^\ell_L u^r_L) = \grandO{L T_L^2}$ and the Cauchy-Schwarz inequality for the Coulomb energy \cite[Theorem 9.8]{lieb2001analysis}: $D(v,w) \leq \sqrt{D(v,v)} \sqrt{D(w,w)}$.
\end{proof}
All the interaction terms we consider in Lemma \ref{lemma: tunneling} only depend on the long range behavior of $u$. These are thus exponentially small, irrespective of the dimension. In the next lemma, we study interaction terms which appear to be exponentially small in dimension 3 but not in dimension~2.
We recall that $m_1 = \int_{\R^2} \abs{u(x)}^2 \abs{x}^2 \diff x$ and $m_2 = \int_{\R^2} \abs{u(x)}^2 \abs{x}^4\diff x$.
\begin{lemma}[Interaction terms II]
	\label{lemma: interaction_left_right_d=3}
	We have
	\begin{gather}
	\label{eq: interaction_terms_d=3_first}
	\int_{\R^3} V^\ell_L \abs{u^r_L}^2 = 
	\begin{cases}
	\displaystyle- \left(\frac{1}{L} + \frac{m_1}{4L^3} + \frac{9m_2}{64L^5}\right) + \grandO{\frac{1}{L^{7}}}& \text{if }d=2 \, ,\\
	\displaystyle- \frac{1}{L} + \grandO{\frac{T_L^2}{L^2}} & \text{if }d=3 \, ,
	\end{cases} \\
	\label{eq: interaction_terms_d=3_second}
	D(\lvert u^\ell_L\rvert^2,\lvert u^r_L\rvert^2) =
	\begin{cases}
	\displaystyle \frac{1}{2L} + \frac{m_1}{4L^3} +  \frac{9(m_2 + 2m_1^2)}{ 64L^5} + \grandO{\frac{1}{L^{7}}}& \text{if }d=2 \, ,\\
	\displaystyle \frac{1}{2L} + \grandO{T_L^2} & \text{if }d=3 \pt
	\end{cases}
	\end{gather}
\end{lemma}
\begin{proof}
	We start with the three-dimensional case. By Newton's theorem, we have
	\begin{align*}
	\int_{\R^3} V^\ell_L \abs{u^r_L}^2 = -\left(\abs{u}^2 \ast |\cdot|^{-1}\right)(2\mathbf{x}_L) = -\frac{1}{L} + \int_{\abs{y} \geq L} \abs{u(y)}^2 \left(\frac{1}{L} - \frac{1}{\abs{y}}\right)\diff y\pt
	\end{align*}
	To estimate the last integral, we use the pointwise exponential bound on $u$ from Proposition \ref{prop: u_exponential_falloff} and get
	\[
	\abs{\int_{\abs{y} \geq L} \abs{u(y)}^2 \left(\frac{1}{L} - \frac{1}{\abs{y}}\right)\diff y}
	\lesssim \frac{e^{-2\sqrt{\abs{\mu}}}L}{2\sqrt{\abs{\mu}}L} - E_1(2\sqrt{\abs{\mu}}L) = \grandO{L^{-2}T_L^2}\, ,
	\]
	where $E_1(z) \coloneqq \int_z^\infty \frac{e^{-t}}{t} \diff t$ denotes the exponential integral with parameter 1. We get the $O$ by using the asymptotic expansion for $E_1$, see \cite[Eq. 5.1.51]{abramowitz1964handbook}
	\[
	E_1(z) = \frac{e^{-z}}{z} \left( \sum_{k =0}^n \frac{(-1)^k k!}{z^k} + \grandO{\frac{1}{\abs{z}^{n+1}}}\right) \pt
	\]
	This shows \eqref{eq: interaction_terms_d=3_first} when $d=3$.
	To get \eqref{eq: interaction_terms_d=3_second}, we use again Newton's theorem
	\begin{align*}
	D(\lvert u^\ell_L\rvert^2,\lvert u^r_L\rvert^2) 
	= \frac{1}{2} \left(\abs{u}^2\ast \abs{u}^2\ast |\cdot|^{-1}\right)(2\mathbf{x}_L)
	= \frac{1}{2} \left(\left(\abs{u}^2 \ast |\cdot|^{-1}\right)(2\mathbf{x}_L) - \left(\abs{u}^2 \ast \varphi_u\right)(2\mathbf{x}_L)\right) \, ,
	\end{align*}
	with $\varphi_u(x) = \int_{\abs{y} \geq \abs{x}} \abs{u(y)}^2 \left(\frac{1}{\abs{x}} - \frac{1}{\abs{y}}\right)\diff y$. The same computation as above shows that $\varphi_u(x) = \grandO{\abs{x}^{-2}e^{-2\sqrt{\abs{\mu}}}\abs{x}}$. It remains to apply Lemma \ref{lemma: lemma_technical_exponential_convolution} to conclude.
	
	In the two-dimensional case, the expansion \eqref{eq: interaction_terms_d=3_first} is a consequence of \eqref{potential_expansion4} and the fact that
	\[
	\int_{\R^2} V^\ell_L \abs{u^r_L}^2 = - \left(\abs{u}^2 \ast |\cdot|^{-1}\right)(2\mathbf{x}_L) \pt
	\]
	Now, we show \eqref{eq: interaction_terms_d=3_second}.
	By translation invariance of $D(\cdot,\cdot)$, we have $D(\lvert u^\ell_L\rvert^2,\lvert u^r_L\rvert^2) = D((\tau^r_L)^2\abs{u}^2,\abs{u}^2)$. If we denote by $B_L$  the ball centered at the origin and with radius $L/4$, we have
	\begin{align*}
	D(\lvert u^\ell_L\rvert^2,\lvert u^r_L\rvert^2)
	= D((\tau^r_L)^2 (\mathds{1}_{B_L^c}\abs{u}^2), \abs{u}^2) + D( (\tau^r_L)^2(\mathds{1}_{B_L}\abs{u}^2), \abs{u}^2) \pt
	\end{align*}
	The first term is coarsely estimated by the Hardy-Littlewood-Sobolev inequality
	\[
	D((\tau^r_L)^2(\mathds{1}_{B_L^c}\abs{u}^2),  \abs{u}^2) \lesssim \lVert u\rVert^2_{L^{8/3}(B_L^c)} = \grandO{\frac{1}{L^7}} \pt
	\]
	To get the $\grandO{L^{-7}}$, we have used the exponential bounds on $u$ (see Proposition \ref{prop: u_exponential_falloff}).	
	Now, we turn our attention to the second term. We recall the expansion formula \eqref{potential_expansion4} 
	\begin{align}
	\label{eq: on_avance}
	\left( \abs{u}^2 * |\cdot|^{-1}\right)(x) = \frac{1}{\abs{x}} + \frac{m_1}{4\abs{x}^3} + \frac{9m_2}{64\abs{x}^5} + \grandO{\frac{1}{\abs{x}^{7}}} \, ,
	\end{align}
	as $\abs{x}$ goes to $\infty$. In Appendix \ref{sec: expansion_formula}, we also show that for all $a>0$
	\begin{align}
	\label{eq: mais_c_est_dure}
	(\mathds{1}_{B_L}\abs{u}^2 \ast |\cdot|^{-a})(2\mathbf{x}) = \frac{1}{L^{a}} + \frac{a^2}{4}\frac{m_1}{L^{2+a}} + \left(\frac{a^2}{16} + \frac{a^3}{16}+\frac{a^4}{64} \right)\frac{ m_2}{L^{4+a}} + \grandO{\frac{1}{L^{7}}} \pt
	\end{align}
	Inserting \eqref{eq: on_avance} into
	\[
	D(( \tau^r_L)^2(\mathds{1}_{B_L}\abs{u}^2), \abs{u}^2) = \frac{1}{2}\left(\mathds{1}_{B_L} \abs{u}^2 \ast \left( \abs{u}^2 \ast |\cdot|^{-1}\right)\right)(2\mathbf{x}_L) \, ,
	\]
	and using \eqref{eq: mais_c_est_dure} give \eqref{eq: interaction_terms_d=3_second}.
\end{proof}

\section{Construction of quasi-modes}\label{sec:construction-of-quasi-modes}

In this section, we consider a sequence $(L_n)_{n\in \N}$ such that $L_n \to \infty$ as $n\to \infty$.  We want to show that to leading order the solution $u_{L_n}^+$ of the minimization problem \eqref{eq: restricted_hartree} splits into two parts which are given by the monoatomic model minimizer translated by $+ \mathbf{x}_{L_n}$ or $-\mathbf{x}_{L_n}$.

For simplicity, we will use the subscript $n$ instead of $L_n$. For instance, we will write $\mathbf{x}_n$ instead of $\mathbf{x}_{L_n}$. We will use the $r$ (resp. $\ell$) superscript for quantities related to the monoatomic model translated by $\mathbf{x}_n$ (resp. $-\mathbf{x}_n$), see \eqref{eq: definition_translated_potential} and \eqref{eq: euler_lagrange_equation_one_nucleus}.

Let $d\in \{2,3\}$. It will be handful to define a partition of unity $(\chi^\ell)^2+(\chi^r)^2 =1$ such that for all $x \in\R^d$ and $\kappa \in \{\ell,r\}$
\[
\chi^\kappa \in \mathcal{C}^\infty(\R^d),\quad 0 \leq \chi^\kappa (x) \leq 1,\quad  \norme{\nabla \chi^\kappa}_{L^\infty}< \infty \et \mathcal{R}[\chi^r] = \chi^\ell \, ,
\]
where $\mathcal{R}$ is the natural action of the reflection symmetry with respect to the hyperplane $\{x_1=0\}$ on $L^2(\R^d)$, see \eqref{eq: def_reflection}. Moreover, we assume that there exists $\delta > 0$ such that
\[
\supp \chi^\ell \subset\{ x_1 \leq \delta\} \et \supp \chi^r \subset \{x_1 \geq - \delta\} \pt
\]
Then, we denote the $L_n-$dependent partition of unity by
\[
\chi^\ell_n(x) = \chi^\ell(x/\sqrt{L_n}) \et \chi^r_n (x) = \chi^r(x/\sqrt{L_n}) \pt
\]
We have, in particular
\begin{gather}
\label{eq: properties_chi1}
\supp \chi^\ell_n \subset\{x_1 \leq \sqrt{L_n}\delta\},\quad \supp \nabla \chi^\ell_n\subset\{\abs{x_1}\leq \sqrt{L_n}\delta\} \, , \\
\label{eq: properties_chi2}\norme{\nabla \chi_n^\ell}_{L^\infty(\R^d)} = \grandO{L_n^{-1/2}} \, , \quad \norme{\Delta \chi_n^\ell}_{L^\infty(\R^d)} = \grandO{L_n^{-1}} \, ,
\end{gather}
and similar statements for $\chi^r_n$. The cutoff function $\chi_n^r$ (resp. $\chi_n^\ell$) will be useful to localize the right (resp. left) part of $u_n^+$. The scale $\sqrt{L_n}$ is chosen for convenience, any $L_n^{\beta}$ with $\beta < 1$ would do.

Finally, we introduce the left and right translation operators
\begin{align}
\label{eq: translation_operators}
\tau_n^r \coloneqq \tau_{\mathbf{x}_n} \et \tau_n^\ell \coloneqq \tau_{-\mathbf{x}_n} \pt
\end{align}

\subsection{Approximation for \texorpdfstring{$u_n^+$}{un+} in \texorpdfstring{$H^1(\R^d)$}{H1(Rd)}}\label{sec:first-approximation-for-u}

Our first result shows that $u_n^+$ splits into two symmetric parts which both converge, up to a translation, toward the minimizer $u$ of the one nucleus model~\eqref{eq: min_problem_ow} in $H^1(\R^d)-$norm.
\begin{prop}[Strong convergence in $H^1(\R^d)$]
	\label{prop: strong_convergence_H1}
	We have the strong convergence
	\begin{align}
	\label{a priori strong convergence1}
	\norme{\chi^\ell_n u_n^+ - u(\cdot + \mathbf{x}_n)}_{H^1(\R^d)} = \norme{\chi^r_n u_n^+ - u(\cdot - \mathbf{x}_n)}_{H^1(\R^d)} \limit{n \to \infty}0 \pt
	\end{align}
\end{prop}
\begin{proof}[Proof]
	In this proof, we will frequently use the interaction estimates from Lemma \ref{lemma: tunneling} and Lemma~\ref{lemma: interaction_left_right_d=3}. We recall that $u_n^r = u(\cdot - \mathbf{x}_n)$ and $u_n^\ell = u(\cdot+\mathbf{x}_n)$. We insert the trial state
	\[
	u_\mathrm{trial} = \frac{ \sqrt{2}(u^r_n +  u^\ell_n)}{\norme{ u^r_n +  u^\ell_n}_{L^2(\R^d)}}
	\]
	into the expression of the energy functional $\mathcal{E}_n$ defined in \eqref{eq: restricted_hartree_functional} and look for an upper bound involving $\mathcal{E}(u)$. We give a precise estimate of the interaction energy between $u^\ell_n$ and $u^r_n$
	\begin{align*}
	\mathcal{E}_n(u_\mathrm{trial}) - 2 \mathcal{E}(u) \, ,
	\end{align*}
	which will also be useful later in the proof of Proposition \ref{prop: rate_convergence}.
	First, using that $\norme{ u^\ell_n}_{L^2(\R^d)} = \norme{u^r_n}_{L^2(\R^d)} = 1$, we have
	\[
	\norme{ u^r_n +  u^\ell_n}^{-2}_{L^2(\R^d)}
	=\left(2 + 2\pdtsc{u^r_n}{ u^\ell_n}_{L^2(\R^d)}\right)^{-1} = \frac{1}{2} + \grandO{L_n^{\frac{d-1}{2}}T_n} \, ,
	\]
	where $T_n = T_{L_n} \coloneqq \exp(-\abs{\mu}^{1/2}L_n)$ is the tunneling term introduced in \eqref{eq: tunneling_term}.
	Then, computing the energy, we have
	\begin{align*}
	\mathcal{E}_n(u_\mathrm{trial}) &= \left(1+ \grandO{L_n^{\frac{d-1}{2}}T_n}\right) \mathcal{E}_n(u^\ell_n + u^r_n) \\
	&\leq 2\left(1+ \grandO{L_n^{\frac{d-1}{2}}T_n}\right) \left( \mathcal{E}(u) + \int_{\R^d} \nabla u^r_n \cdot \nabla u^\ell_n + 2\int_{\R^d} V^\ell_n u^\ell_n u^r_n + 4 D(\lvert u^\ell_n \rvert^2,u^\ell_n u^r_n) \right. \\
	& \left. + \int_{\R^d} V^\ell_n \lvert u^r_n \rvert^2 + D(\lvert u^\ell_n \rvert^2, \lvert u^r_n \rvert^2) + \frac{1}{2L_n} + 2 D(u^\ell_n u^r_n,u^\ell_n u^r_n) \right) \pt
	\end{align*}
	In this previous estimate, we have used the reflection symmetry and the translation invariance of the integral and, more particularly, the following identities
	\begin{gather*}
	\int_{\R^d} V^\ell_n \lvert u^r_n \rvert^2 = \int_{\R^d} V^r_n \lvert u^\ell_n \rvert^2 \, ,\quad \int_{\R^d} V^\ell_n u^\ell_n u^r_n = \int_{\R^d} V^r_n u^\ell_n u^r_n\, , \quad D(\lvert u^r_n \rvert^2, u_\ell u^r_n) = D(\lvert u^\ell_n \rvert^2, u^\ell_n u^r_n) \, , \\
	\mathcal{E}^\ell_n(u^\ell_n) = \mathcal{E}^r_n(u^r_n) = \mathcal{E}(u) \pt
	\end{gather*}
	Multiplying by $u^r_n$ the Euler-Lagrange equation \eqref{eq: euler_lagrange_equation_one_nucleus} satisfied by $u^\ell_n$ and integrating, we obtain 
	\[
	\int_{\R^d} \left(\nabla u^\ell_n \cdot \nabla u^r_n + V^\ell_n u^\ell_n u^r_n\right) + 2 D(\lvert u^\ell_n \rvert^2, u^\ell_n u^r_n) = \mu \int_{\R^d} u^\ell_n u^r_n \pt
	\]
	We can thus simplify our estimate into
	\begin{align*}
	\mathcal{E}_n(u_\mathrm{trial}) 
	=&~ 2 \left(1+ \grandO{L_n^{\frac{d-1}{2}}T_n}\right)\left(\mathcal{E}(u) - \int_{\R^d} \nabla u^\ell_n \cdot \nabla u^r_n + 2\mu \int_{\R^d} u^\ell_n u^r_n + \int_{\R^d} V^\ell_n \lvert u^r_n \rvert^2 \right. \\
	&\left. + D(\lvert u^\ell_n \rvert^2, \lvert u^r_n \rvert^2) + \frac{1}{2L_n} + 2 D(u^\ell_n u^r_n,u^\ell_n u^r_n) \right) \pt
	\end{align*}
	In the $d=2$ case, we use the estimates from Lemma \ref{lemma: tunneling} and the expansions from Lemma \ref{lemma: interaction_left_right_d=3} to get
	\begin{align}
	\label{eq: upper_bound_energy}
	\mathcal{E}_n(u_n^+) \leq \mathcal{E}_n(u_\mathrm{trial}) = 2 \mathcal{E}(u) + \left(\frac{3\mu_1}{4}\right)^2 \frac{1}{L_n^5} + \grandO{\frac{1}{L_n^7}} \, ,
	\end{align}
	where $m_1 = \int_{\R^2} \abs{u(y)}^2 \abs{y}^2 \diff y$ is the second moment of the probability distribution $\abs{u}^2$. In the $d=3$ case, we use Lemma \ref{lemma: interaction_left_right_d=3} and find
	\begin{align}
	\label{eq: upper_bound_energy_d=3}
	\mathcal{E}_n(u_n^+) \leq \mathcal{E}_n(u_\mathrm{trial}) = 2 \mathcal{E}(u) + \grandO{L_nT_n} \pt
	\end{align}
	In both cases, passing to the limit in \eqref{eq: upper_bound_energy} and \eqref{eq: upper_bound_energy_d=3}, we get
	\begin{align}
	\label{eq: a priori upper bound_bis}
	\limsup_{n\to\infty} \mathcal{E}_n(u_n^+) \leq \lim_{n\to\infty} \mathcal{E}_n(u_\mathrm{trial}) \leq 2 \mathcal{E}(u) = 2I \pt
	\end{align}
	Next, we seek a lower bound on the energy of $u_n^+$.
	Using the IMS formula \cite[Theorem 3.2]{cycon1987schrodinger}, the positivity of $D(\cdot, \cdot)$ and the reflection symmetry $\mathcal{R}$, we can bound $\mathcal{E}_n(u_n^+)$ from below as follows
	\begin{align*}
	\mathcal{E}_n(u_n^+) 
	=& \int_{\R^d} \lvert \nabla (\chi^\ell_n u_n^+) \rvert^2 + \int_{\R^d} \lvert \nabla (\chi^r_n u_n^+) \rvert^2 - \int_{\R^d} \left(\lvert \nabla \chi^\ell_n\rvert^2+\lvert \nabla \chi^r_n \rvert^2\right) \lvert u_n^+ \rvert^2 + \frac{1}{L_n} \\
	&+ \int_{\R^d} V_n (\lvert\chi^\ell_n u_n^+\rvert^2 + \lvert\chi^r_nu_n^+\rvert^2) + D(\lvert u_n^+ \rvert^2,\lvert u_n^+ \rvert^2) \\
	\geq & ~\mathcal{E}^\ell_n(\chi^\ell_n u_n^+) + \mathcal{E}^r_n(\chi^r_nu_n^+) + 2\left( \int_{\R^d} V^\ell_n \lvert\chi^r_nu_n^+\rvert^2 + D\left(\abs{\chi^\ell_n u_n^+}^2,\abs{\chi^r_n u_n^+}^2\right) -\int_{\R^d} \lvert\nabla \chi^\ell_n\rvert^2 \lvert u_n^+ \rvert^2\right)  \\
	\geq &~ 2 \mathcal{E}^\ell_n(\chi^\ell_n u_n^+) + \grandO{L_n^{-1}} \pt
	\end{align*}
	The $\grandO{L_n^{-1}}$ comes from the properties of $\chi_n^\ell$, see \eqref{eq: properties_chi1} and \eqref{eq: properties_chi2}.
	Passing to the limit, we find
	\[
	\liminf_{n\to\infty} \mathcal{E}_n(u_n^+) \geq 2 \liminf_{n\to\infty} \mathcal{E}^\ell_n(\chi^\ell_n u_n^+) \pt
	\]
	If we insert this lower bound into \eqref{eq: a priori upper bound_bis}, we deduce 
	\[
	\lim\limits_{n \to \infty} \mathcal{E}(\tau^r_n(\chi^\ell_n u_n^+) ) = \mathcal{E}(u) = I \, ,
	\]
	where $\tau_n^r$ is the right translation operator defined by \eqref{eq: translation_operators}.
	Moreover, by reflection symetry, $\norme{\chi^\ell_n u_n^+}_{L^2(\R^2)}=\norme{\chi^r_n u_n^+}_{L^2(\R^2)}=1$. Thus $\tau^r_n(\chi^\ell_n u_n^+)$ is a minimizing sequence for the minimization problem \eqref{eq: min_problem_ow}. By reflection symmetry, $\tau^\ell_n(\chi^r_nu_n^+)$ is also minimizing sequence for $I$. By the precompactness of all minimizing sequences (see \cite[Section VII]{lieb1981thomasfermi}), we have the strong convergences
	\begin{align*}
	\tau^r_n(\chi^\ell_n u_n^+) \limit{n \to \infty} u \et \tau^\ell_n(\chi^r_n u_n^+) \limit{n\to \infty} u \quad \text{in} \quad H^1(\R^d) \, ,
	\end{align*}
	up to a subsequence.
	However, since $u$ is unique up to a phase and $u_n^+ >0$, this statement is true for every sequence $(L'_n)_{n\geq 0}$ such that $L'_n \to \infty$ as $n\to \infty$. Hence, the sequences $\tau^r_n(\chi^\ell_n u_n^+)$ and $\tau^\ell_n(\chi^r_n u_n^+)$ admit only one limit point in $H^1(\R^d)$ and every subsequence admits a converging subsequence in $H^1(\R^d)$ (which converges toward $u$). This shows \eqref{a priori strong convergence1}.
\end{proof}
\begin{cor}[Uniform bound in $H^1(\R^d)$]
	\label{cor: boundedness_u}
	We have
	\begin{align}
	\label{eq: boundedness_u}
	\sup_{n\geq 0} \norme{u_n^+}_{H^1(\R^d)} <\infty \pt
	\end{align}
\end{cor}
\begin{proof}
	Estimate \eqref{eq: upper_bound_energy} from the proof of Proposition \ref{prop: strong_convergence_H1} gives
	\[
	\sup_{n\geq 0} \mathcal{E}_n(u_n^+) < \infty \pt
	\]
	Using the positivity of $D(\cdot,\cdot)$ and the fact that $|\cdot|^{-1}$ is infinitesimally form-bounded with respect to $-\Delta$ (see for instance \cite[Theorem X.19]{reed1975methodsII}), we have $\mathcal{E}_n(v) \geq \epsilon\norme{\nabla u_n^+}^2_{L^2(\R^d)} - C_\epsilon $ for all $\epsilon>0$ and some constant $C_\epsilon$. The uniform bound \eqref{eq: boundedness_u} follows immediately.
\end{proof}
\begin{rem}
	\label{rem: regularity_h_L}
	A consequence of Corollary  is the validity of Proposition \ref{prop: study of d(h)} when $h$ is replaced by $h_L$. Indeed, in the last step of the proof, we need to use the fact that $ \sup_{L \geq 1} \norme{\rvert u_L^+ \lvert^2 \ast |\cdot|^{-1}}_{L^\infty(\R^d)} < \infty$ which comes from Lemma \ref{lemma: technical2_inequality1} and Corollary \ref{cor: boundedness_u}. Also, in that case there are $2^n$ terms of order $n$ appearing in \eqref{eq: that}. These are bounded by
	\begin{align*}
	&\norme{(-\Delta+\nu)^{1/2 - \alpha} \left[(-\Delta+\nu)^{-1/2} \left(\abs{x - \mathbf{x}_L}^{-1}+ \abs{x - \mathbf{x}_L}^{-1}\right) (-\Delta+\nu)^{-1/2}\right]^n (-\Delta+\nu)^{-1/2}} \\
	&\leq \nu^{-\alpha} \left(\frac{2C_\alpha}{\sqrt{\nu}}\right)^n \pt
	\end{align*}	
\end{rem}

\subsection{Convergence of the Lagrange multiplier}

Using Proposition \ref{prop: strong_convergence_H1}, we show that the lowest eigenvalue $\mu_n^+$ of the mean-field hamiltonian $h_n$ for the two nuclei model \eqref{eq: restricted_hartree} converges toward the one associated with the one nucleus model \eqref{eq: min_problem_ow} as $n \to \infty$.
\begin{prop}
	\label{lemma: convergence energy}
	We have $\mu_n^+ \to \mu$ as $n\to \infty$.
\end{prop}
\begin{proof}
	We recall the Euler-Lagrange equation \eqref{eq: EL_equation} which is satisfied by $u^+_n$
	\begin{align*}
	\left(-\Delta + V_n + \lvert u^+_n \rvert^2 \ast |\cdot|^{-1}\right) u^+_n = \mu_n^+u^+_n \, ,
	\end{align*}
	where $\mu_n^+ < 0$. We localize this equation in the sector $\{x_1 \leq \sqrt{L_n} \delta\}$ using the cutoff function $\chi^\ell_n$
	\begin{align}
	\label{eq: EL_localized}
	\left(-\Delta + V_\ell + \lvert\chi^\ell_n u^+_n\rvert^2 \ast |\cdot|^{-1}\right) \chi^\ell_n u^+_n 
	=& ~(-\Delta \chi^\ell_n) u^+_n - 2\nabla \chi^\ell_n \cdot \nabla u^+_n - V^r_n \chi^\ell_n u^+_n\\
	&\nonumber - \left(\lvert\chi^r_n u^+_n\rvert^2 \ast |\cdot|^{-1}\right) \chi^\ell_n u^+_n +  \mu_n^+ \chi^\ell_n u^+_n \pt \nonumber
	\end{align}
	We denote
	\[
	w_n \coloneqq(-\Delta \chi^\ell_n) u^+_n - 2\nabla \chi^\ell_n \cdot \nabla u^+_n - V^r_n \chi^\ell_n u^+_n  \pt
	\]
	From Corollary \ref{cor: boundedness_u}, $\norme{u^+_n}_{H^1(\R^d)}$ is uniformly bounded with respect to $n$. Hence, the terms involving the derivatives of $\chi^\ell_n$ are $\grandO{L_n^{-1/2}}$ in $L^2(\R^d)-$norm. Moreover, since $\supp \chi^\ell_n \subset \{x_1 \leq \sqrt{L_n}\delta\}$, we have $\norme{V^r_n \chi^\ell_n u^+_n}_{L^2(\R^d)} = \grandO{L_n^{-1}}$.
	Hence $\norme{w_{n}}_{L^2(\R^d)} \to 0$ as $n\to \infty$. 
	Now, we multiply \eqref{eq: EL_localized} by $\chi^\ell_n u^+_n$, integrate and translate by $+\mathbf{x}_n$ to get
	\begin{align}
	\label{eq: lemma convergence energy}
	\mathcal{E}(\tau^r_n(\chi^\ell_n u^+_n)) + D\left(\abs{\tau^r_n(\chi^\ell_n u^+_n)}^2,\abs{\tau^r_n(\chi^\ell_n u^+_n)}^2\right) 
	= \pdtsc{w_n}{\chi^\ell_n u^+_n}_{L^2(\R^d)} + \mu_n^+ \pt
	\end{align}
	We have used the fact that $\int_{\R^d} \abs{\tau^r_n(\chi^\ell_n u^+_n)}^2 = \int_{\R^d} \abs{\chi^\ell_n u^+_n}^2= 1$.
	By Proposition \ref{prop: strong_convergence_H1}, $\tau^r_n(\chi^\ell_n u^+_n) - u$ converges to 0 in $H^1(\R^d)$ as $n\to \infty$. So, by the strong continuity of the energy functional $v\mapsto\mathcal{E}(v)$ and of the Coulomb energy $v \mapsto D(v,v)$ in $H^1(\R^d)-$norm, we can take the limit as $n\to \infty$ in the left term of \eqref{eq: lemma convergence energy}
	\[
	\lim\limits_{n\to \infty}\left(\mathcal{E}(\tau^r_n(\chi^\ell_n u^+_n)) + D\left(\abs{\tau^r_n(\chi^\ell_n u^+_n)}^2,\abs{\tau^r_n(\chi^\ell_n u^+_n)}^2\right)\right) = \pdtsc{h u}{u}_{L^2(\R^d)} = \mu \pt
	\]
	Then, taking the limit into \eqref{eq: lemma convergence energy} leads to $\lim\limits_{n \to\infty} \mu_n^+ = \mu$.
\end{proof}
The next corollary results from Remark \ref{rem: regularity_h_L} and Proposition \ref{lemma: convergence energy}.
\begin{cor}[Uniform bound in $H^{2-\alpha}(\R^d)$]
	\label{cor: boundedness_u_better}
	$u_n^+$ is uniformly bounded in $H^{2-\alpha}(\R^d)$ for $\alpha \in \intof{0}{1}$ if $d=2$ and $\alpha=0$ if $d=3$. In particular, $u_n^+$ is uniformly bounded in $L^\infty(\R^d)$.
\end{cor}

\subsection{Approximation for \texorpdfstring{$u_n^+$}{un+} in higher Sobolev spaces}

The convergence of $\mu^+_n$ toward $\mu$ allows us to extend the result stated in Proposition \ref{prop: strong_convergence_H1} by showing that strong convergence also holds in $H^{2-\alpha}(\R^2)$ for any $\alpha >0$ if $d=2$ and in $H^2(\R^3)$ if $d=3$. 
\begin{prop}[Strong convergence in higher Sobolev spaces]
	\label{prop: stronger_convergence}
	Let $0< \alpha \leq 1$ if $d=2$ and $\alpha=0$ if $d=3$. We have the strong convergence
	\[
	\norme{\chi^\ell_n u^+_n - u(\cdot + \mathbf{x}_n)}_{ H^{2 - \alpha}(\R^d)} = \norme{\chi^r_n u^+_n - u(\cdot - \mathbf{x}_n)}_{ H^{2 - \alpha}(\R^d)} \limit{n \to \infty} 0 \pt
	\]
\end{prop}
\begin{proof}
	In the two-dimensional case, this is a consequence of the interpolation inequality $\norme{v}_{H^{2-2\alpha}(\R^d)} \leq \norme{v}_{H^{1}(\R^d)}^{\theta} \norme{v}_{H^{2-\alpha}(\R^d)}^{1-\theta}$ (for $\alpha \leq 3/4$ and with $\theta = \frac{\alpha}{3/2-\alpha}$), Proposition~\ref{lemma: convergence energy} and Corollary \ref{cor: boundedness_u_better}. However, this argument does not adapt to the three-dimensional case and $\alpha =0$. So we treat it apart. By reflection symmetry, we have 
	\[
	\norme{\chi^\ell_n u^+_n - u(\cdot + \mathbf{x}_n)}_{ H^{2}(\R^3)} = \norme{\chi^r_n u^+_n - u(\cdot - \mathbf{x}_n)}_{ H^{2}(\R^3)} \pt
	\]
	We recall that $h_n^\ell$ and $u_n^\ell$ are defined by \eqref{eq: definition_translated_potential} and \eqref{eq: euler_lagrange_equation_one_nucleus}. In particular, $u_n^\ell$ solves the Euler-Lagrange equation $h_n^\ell u_n^\ell = \mu u_n^\ell$. As a direct consequence of Proposition \ref{prop: study of d(h)}, we have the bound
	\[
	\norme{v}_{H^{2}(\R^3)} \leq \norme{h_n^\ell v}_{L^2(\R^3)} + C \norme{v}_{L^2(\R^3)} \, ,
	\]
	for any $v \in \mathcal{D}(h_n^\ell)$. By Proposition \ref{prop: strong_convergence_H1}, we already have $\lim\limits_{n\to \infty} \norme{\chi_n^\ell u_n^+ - u_n^\ell}_{L^2(\R^3)} = 0$.	Thus, to prove Proposition \ref{prop: stronger_convergence} in the three-dimensional case, it is sufficient to prove that $h_n^\ell(\chi_n^\ell u_n^+ - u_n^\ell)$ converges to 0 in $L^2(\R^3)-$norm as $n \to \infty$. As in the proof of Proposition \ref{lemma: convergence energy}, we localize the Euler-Lagrange equation \eqref{eq: EL_equation} using the cutoff function $\chi^\ell_n$
	\begin{align*}
	(-\Delta + V_n + \lvert u^+_n \rvert^2 \ast |\cdot|^{-1}) \chi^\ell_n u^+_n = \mu^+_n \chi^\ell_n u^+_n - (\Delta \chi^\ell_n)u^+_n - 2 \nabla \chi^\ell_n \cdot \nabla u^+_n \pt
	\end{align*}
	Subtracting the Euler-Lagrange equation $h_n^\ell u_n^\ell = \mu u_n^\ell$ and rearranging the terms, we get
	\begin{align}
	\label{eq: 30_03_2020}
	h_n^\ell(\chi_n^\ell u_n^+ - u_n^\ell) 
	=&~ \mu^+_n \chi_n^\ell u_n^+ - \mu  u_n^\ell  -  V_n^r \chi^\ell_n u^+_n -  (\Delta \chi^\ell_n)u^+_n - 2 \nabla \chi^\ell_n \cdot \nabla u^+_n \\
	\label{eq: 30_03_2020_bis}
	&\notag +  \left[(\lvert u_n^\ell\rvert^2 -  \lvert \chi^\ell_n u^+_n\rvert^2)\ast |\cdot|^{-1}\right]\chi^\ell_n u^+_n -  (\lvert\chi^r_n u^+_n\rvert^2 * |\cdot|^{-1})\chi^\ell_n u^+_n  \pt
	\end{align}
	By Proposition \ref{prop: strong_convergence_H1} and Proposition \ref{prop: rate_convergence_ground_state_energy}, we have $\lim\limits_{n\to\infty}\norme{\mu^+_n \chi_n^\ell u_n^+ - \mu  u_n^\ell}_{L^2(\R^3)} = 0$.
	From the definition of $\chi_n^\ell$, we have the estimates
	\[
	\norme{\Delta \chi^\ell_n}_{L^\infty(\R^3)} = \grandO{L_n^{-1}}\, ,\quad \norme{ \nabla \chi^\ell_n}_{L^\infty(\R^3)} = \grandO{L_n^{-1/2}}\, , \quad \norme{V_n^r \chi^\ell_n u^+_n}_{L^2(\R^3)} = \grandO{L_n^{-1}} \pt
	\]
	The first two come from \eqref{eq: properties_chi2} and the last one is a consequence of \eqref{eq: properties_chi1}. Then, the three last terms of the right side of \eqref{eq: 30_03_2020} converge to 0. A direct consequence of Lemma \ref{lemma: technical2} is the continuity of the multilinear application $ (u,v,w) \in H^1(\R^3)^3 \longmapsto \left((uv) \ast |\cdot|^{-1}\right) w \in L^2(\R^3)$.
	From this previous statement, Proposition \ref{prop: strong_convergence_H1} and the uniform bound \eqref{eq: boundedness_u} in $H^1(\R^3)-$norm, we deduce
	\[
	\norme{\left[(\lvert u_n^\ell\rvert^2 -  \lvert \chi^\ell_n u^+_n\rvert^2)\ast |\cdot|^{-1}\right]\chi^\ell_n u^+_n }_{L^2(\R^3)} \limit{n\to \infty} 0 \, ,
	\]
	and the following identity
	\[
	\lim\limits_{n\to \infty}\norme{(\lvert\chi^r_n u^+_n\rvert^2 * |\cdot|^{-1})\chi^\ell_n u^+_n}_{L^2(\R^3)} = \lim\limits_{n\to\infty} \norme{(\lvert u_n^r\rvert^2 * |\cdot|^{-1}) u_n^\ell}_{L^2(\R^3)} \pt
	\]
	The right side of the previous identity is zero. Indeed, we have
	\[
	\norme{(\lvert u_n^r\rvert^2 * |\cdot|^{-1}) u_n^\ell}_{L^2(\R^3)} \lesssim \norme{\frac{1}{\abs{x -\mathbf{x}_n}+1} u_n^\ell}_{L^2(\R^3)} = \grandO{L_n^{-1}} \pt
	\]
	This concludes the proof of Proposition \ref{prop: stronger_convergence}.
\end{proof}

\subsection{\emph{A priori} exponential decay bounds for \texorpdfstring{$u_n^+$}{un+}}

We denote by $V_n^\mathrm{MF}$ the mean-field potential associated with the diatomic model. It is defined by
\begin{align}
\label{eq: mean_field_potential_diatomic}
V_n^\mathrm{MF} \coloneqq V_n + \lvert u_n^+ \rvert^2 \ast |\cdot|^{-1} \, ,
\end{align}
where $V_n(x) = -(\abs{x-\mathbf{x}_n}^{-1} + \abs{x+\mathbf{x}_n}^{-1})$ is the diatomic potential. 
\begin{prop}[Exponential decay bounds on $u^+_n$]
	\label{cor: exponential decay u^+_n}
	For all $\epsilon \in\intoo{0}{1}$ and for all $n$ large enough, there exists $C_{\epsilon}> 0$ such that the following pointwise estimates hold
	\begin{gather}
	\label{eq: a priori exp decay_u^+_n}
	\frac{1}{C_\epsilon} \left(e^{-(1+\epsilon)\abs{\mu}^{\frac{1}{2}} \abs{x - \mathbf{x}_n}} + e^{-(1+\epsilon)\abs{\mu}^{\frac{1}{2}} \abs{x+\mathbf{x}_n}}\right) \leq u^+_n(x) \leq C_{\epsilon} \left(e^{-(1-\epsilon)\abs{\mu}^{\frac{1}{2}} \abs{x - \mathbf{x}_n}} + e^{-(1-\epsilon)\abs{\mu}^{\frac{1}{2}} \abs{x+\mathbf{x}_n}}\right)  \, , \\
	\label{eq: a priori exp decay_u^+_n_nabla}
	\abs{\nabla u_n^+(x)} \leq C_{\epsilon} \left(e^{-(1-\epsilon)\abs{\mu}^{\frac{1}{2}} \abs{x - \mathbf{x}_n}} + e^{-(1-\epsilon)\abs{\mu}^{\frac{1}{2}} \abs{x+\mathbf{x}_n}}\right)   \, ,
	\end{gather}
	for all $x\in\R^d$. Let $\alpha = 0$ if $d=3$ and $\alpha \in \intof{0}{1}$ if $d=2$. There exists $C_\epsilon' >0$ such that for all $R >0$, we have
	\begin{align}
	\label{eq: a priori exp decay_u^+_n_integral}
	\norme{u_n^+}_{H^{2-\alpha}(B(-\mathbf{x}_n,R)^c) \cap B(\mathbf{x}_n,R)^c)} \leq C_\epsilon' e^{-(1-\epsilon)\abs{\mu}^{\frac{1}{2}}R} \pt
	\end{align}
\end{prop}
\begin{proof}
	We want to apply the second comparison Lemma \ref{lemma: comparison_hoffmann_ostenhoff}. Let $\epsilon \in \intoo{0}{\abs{\mu}}$. First, using Proposition~\ref{prop: stronger_convergence} and Lemma \ref{lemma: technical2}, we can approximate the two nuclei mean-field potential $V_n^\mathrm{MF}$ by a superposition of translated monoatomic mean-field potentials. For $n$ large enough, we have
	\[
	\norme{V_n^\mathrm{MF} - \left(\tau_n^\ell + \tau_n^r\right)V^\mathrm{MF}}_{L^\infty(\R^d)} \leq \epsilon/8 \pt
	\]
	Now, we choose $n$ large enough so that $\abs{\mu_n^+ - \mu} \leq \epsilon/8$. We insert these two estimates in the Euler-Lagrange equation \eqref{eq: EL_equation} solved by $u_n^+$ and we obtain
	\begin{align*}
	\left(-\Delta + \left(\tau_n^\ell + \tau_n^r\right)V^\mathrm{MF} - \mu - \epsilon/2 \right) u_n^+ \leq 0 \pt
	\end{align*}
	We have also used that $u_n^+ >0$. We have just shown that $u_n^+$ is a positive subsolution for the operator $\tilde{h}_n \coloneqq -\Delta + \left(\tau_n^\ell + \tau_n^r\right)V^\mathrm{MF} - \mu - \epsilon/2$. Moreover, by Lemma \ref{lemma: long range behavior_MF_potential}, there exists $R_\epsilon>0$ such that $\abs{\left(\tau_n^\ell + \tau_n^r\right)V^\mathrm{MF}} \leq \epsilon/2 $ on the open set $\Omega_{n,\epsilon} \coloneqq \{\abs{x\pm \mathbf{x}_n} > R_\epsilon\}$. In particular, we have $\left(\tau_n^\ell + \tau_n^r\right)V^\mathrm{MF} - \mu - \epsilon/2 \geq 0$ on $\Omega_{n,\epsilon}$.
	
	We define $f(x) = \exp(-\gamma \abs{x - \mathbf{x}_n}) + \exp(-\gamma \abs{x + \mathbf{x}_n})$ for $\gamma < \abs{\mu}^{1/2}$. A computation gives
	\[
	\left(-\Delta + \left(\tau_n^\ell + \tau_n^r\right)V^\mathrm{MF} - \mu - \epsilon/2\right)f(x) \geq f(x) \left(-\gamma^2 + \lvert \mu \rvert - \epsilon\right) \, ,
	\]
	on $\Omega_{n,\epsilon}$. If we choose $\gamma = \sqrt{\abs{\mu} - \epsilon}$ then $f$ is a supersolution for $\tilde{h}_n$. Similarly, we can show that $f$ is a supersolution for $h_n$ in the region $\{\abs{x} \geq L_n/2 + R_\epsilon\}$ and the first comparison Lemma \ref{lemma: comparison_supersolution} shows that $u_n^+$ tends to 0 at infinity for all $n$. From Proposition \ref{prop: stronger_convergence}, there exists $C_\epsilon >0$ (independent from $n$), such that $u_n^+ \leq C_\epsilon f$ on $\R^d \setminus \Omega_{n,\epsilon}$. Now, we can apply the second comparison Lemma \ref{lemma: comparison_hoffmann_ostenhoff} which shows the upper bound in \eqref{eq: a priori exp decay_u^+_n}. The lower bound can be shown similarly.
	
	The proof of the pointwise estimate on $\nabla u_n^+$ is similar to those of proof of Proposition \ref{prop: u_exponential_falloff}, so we will only sketch the arguments. We introduce the Yukawa potential $\tilde{Y}_n$ defined as the solution in the sense of distribution of the equation
	\[
	(-\Delta - \mu_n^+)\tilde{Y}_n = \delta_0 \, ,
	\]
	where $\delta_0$ is the Dirac distribution on 0. Then, we write
	\[
	u_n^+ = - \tilde{Y}_n \ast (V_n^\mathrm{MF} u_n^+) \pt
	\]
	Using the pointwise estimate \eqref{eq: a priori exp decay_u^+_n} on $u_n^+$ and the asymptotic $\abs{\nabla \tilde{Y}_n(x)} = \grandO{e^{-\sqrt{\abs{\mu} - \epsilon}\abs{x}}}$ (see \cite[Theorem 6.23]{lieb2001analysis}), this implies 
	\[
	\abs{\nabla u_n^+} \lesssim g \ast \left[\left(1 + \frac{1}{\abs{x-\mathbf{x}_n}} + \frac{1}{\abs{x+\mathbf{x}_n}}\right)\left(\tau_n^\ell g + \tau_n^r g\right)\right] \, ,
	\]
	where $g(x) = \exp(-\sqrt{\abs{\mu} - \epsilon}\abs{x})$.
	Then, it remains to invoke Lemma \ref{lemma: lemma_technical_exponential_convolution} to show \eqref{eq: a priori exp decay_u^+_n_nabla}. We can absord any remaining polynomial terms by slightly modifying the constants. Finally, the integral estimate \eqref{eq: a priori exp decay_u^+_n_integral} is shown exactly as estimate \eqref{eq: a priori exp decay3} from Proposition \ref{prop: u_exponential_falloff}.
\end{proof}

\section{Precising the rate of convergence}

In the previous section, we have shown the local convergence of $u_n^+$ to $u$ in some Sobolev spaces. In this section, we show the existence of a first excited state $u_n^-$ for $h_n$ and give several rates of convergence for $u_n^\pm$ and $\mu_n^\pm$.

\subsection{Stability of the monoatomic model in \texorpdfstring{$H^1(\R^d)$}{H1(Rd)}-norm}
Let $d\in \{2,3\}$.
The next proposition is a stability result in $H^1(\R^d)-$norm on the monoatomic model which allows us to convert energy estimates into $H^1(\R^d)-$norm estimates. The proof follows the arguments given by Carlen, Frank and Lieb in \cite[Theorem 5.1]{carlen2014stability}.
In the sequel, we denote by 
\[
q_h : v \in H^1(\R^d) \longmapsto \int_{\R^d} \abs{\nabla v(x)}^2 \diff x + \int_{\R^d} V^\mathrm{MF}(x)\abs{v(x)}^2  \diff x \, ,
\]
the quadratic form associated with the mean-field operator $h$. 
\begin{prop}[Stability of the monoatomic model \eqref{eq: min_problem_ow}]
	\label{prop: stability}
	There exists $C>0$ such that for all $v \in H^1(\R^d)$ such that $\norme{v}_{L^2(\R^d)} = 1$, we have
	\begin{gather}
	\label{eq: stability_operator}
	q_h(v) \geq \mu +C \min_{\theta \in \intff{0}{2\pi}} \norme{e^{i\theta} v - u}_{H^1(\R^d)}^2 \et 
	\mathcal{E}(v) \geq \mathcal{E}(u) + C \min_{\theta \in \intff{0}{2\pi}} \norme{e^{i\theta} v - u}_{H^1(\R^d)}^2 \pt
	\end{gather}
\end{prop}
First, we show that, up to a additive constant, the quadratic form associated with $h$ is equivalent to the $H^1(\R^d)-$norm.
\begin{lemma}[Properties of $v \mapsto q_h(v)$]
	\label{lemma: properties_q_h}
	For all $\epsilon \in \intoo{0}{1}$, there exists constants $C_\epsilon,C_\epsilon'>0$ such that for all $v \in H^1(\R^d)$, we have
	\begin{align}
	\label{eq: q_h_equivalent}
	(1-\epsilon)\norme{v}^2_{H^1(\R^d)} \leq q_h(v) + C_\epsilon \norme{v}^2_{L^2(\R^d)} \leq C_\epsilon' \norme{v}^2_{H^1(\R^d)} \pt
	\end{align}
\end{lemma}
\begin{proof}
	By \cite[Chap. 11]{lieb2001analysis}, for any $\epsilon >0$, there exists $C_\epsilon>0$ such that
	\[
	\abs{\int_{\R^d} \frac{\abs{v(x)}^2}{\abs{x}} \diff x} \leq \epsilon \norme{\nabla v}^2_{L^2(\R^d)} + C_\epsilon \norme{v}^2_{L^2(\R^d)} \, ,
	\]
	for all $v \in H^1(\R^d)$.
	Recall that $\abs{u}^2 \ast |\cdot|^{-1} \in L^\infty(\R^d)$ by Lemma \ref{lemma: technical2}, then we have
	\[
	0 \leq \int_{\R^d} \left(\abs{u}^2 \ast |\cdot|^{-1}\right) \abs{v}^2 \leq \norme{\abs{u}^2 \ast |\cdot|^{-1}}_{L^\infty(\R^d)} \norme{v}^2_{L^2(\R^d)} \pt
	\]
	Estimate \eqref{eq: q_h_equivalent} follows from these two estimates.
\end{proof}
The next lemma is a local version of Proposition \ref{prop: stability}.
\begin{lemma}
	\label{lemma: stability_local}
	There exists $\delta,C>0$ such that for all $v \in H^1(\R^d)$ with $\norme{v}_{L^2(\R^d)} = 1$, we have
	\[
	\min_{\theta \in \intff{0}{2\pi}} \norme{e^{i\theta} v - u}_{L^2(\R^d)} \leq \delta \implies
	q_h(v) \geq \mu +C \min_{\theta \in \intff{0}{2\pi}} \norme{e^{i\theta} v - u}^2_{H^1(\R^d)} \pt
	\]
\end{lemma}
\begin{proof}
	Let $v \in H^1(\R^d)$ such that $\norme{v}_{L^2(\R^d)}=1$. Let $\theta \in \intff{0}{2\pi}$ such that
	\[
	\norme{e^{i\theta} v - u}_{L^2(\R^d)} =  \min_{\tilde{\theta} \in \intff{0}{2\pi}} \norme{e^{i\tilde{\theta}} v - u}_{L^2(\R^d)}\pt
	\]
	We have $q_h(v) = \mu + q_{h - \mu}(e^{i\theta}v)$.	We denote $\tilde{v} = e^{i\theta}v$ and we have $\pdtsc{\tilde{v}}{u}_{L^2(\R^d)} \in \R_+$.
	Let $w  = \tilde{v} - \pdtsc{u}{\tilde{v}}_{L^2(\R^d)}u$ be the orthogonal projection in $L^2(\R^d)$ of $\tilde{v}$ on the subspace $(\spn u)^\perp$. An integration by parts shows $q_{h-\mu}(\tilde{v}) = q_{h-\mu}(w)$. Let $\epsilon \in \intoo{0}{1}$. By Lemma \ref{lemma: properties_q_h}, there exists a constant $C_\epsilon\geq \abs{\mu}$ such that $q_h + C_\epsilon \geq \epsilon \norme{\cdot}_{H^1(\R^d)}$. By the functional calculus, we also have $q_{h-\mu}(\tilde{v}) \geq G \norme{w}^2_{L^2(\R^d)}$ where $G>0$ denotes the spectral gap between $\mu$ and the remaining spectrum of $h$. These three previous estimates imply
	\begin{align}
	\label{eq: 02_04_2020}
	q_{h - \mu}(\tilde{v}) \geq c_\epsilon \left(q_h(w) + C_\epsilon\norme{w}_{L^2(\R^d)}^2\right) \geq c_\epsilon \epsilon \norme{w}_{H^1(\R^d)}^2 \, ,
	\end{align}
	with $c_\epsilon = \frac{G}{G+\mu+C_\epsilon}$. Moreover, we have
	\begin{align*}
	\norme{w}^2_{H^1(\R^d)} 
	& = \norme{\tilde{v}-u}^2_{H^1(\R^d)} + \abs{1 - \pdtsc{u}{\tilde{v}}_{L^2(\R^d)}}^2 \norme{u}_{H^1(\R^d)}^2 + 2 \Re\left(\pdtsc{\tilde{v}-u}{\tilde{v}}_{L^2(\R^d)}\pdtsc{\tilde{v} -u}{u}_{H^1(\R^d)}\right) \\
	& \geq\norme{\tilde{v}-u}^2_{H^1(\R^d)} +2 \Re\left(\pdtsc{\tilde{v}-u}{\tilde{v}}_{L^2(\R^d)}\pdtsc{\tilde{v} -u}{u}_{H^1(\R^d)}\right) \, ,
	\end{align*}
	where $\pdtsc{\cdot}{\cdot}_{H^1(\R^d)}$ denotes the usual scalar product in $H^1(\R^d)$.
	Using the Cauchy-Schwarz inequality and the identity (recall that $\pdtsc{u}{\tilde{v}}_{L^2(\R^d)} \in \R_+$)
	\[
	\pdtsc{\tilde{v}-u}{\tilde{v}}_{L^2(\R^d)} = \frac{1}{2}\pdtsc{\tilde{v}-u}{\tilde{v}-u}_{L^2(\R^d)}\, ,
	\]
	we obtain the bound
	\begin{align*}
	\abs{2\Re\left(\pdtsc{\tilde{v}-u}{\tilde{v}}_{L^2(\R^d)}\pdtsc{\tilde{v} -u}{u}_{H^1(\R^d)}\right)} 
	&\leq \norme{u}_{H^1(\R^d)} \norme{\tilde{v}-u}_{H^1(\R^d)} \norme{\tilde{v}-u}_{L^2(\R^d)}^2 \\
	&\leq \norme{u}_{H^1(\R^d)} \norme{\tilde{v}-u}_{H^1(\R^d)}^2 \norme{\tilde{v}-u}_{L^2(\R^d)} \pt 
	\end{align*}
	If $\norme{\tilde{v}-u}_{L^2(\R^d)} \leq (1-\beta) \norme{u}_{H^1(\R^d)}^{-1} \eqqcolon \delta $ for some $\beta \in \intoo{0}{1}$, we get
	\[
	q_h(v) \geq \mu + c_\epsilon\beta\epsilon \norme{\tilde{v} - u}_{H^1(\R^d)}^2 \geq \mu + c_\epsilon\beta \epsilon\min_{\theta \in \intff{0}{2\pi}} \norme{e^{i\theta} v - u}_{H^1(\R^d)}^2\, ,
	\]
	which concludes the proof of Lemma \ref{lemma: stability_local}.
\end{proof}
\begin{proof}[Proof of Proposition \ref{prop: stability}]
	First, using the positive-definiteness of $D(\cdot, \cdot)$, we find
	\[
	\mathcal{E}(v) = \mathcal{E}(u) + q_h(v) - \mu \norme{v}_{L^2(\R^2)}^2 + D(\abs{u}^2 - \abs{v}^2,\abs{u}^2 - \abs{v}^2) \geq \mathcal{E}(u) + q_h(v) - \mu \norme{v}_{L^2(\R^2)}^2 \pt
	\]	
	Then the right side of \eqref{eq: stability_operator} is a consequence of its left side.	
	To show the left part of \eqref{eq: stability_operator}, we argue by contradiction and assume that there exists a sequence $(\delta_n)_{n\in \N}$ of positive numbers and a sequence $(v_n)_{n\in\N}$ of functions in $H^1(\R^d)$ such that
	\begin{gather*}
	\delta_n \limit{n\to\infty} 0\, , \quad \norme{v_n}_{L^2(\R^d)}=1 \et
	q_h(v_n) \leq \mu + \delta_n \min_{\theta \in \intff{0}{2\pi}} \norme{v_n - u}^2_{H^1(\R^d)} \pt
	\end{gather*}
	First we show that $\min_{\theta \in \intff{0}{2\pi}} \norme{v_n - u}^2_{H^1(\R^d)}$ is bounded uniformly in $n$.
	For $n$ large enough, we have
	\[
	q_h(v_n) \leq \mu + \frac{1}{2} \min_{\theta \in \intff{0}{2\pi}} \norme{e^{i\theta}v_n - u}^2_{H^1(\R^d)} \pt
	\]
	Moreover, by the Cauchy-Schwarz inequality, we have
	\begin{align*}
	\min_{\theta \in \intff{0}{2\pi}} \norme{e^{i\theta}v_n - u}^2_{H^1(\R^d)} 
	&\leq \norme{u}^2_{H^1(\R^d)} + \norme{v_n}^2_{H^1(\R^d)} + 2\norme{u}_{H^1(\R^d)}\norme{v_n}_{H^1(\R^d)} \\
	&\leq 3\norme{u}^2_{H^1(\R^d)} + \frac{3}{2}\norme{v_n}^2_{H^1(\R^d)} \pt
	\end{align*}
	By Lemma \ref{lemma: properties_q_h}, there exists a constant $C'>0$ such that $\pdtsc{v_n}{h v_n}_{L^2(\R^d)} \geq \frac{7}{8} \norme{v_n}^2_{H^1(\R^d)} - C'$. Then, we have shown that there exists $C'' >0$ such that $\frac{1}{8} \norme{v_n}^2_{H^1(\R^d)} \leq C''$	for all $n\in\N$. As a consequence, we have $\lim\limits_{n\to\infty} q_h(v_n) = \mu$ and, by the functional calculus and the fact $\mu$ is non degenerate, we deduce 
	\begin{align}
	\label{eq: croatian_reggea}
	\min_{\theta \in \intff{0}{2\pi}}\norme{e^{i\theta}v_n - u}_{L^2(\R^d)} = 2\left(1 - \abs{\pdtsc{v_n}{u}_{L^2(\R^d)}}\right) \limit{n \to \infty} 0 \pt
	\end{align}
	Let $\epsilon \in \intoo{0}{1}$ and  $C_\epsilon \geq \abs{\mu}$ as in Lemma \ref{lemma: properties_q_h}. An integration by parts shows
	\begin{align*}
	(1-\epsilon)\min_{\theta \in \intff{0}{2\pi}} \norme{e^{i\theta}v_n - u}_{H^1(\R^d)}^2 
	&\leq q_{h+C_\epsilon}(v_n) + q_{h+C_\epsilon}(u) - 2(\mu+C_\epsilon) \abs{\pdtsc{v_n}{u}_{L^2(\R^d)}} \pt
	\end{align*}
	Taking the limit, we obtain
	\[
	\lim\limits_{n\to \infty}\min_{\theta \in \intff{0}{2\pi}} \norme{e^{i\theta}v_n - u}_{H^1(\R^d)}^2 \leq \frac{2\left(\mu + C_\epsilon\right)}{1-\epsilon}\lim\limits_{n\to \infty}\left(1 - \abs{\pdtsc{v_n}{u}_{L^2(\R^d)}}\right)  = 0 \, ,
	\]
	which is a contradiction by Lemma \ref{lemma: stability_local}. This concludes the proof of Proposition \ref{prop: stability}.	
\end{proof}

\subsection{Rate of convergence for \texorpdfstring{$u_n^+$}{un+}}
\begin{prop}[Rate of convergence for $u_n^+$]
	\label{prop: rate_convergence}
	For any $\epsilon \in \intoo{0}{3}$, there exists a constant $C_\epsilon>0$ such that the following estimates hold
	\begin{gather}
	\norme{\abs{u(\cdot + \mathbf{x}_n)}^2 - \abs{\chi_n^\ell u_n^+}^2}_{L^1(\R^d)} = \norme{\abs{u(\cdot - \mathbf{x}_n)}^2 - \abs{\chi_n^r u_n^+}^2}_{L^1(\R^d)} \leq
	\begin{cases}
	C_\epsilon L_n^{-3 + \epsilon} & \text{if }d=2 \, , \\
	\grandO{L_n^{-\infty}} & \text{if }d=3 \, ,
	\end{cases} \\
	\label{eq: first a priori estimate}
	\norme{u(\cdot + \mathbf{x}_n) - \chi_n^\ell u_n^+}_{H^1(\R^d)} = \norme{u(\cdot - \mathbf{x}_n) - \chi_n^ru_n^+}_{H^1(\R^d)}  \leq
	\begin{cases}
	C_\epsilon L_n^{-3 + \epsilon} & \text{if }d=2 \, , \\
	\grandO{L_n^{-\infty}} & \text{if }d=3 \pt
	\end{cases}
	\end{gather}
\end{prop}
\begin{proof}
	We denote $v_n \coloneqq \tau_n^r \chi_n^\ell u_n^+$ (the right translation operator $\tau_n^r$ and the function are defined in the beginning of Section \ref{sec:construction-of-quasi-modes}). Notice that most of the mass of $v_n$ is localized in the vicinity of the origin. Using the reflection symmetry $\mathcal{R}$, we have $\norme{v_n}_{L^2(\R^d)} = 1$. 
	Recall the notation $V_n^\ell$, $V_n^r$, $u_n^\ell$ and $u_n^r$ defined in \eqref{eq: definition_translated_potential}.
	As in the proof of Proposition \ref{prop: strong_convergence_H1}, we look for a lower bound on $\mathcal{E}_n(u_n^+)$ using the IMS formula but without discarding the positive term $D(\abs{\chi_n^\ell u_n^+}^2,\abs{\chi_n^r u_n^+}^2)$. We find
	\begin{align*}
	\mathcal{E}_n(u_n^+)
	= &~2 \mathcal{E}(v_n) +2 \left(\int_{\R^d} V_n^r \abs{\chi^\ell_n u_n^+}^2 + D(\abs{\chi^\ell_n u_n^+}^2,\abs{\chi^r_nu_n^+}^2)  + \frac{1}{2L_n} - \int_{\R^d} \abs{\nabla \chi^\ell_n}^2 \abs{u_n^+}^2\right)  \\
	\geq&~ 2\mathcal{E}(u) + 2 \left(\int_{\R^d} V_n^r \abs{\chi^\ell_n u_n^+}^2 + D(\abs{\chi^\ell_n u_n^+}^2,\abs{\chi^r_nu_n^+}^2)  + \frac{1}{2L_n} - \int_{\R^d} \abs{\nabla \chi^\ell_n}^2 \abs{u_n^+}^2\right) \\
	&+ C \norme{v_n - u}^2_{H^1(\R^d)} \, ,
	\end{align*}
	for some constant $C>0$. To get the lower bound, we have used Proposition \ref{prop: stability} and the positiveness of $v_n$ and $u$. As $\supp \nabla \chi_n^\ell \subset \{\abs{x_1} \leq \sqrt{L_n}\delta \}$, by the exponential decay of $u_n^+$ proved in Proposition~\ref{cor: exponential decay u^+_n}, the term involving $\nabla \chi_n^\ell$ is a $\grandO{L_n^{-\infty}}$.
	We remark that we can express the potential terms as 
	\begin{gather*}
	\int_{\R^d} V_n^r \abs{\chi^\ell_n u_n^+}^2 = - \left(\lvert v_n \rvert^2 \ast |\cdot|^{-1}\right)(2\mathbf{x}_n) \et \\ 
	D(\abs{\chi^\ell_n u_n^+}^2,\abs{\chi^r_nu_n^+}^2) = \frac{1}{2} \left(\lvert v_n \rvert^2 \ast  \mathcal{S}[\lvert v_n\rvert^2 ]  \ast |\cdot|^{-1}\right)(2\mathbf{x}_n) \, ,
	\end{gather*}
	where the transformation $\mathcal{S}$ is defined by $\mathcal{S}[v](x) = v (-Rx)$ with $R$ the reflection matrix with respect to the first coordinate defined in \eqref{eq: def_reflection}. Notice that $S\mathbf{x}_n = \mathbf{x}_n$ and $\mathcal{S}[u] = u$. Hence, we have obtained the lower bound
	\[
	\mathcal{E}_n(u_n^+) \geq 2 \left(\mathcal{E}(u) + q_n(v_n) \right) + C\norme{v_n - u}^2_{H^1(\R^d)} + \grandO{L_n^{-\infty}} \, ,
	\]
	where $q_n$ is a quartic function plus a constant defined by
	\[
	q_n (v) \coloneqq \left(\left( \frac{\abs{v}^2 \ast \mathcal{S}[\abs{v}^2]  }{2} - \abs{v}^2\right) \ast |\cdot|^{-1}\right)(2\mathbf{x}_n) + \frac{1}{2L_n} \pt
	\]
	Let $(R_n)_{n\in\N}$ be an increasing sequence such that $R_n \leq L_n/4$ and $L_n^{\gamma} = \petito{R_n}$ for some $\gamma \in \intoo{0}{1}$.
	In the following, we will choose the constants in the $O$'s independently of the choice of the sequence $R_n$. Let $\chi \in \mathcal{C}^\infty(\R^d)$ be a cutoff function such that
	\begin{align*}
	0 \leq \chi \leq 1\, ,\quad \chi \equiv 1 \quad \text{on} \quad B(0,1/2) \et \supp \chi \subset B(0,1) \pt
	\end{align*}
	Let $\chi_n(x) = \chi(x/R_n)$. We define $v_{n,1} \coloneqq \chi_n (\lvert v_n \rvert^2 - \abs{u}^2)$ and $v_{n,2} \coloneqq (1-\chi_n) (\lvert v_n \rvert^2 - \abs{u}^2)$. Hence, we have the decomposition 
	\begin{align}
	\label{eq: decomposition_cool_and_nice}
	\lvert v_n \rvert^2 = \abs{u}^2 + v_{n,1} + v_{n,2} \pt
	\end{align}
	where $\supp v_{n,1} \subset B(0,R_n)$ and $\supp v_{n,2} \subset B(0,\frac{R_n}{2})^c$.
	In particular, using the fact $u$ and $v_n$ are exponentially decaying away from the origin (see Proposition \ref{prop: u_exponential_falloff} and Proposition \ref{cor: exponential decay u^+_n}), we have $\lVert v_{n,2} \rVert_{W^{2-\alpha,1}(\R^d)} = \grandO{L_n^{-\infty}}$ where $\alpha \in \intoo{0}{1}$ if $d=2$ and $\alpha = 0$ if $d=3$. Moreover, from Proposition \ref{prop: strong_convergence_H1} and Corollary \ref{cor: boundedness_u_better}, there exists $\nu \geq 0$ such that $\lVert v_{n,1}\rVert_{L^1(\R^d)} = \grandO{L_n^{-\nu}}$. We expand $q(v_n)$ with respect to the decomposition \eqref{eq: decomposition_cool_and_nice} and we can show, by means of Sobolev embeddings and Young's inequality, that all the terms involving $v_{n,2}$ are $\grandO{L^{-\infty}_n}$. Then, we have
	\begin{align}
	\label{eq: 11010}
	q_n( v_n ) 
	=& ~q_n(u) + \left[\left(\abs{u}^2 \ast \left(\frac{v_{n,1} + \mathcal{S}[v_{n,1}]}{2}\right) - v_{n,1}\right) \ast |\cdot|^{-1}\right](2\mathbf{x}_n) \\
	&\label{eq: 11011}+ \frac{1}{2} \left[v_{n,1} \ast \mathcal{S}[v_{n,1}] \ast |\cdot|^{-1}\right](2\mathbf{x}_n) + \grandO{L_n^{-\infty}}\pt
	\end{align}
	Assume $d=2$. From the expansion \eqref{potential_expansion4}, we have $(\abs{u}^2 \ast |\cdot|^{-1})(x) = |x|^{-1} + \grandO{|x|^{-3}}$ as $\abs{x} \to \infty$. Recall that $v_{n,1}$ is compactly supported in the ball $B(0,R_n) \subset B(0,L_n/4)$. Then, we use 
	\[
	\left(\mathcal{S}[v_{n,1}]\ast |\cdot|^{-1}\right)(2\mathbf{x}_n) = \left(v_{n,1}\ast |\cdot|^{-1}\right)(2\mathbf{x}_n) \, , 
	\]
	to get
	\begin{align*}
	&\left[\left(\abs{u}^2 \ast \left(\frac{v_{n,1} + \mathcal{S}[v_{n,1}]}{2}\right) - v_{n,1}\right) \ast |\cdot|^{-1}\right](2\mathbf{x}_n) = \grandO{\left(v_{n,1} \ast |\cdot|^{-3}\right)(2\mathbf{x}_n)} \pt
	\end{align*}
	As $\supp v_{n,1} \subset B(0,R_n) \subset B(0,L_n/4)$, the quantity $\left(v_{n,1} \ast |\cdot|^{-3}\right)(2\mathbf{x}_n)$ is well defined (recall that $\abs{2\mathbf{x}_n} = L_n$) and we can estimate it by $\grandO{\lVert v_{n,1}\rVert_{L^1(\R^2)} L_n^{-3}} = \grandO{L_n^{-3-\nu}}$. If $d=3$, we use Newton's theorem instead of expansion \eqref{potential_expansion4} and we can estimate by $\grandO{L_n^{-\nu} T_n} = \grandO{L_n^{-\infty}}$.
	For the remaining term of \eqref{eq: 11011}, we remark that 
	\[
	\supp \left(v_{n,1} \ast \mathcal{S}[v_{n,1}]\right) \subset B(0,2R_n) \subset B(0,L_n/2)\pt
	\]
	Irrespective of the dimension, we have
	\begin{align}
	\label{eq: 1101}
	\abs{\left[v_{n,1} \ast \mathcal{S}[v_{n,1}] \ast |\cdot|^{-1}\right](2\mathbf{x}_n)} = \grandO{ L_n^{-1} \norme{v_{n,1} \ast \mathcal{S}[v_{n,1}]}_{L^1(\R^d)}} = \grandO{L_n^{-1-2\nu}} \pt
	\end{align}
	In the last inequality, we have used Young's inequality. Hence, we end up with
	\begin{align}
	\label{eq: lower_bound_28_12}
	\mathcal{E}_n(u_n^+) \geq 2\left(\mathcal{E}(u) + q_n(u)\right) + C\norme{v_n - u}^2_{H^1(\R^d)} + \grandO{L_n^{-2\gamma(\nu)}} \, ,
	\end{align}
	where
	\begin{align*}
	\gamma(\nu) \coloneqq
	\begin{cases}
	\min\left(\frac{1}{2} + \nu, \frac{3}{2} + \frac{\nu}{2}\right) & \text{if }d=2 \, ,\\
	\frac{1}{2} + \nu & \text{if }d=3 \pt
	\end{cases}
	\end{align*}
	Now, following the proof of Proposition \ref{prop: strong_convergence_H1}, we can obtain the upper bound
	\begin{align}
	\label{eq: upper_bound_28_12}
	\mathcal{E}_n(u_n^+) \leq 2 \left(\mathcal{E}(u) + q_n(u)\right) + \grandO{L_n^{-\infty}} \pt
	\end{align}
	From \eqref{eq: lower_bound_28_12} and \eqref{eq: upper_bound_28_12}, we obtain
	\begin{align}
	\label{eq: 28_12}
	\norme{v_n - u}_{H^1(\R^d)} = \grandO{L_n^{-\gamma(\nu)}} \pt
	\end{align}
	Recall $\supp v_{n,1} \subset B(0,R_n)$ and $\sup_n \norme{v_n}_{L^\infty(\R^d)} <\infty$ (see Proposition \ref{prop: stronger_convergence}). By the Cauchy-Schwarz inequality, we get
	\[
	\norme{v_{n,1}}_{L^1(\R^d)} = \grandO{R_n L_n^{-\gamma(\nu)}} \pt
	\]
	Recall that we can choose $R_n$ of the form $L_n^\gamma/4$ for any $\gamma \in \intoo{0}{1}$ small enough. Then, an induction argument shows that 
	\begin{align*}
	\norme{v_{n,1}}_{L^1(\R^d)} = \begin{cases}
	\grandO{L_n^{-3 + \epsilon}} & \text{if }d=2 \, ,\\
	\grandO{L_n^{-\infty}} & \text{if }d=3 \, ,
	\end{cases} 		
	\end{align*}
	for any $\epsilon >0$ (notice that $3$ is the unique fixed point of the map $\nu \mapsto \gamma(\nu)$ when $d=2$). Recalling estimate~\eqref{eq: 28_12} with $\nu = 3 - \epsilon$ if $d=2$ and $\nu = k$ for any $k\in\N$ if $d=3$ and this concludes the proof of Proposition \ref{prop: rate_convergence}.
\end{proof}

\begin{rem}
	In the three-dimensional case, we are only able to obtain superpolynomial convergence rates but we think they should be exponential.
\end{rem}

\subsection{Rate of convergence for mean-field potentials}
We recall that the mean-field hamiltonians associated with monoatomic and diatomic model are defined by
\begin{gather*}
V^\mathrm{MF} = - \frac{1}{|\cdot|} + \abs{u}^2 \ast |\cdot|^{-1} \et
V_n^\mathrm{MF} = - \left(\frac{1}{\abs{\cdot - \mathbf{x}_n}}+\frac{1}{\abs{\cdot + \mathbf{x}_n}}\right) + \lvert u_n^+ \rvert^2 \ast |\cdot|^{-1} \pt
\end{gather*}
Recall that $\tau_n^r$ (resp. $\tau_n^\ell$) denotes the translation operator by $\mathbf{x}_n$ (resp. $-\mathbf{x}_n$).
\begin{prop}
	\label{lemma: lemme_utile_mais_relou}
	Let $w$ such that for all $x\in\R^d$, we have $\abs{w(x)} \leq c \exp(-\alpha\abs{x})$ where $c,\alpha >0$. Then, there exists a constant $C(c,\alpha,\epsilon)$ such that we have the estimates
	\begin{align*}
	\abs{\int_{\R^d} w \left(\tau_{n}^{r} V_n^\mathrm{MF} - V^\mathrm{MF}\right)}\leq
	\begin{cases}
	C(c,\alpha,\epsilon) L_n^{-3+\epsilon}\left( \norme{w}_{L^1(\R^d)} + \norme{w}_{L^2(\R^d)}\right) + \grandO{L_n^{-\infty}} & \text{if }d=2 \, ,\\
	\grandO{L_n^{-\infty}}\left(1+ \norme{w}_{L^1(\R^d)} + \norme{w}_{L^2(\R^d)}\right) & \text{if }d=3 \pt
	\end{cases} 
	\end{align*}
	for any $\epsilon >0$. The same estimate holds by replacing $\tau_n^r$ with $\tau_n^\ell$.
\end{prop}
\begin{proof}
	Let $r \in \intoo{d}{\infty}$. As in the proof of Proposition \ref{prop: rate_convergence}, we use the notation $v_n = \tau_n^r \chi_n^\ell u_n^+$. Using Lemma \ref{lemma: technical2}, the facts (see Corollary \ref{cor: boundedness_u_better} and Proposition \ref{prop: study of d(h)})
	\[
	\sup_n \norme{v_n}_{L^\infty(\R^d)} < \infty \et u \in L^\infty(\R^d) \, ,
	\]
	and Proposition \ref{prop: rate_convergence}, we obtain the following estimate
	\begin{align}
	\label{eq: preliminary_estimate}
	\norme{\left(\lvert v_n \rvert^2 - \abs{u}^2\right)\ast |\cdot|^{-1}}_{L^r(\R^d)} =
	\begin{cases}
	\grandO{L_n^{-3+\epsilon}} & \text{if }d=2\, , \\
	\grandO{L_n^{-\infty}} & \text{if }d=3\pt
	\end{cases} 
	\end{align}
	Let $w \in L^1(\R^d) \cap L^2(\R^d)$ such that $\abs{w(x)} \leq c \exp(-\alpha \abs{x})$. Let $r' \in \intoo{1}{\frac{d}{d-1}}$ such that $\frac{1}{r} + \frac{1}{r'} =1$. Then, using Hölder's inequality, we have
	\begin{align*}
	\abs{\int_{\R^d} w \left(\tau_{\mathbf{x}_n} V_n^\mathrm{MF} - V^\mathrm{MF}\right)}
	\leq & ~\abs{\int_{\R^d} w \left(\lvert v_n \rvert^2 - \abs{u}^2\right)\ast |\cdot|^{-1}} + \abs{\left[w \ast \left(\mathcal{R}[\lvert v_n \rvert^2] - \abs{u}^2\right) \ast |\cdot|^{-1} \right](2\mathbf{x}_n)}\\
	& + \abs{\left[w\ast \left(\abs{u}^2 \ast |\cdot|^{-1} - |\cdot|^{-1}\right)\right](2\mathbf{x}_n)}\\
	\leq & ~ \norme{w}_{L^{r'}(\R^d)} \norme{\left(\lvert v_n \rvert^2 - \abs{u}^2\right)\ast |\cdot|^{-1}}_{L^r(\R^d)} \\
	&+ \abs{\left[w\ast \left(\abs{u}^2 \ast |\cdot|^{-1} - |\cdot|^{-1}\right)\right](2\mathbf{x}_n)} \pt
	\end{align*}
	By \eqref{eq: preliminary_estimate} and Hölder's inequality, the first term is bounded by
	\begin{align*}
	\norme{w}_{L^{r'}(\R^d)} \norme{\left(\lvert v_n \rvert^2 - \abs{u}^2\right)\ast |\cdot|^{-1}}_{L^r(\R^d)}
	=
	\begin{cases}
	\grandO{L_n^{-3+\epsilon}} \left(\norme{w}_{L^1(\R^d)}+\norme{w}_{L^2(\R^d)}\right) & \text{if }d=2 \, , \\
	\grandO{L_n^{-\infty}} \left(\norme{w}_{L^1(\R^d)}+\norme{w}_{L^2(\R^d)}\right) & \text{if }d=3 \, ,
	\end{cases}
	\end{align*}
	where the constants appearing in the $O$'s do not depend on $w$.
	It remains to estimate the second term. Assume $d=2$. From Lemma \ref{lemma: potential_expansion}, we can write $\abs{u}^2 \ast |\cdot|^{-1} - |\cdot|^{-1} = \grandO{|\cdot|^{-3}}$ as $\abs{x} \to \infty$.	We use the decomposition $w = w_{n,1} + w_{n,2} + w_{n,3}$ where $w_{n,1} = \mathds{1}_{B(0,L_n/2)} w$, $w_{n,2} = \mathds{1}_{B(2\mathbf{x}_n,L_n/2)} w$ and $w_{n,3} = w - w_{n,1} - w_{n,2}$. Because $w$ is exponentially decreasing, we have $\norme{w_{n,2}}_{L^\infty(\R^2)} = \grandO{L_n^{-\infty}}$ and $\norme{w_{n,3}}_{L^1(\R^2)} = \grandO{L_n^{-\infty}}$ where the constants only depend on $c$ and $\alpha$.
	Then, by Young's inequality, we have
	\begin{align*}
	\abs{\left[ w \ast \left(\abs{u}^2 \ast |\cdot|^{-1} - |\cdot|^{-1}\right)\right](2\mathbf{x}_n)} 
	\lesssim &~\abs{\left(w_{n,1} \ast |\cdot|^{-3}\right)(2\mathbf{x}_n)} \\
	&+\norme{w_{n,2}}_{L^\infty(\R^2)} \int_{B(0,L_n/2)} \frac{\diff x}{\abs{x}} + \norme{w_{n,3}}_{L^1(\R^2)} \\
	\lesssim&~ L_n^{-3}\norme{w}_{L^1(\R^2)} + L_n^{-\infty} \pt
	\end{align*}
	When $d=3$, we replace the estimate $\abs{u}^2 \ast |\cdot|^{-1} - |\cdot|^{-1} = \grandO{|\cdot|^{-3}}$ by Newton's theorem and similar arguments show
	\[
	\abs{\left[ w \ast \left(\abs{u}^2 \ast |\cdot|^{-1} - |\cdot|^{-1}\right)\right](2\mathbf{x}_n)}  = \grandO{L_n^{-\infty} \left(1+\norme{w}_{L^1(\R^3)}\right)} \pt
	\]
	This concludes the proof of Proposition \ref{lemma: lemme_utile_mais_relou}.
\end{proof}

\subsection{Rate of convergence for the Lagrange multipliers}
The next proposition shows that $\mu_n^-$ exists for $n$ large enough and gives estimates on the convergence of the Lagrange multipliers $\mu_n^\pm$ toward $\mu$.
\begin{prop}[Rate of convergence of Lagrange multipliers]
	\label{prop: rate_convergence_ground_state_energy}
	The diatomic mean-field hamiltonian $h_n$ admits an excited state. Moreover, for any $\epsilon \in \intoo{0}{3}$, there exists $C_\epsilon>0$ such that
	\begin{gather}
	\label{eq: rate_convergence_ground_state_energy}
	\abs{\mu_n^\pm - \mu} =
	\begin{cases}
	\grandO{L_n^{-3+\epsilon}}  & \text{if }d=2 \, ,\\
	\grandO{L_n^{-\infty}} & \text{if }d=3 \, ,
	\end{cases}
	\end{gather}
	for any $\epsilon >0$.
\end{prop}
\begin{proof}
	Firstly, we estimate $\abs{\mu_n^+ - \mu}$.
	Using the formulas (recall $\int_{\R^d} \lvert u_n^+ \rvert ^2=2$)
	\begin{gather*}
	\mu = \mathcal{E}(u) + D(\abs{u}^2,\abs{u}^2) \et \mu_n^+ = \frac{1}{2} \Big(\mathcal{E}_n(u_n^+) + D(\lvert u_n^+ \rvert^2,\lvert u_n^+ \rvert^2)\Big) \, ,
	\end{gather*}
	we have
	\begin{align}
	\notag
	\abs{\mu_n^+ - \mu} 
	=&~ \frac{1}{2} \abs{\mathcal{E}_n(u_n^+) - 2\mathcal{E}(u)} + \abs{D(\abs{v_n}^2,\abs{v_n}^2) - D(\abs{u}^2,\abs{u}^2)} \\
	\label{eq: 29_12}
	&+ \frac{1}{2} \abs{D\left(\abs{\chi_n^\ell u_n^+}^2,\abs{\chi_n^r u_n^+}^2\right) - \frac{1}{2L_n}} \, ,
	\end{align}
	where we have denoted $v_n = \tau_{\mathbf{x}_n} \chi^\ell_n u_n^+$. From the proof of Proposition \ref{prop: rate_convergence} and the fact (see Lemma~\ref{lemma: interaction_left_right_d=3})
	\[
	q_n(u) = \begin{cases}
	\grandO{L_n^{-5}} & \text{if }d=2\, ,\\
	\grandO{L_n^{-\infty}} &\text{if }d=3\, ,
	\end{cases}
	\]
	we have the estimate
	\[
	\abs{\mathcal{E}_n(u_n^+) - 2 \mathcal{E}(u)} = \begin{cases}
	\grandO{L_n^{-5}}  & \text{if }d=2 \, ,\\
	\grandO{L_n^{-\infty}} & \text{if }d=3 \, ,
	\end{cases} 
	\]
	for any $\epsilon >0$.
	For the second term, we use the Hardy-Littlewood-Sobolev inequality, Hölder's inequality, the uniform estimates (see for instance Proposition \ref{prop: study of d(h)} for $u$ and Corollary \ref{cor: boundedness_u_better} for $u_n^+$)
	\[
	\forall r \in \intff{2}{\infty},\quad \norme{u}_{L^r(\R^d)} < \infty \et \sup_n \norme{u_n^+}_{L^r(\R^d)} < \infty \, ,
	\]
	and Proposition \ref{prop: rate_convergence} to get
	\begin{align*}
	\abs{D(\abs{v_n}^2,\abs{v_n}^2) - D(\abs{u}^2,\abs{u}^2)} 
	&\leq \norme{\abs{u}^2 - \abs{v_n}^2}_{L^{\frac{2d}{2d-1}}(\R^d)} \left(\norme{u}^2_{L^{\frac{4d}{2d-1}}(\R^d)}+ \norme{v_n}^2_{L^{\frac{4d}{2d-1}}(\R^d)}\right) \\
	&=\begin{cases}
	\grandO{L_n^{-3+\epsilon}}  & \text{if }d=2 \, ,\\
	\grandO{L_n^{-\infty}} & \text{if }d=3 \pt
	\end{cases}
	\end{align*}
	As in the proof of Proposition \ref{prop: rate_convergence}, the last term \eqref{eq: 29_12} is estimated using the decomposition \eqref{eq: decomposition_cool_and_nice} with the additional information that
	\[
	\norme{v_{n,1}}_{L^1(\R^d)} + \norme{v_{n,1}}_{L^2(\R^d)} = \begin{cases}
	\grandO{L_n^{-3+\epsilon}}  & \text{if }d=2 \, ,\\
	\grandO{L_n^{-\infty}} & \text{if }d=3 \pt
	\end{cases}
	\]
	After similar arguments as the ones used in the proof of Proposition \ref{prop: rate_convergence}, one can show
	\[
	\abs{D\left(\abs{\chi_n^\ell u_n^+}^2,\abs{\chi_n^r u_n^+}^2\right) - \frac{1}{2L_n}} =\begin{cases}
	\grandO{L_n^{-3+\epsilon}}  & \text{if }d=2 \, ,\\
	\grandO{L_n^{-\infty}} & \text{if }d=3 \pt
	\end{cases}
	\]
	This concludes the proof of the first estimate in \eqref{eq: rate_convergence_ground_state_energy}.
	
	Now, we look at $\mu_n^-$. We recall that $\mu_n^-$ is defined by
	\begin{align}
	\label{eq: definition_mu_n_-}
	\mu_n^- \coloneqq \min\enstq{\pdtsc{v}{h_n v}_{L^2(\R^d)}}{v \in \mathcal{D}(h_n),\quad \norme{v}_{L^2(\R^d)} =1 \et v \perp u_n^+} \pt
	\end{align}
	We also recall that $u_n^r \coloneqq u(\cdot - \mathbf{x}_n)$ and $u_n^\ell \coloneqq u(\cdot + \mathbf{x}_n)$. Let 
	\[
	u_\mathrm{trial} \coloneqq \frac{(u_n^r - u_n^\ell)}{\norme{u_n^r - u_n^\ell}_{L^2(\R^d)}} \pt
	\]
	Using the reflection symmetry $\mathcal{R}$, we have that $u_\mathrm{trial} \perp u_n^+$. Hence $u_\mathrm{trial}$ is a trial state for the minimization problem \eqref{eq: definition_mu_n_-}. We have
	\begin{align*}
	\pdtsc{u_\mathrm{trial}}{h_n u_\mathrm{trial}}_{L^2(\R^d)} = 2 \norme{u_n^r - u_n^\ell}_{L^2(\R^d)}^{-2} \left(\pdtsc{u_n^\ell}{h_n u_n^\ell}_{L^2(\R^d)} - \pdtsc{u_n^\ell}{h_n u_n^r}_{L^2(\R^d)}\right) \pt
	\end{align*}
	On one hand, remark that by the uniform estimate $\sup_n\norme{\lvert u_n^+ \rvert^2 * |\cdot|^{-1}}_{L^\infty(\R^d)} < \infty$
	(which is a consequence of Lemma \ref{lemma: technical2} and Corollary \ref{cor: boundedness_u_better}) and by Lemma \ref{lemma: tunneling}, we have
	\begin{gather*}
	\pdtsc{u_n^\ell}{h_n u_n^r}_{L^2(\R^d)} = \grandO{L_n^{-\infty}} \et 
	\norme{u_n^r - u_n^\ell}_{L^2(\R^d)}^{-2} = \frac{1}{2} \left(1 + \grandO{L_n^{-\infty}}\right) \pt
	\end{gather*}
	On the other hand, we remark that
	\[
	\pdtsc{u_n^\ell}{h_n u_n^\ell}_{L^2(\R^d)} - \pdtsc{u}{hu}_{L^2(\R^d)} = \int_{\R^d} \abs{u}^2 \left(\tau_{\mathbf{x}_n} V_n^\mathrm{MF} - V^\mathrm{MF}\right) \pt
	\]	
	Then, Lemma \ref{lemma: lemme_utile_mais_relou} gives us
	\begin{align*}
	\pdtsc{u_n^\ell}{h_n u_n^\ell}_{L^2(\R^d)} = \mu +\begin{cases}
	\grandO{L_n^{-3+\epsilon}}  & \text{if }d=2 \, ,\\
	\grandO{L_n^{-\infty}} & \text{if }d=3 \pt
	\end{cases}
	\end{align*}
	All in all, we have shown
	\[
	\pdtsc{u_\mathrm{trial}}{h_n u_\mathrm{trial}}_{L^2(\R^d)}\leq \mu + \begin{cases}
	\grandO{L_n^{-3+\epsilon}}  & \text{if }d=2 \, ,\\
	\grandO{L_n^{-\infty}} & \text{if }d=3 \, ,
	\end{cases}
	\]
	for any $\epsilon >0$. In particular, $h_n$ admits a second negative eigenvalue for $n$ large enough and for all $\epsilon>0$, we have
	\[
	\mu_n^- \leq \mu + \begin{cases}
	\grandO{L_n^{-3+\epsilon}}  & \text{if }d=2 \, ,\\
	\grandO{L_n^{-\infty}} & \text{if }d=3 \pt
	\end{cases}
	\]
	We concludes the proof of \eqref{eq: rate_convergence_ground_state_energy} by noticing $\mu_n^- > \mu_n^+$. 
\end{proof}

\subsection{Approximation for \texorpdfstring{$u_n^-$}{un-}}

Now, we come to the study of $u_n^-$ which mainly follows the one conducted for $u_n^+$ (see Proposition~\ref{prop: rate_convergence}). Recall that $u_n^-$ is the eigenfunction associated with the second lowest eigenvalue $\mu_n^-$ of the mean-field hamiltonian \eqref{eq: EL_equation} of the two nuclei model \eqref{eq: restricted_hartree}. We choose the normalization such that $\norme{u_n^-}^2_{L^2(\R^d)} = 2$. Following step by step the proof of \cite[Lemma 4.2]{olgiati2020hartree}, we see that $u_n^-$ is odd with respect to the line $\{x_1=0\}$ that is we have $\mathcal{R}[u_n^-] = - u_n^-$. Moreover, we can choose $u_n^-$ such that $u_n^- >0$ on the half-plane $\{x_1 >0\}$ and $u_n^- < 0$ on the half-plane $\{x_1 <0\}$.

A consequence of Proposition \ref{prop: rate_convergence_ground_state_energy} is that $u_n^-$ shares similar \emph{a priori} pointwise bounds as $u_n^+$ (see Proposition \ref{cor: exponential decay u^+_n}).
\begin{prop}[\emph{A priori }exponential decay estimate for $u^-_n$]
	\label{cor: exponential decay u^-_n}
	For all $\epsilon \in \intoo{0}{1}$ and for all $n$ large enough, there exists a constant $C_{\epsilon}> 0$ such that the following pointwise estimate holds
	\begin{align}
	\label{eq: a priori exp decay_u^-_n}
	\lvert u^-_n(x) \rvert\leq C_{\epsilon} \left(e^{-(1-\epsilon)\abs{\mu}^{\frac{1}{2}} \abs{x - \mathbf{x}_n}} + e^{-(1-\epsilon)\abs{\mu}^{\frac{1}{2}} \abs{x+\mathbf{x}_n}}\right)  \, ,
	\end{align}
	for all $x\in\R^d$.
\end{prop}
\begin{proof}
	We follow the proof of Proposition \ref{cor: exponential decay u^+_n}, applying Lemma \ref{lemma: comparison_hoffmann_ostenhoff} to the domain
	\[
	\enstq{(x_1,\dots,x_d)\in\R^d}{x_1 >0 \et \abs{x-\mathbf{x}_n} > R}\, ,
	\]
	where we have chosen $u_n^-$ positive. Then, we extend the estimate to the left half space by using the reflection symmetry $\mathcal{R}$.
\end{proof}
\begin{prop}[Rates of convergence for $u_n^-$]
	\label{prop: rate_convergence_bis}
	For any $\epsilon \in \intoo{0}{3}$, there exists a constant $C_\epsilon$ such that the following estimates hold
	\begin{gather}
	\label{eq: first a priori estimate_bisbis}
	\norme{\abs{u(\cdot + \mathbf{x}_n)}^2 - \abs{\chi_n^\ell u_n^-}^2}_{L^1(\R^d)} = \norme{\abs{u(\cdot - \mathbf{x}_n)}^2 - \abs{\chi_n^\ell u_n^-}^2}_{L^1(\R^d)}  \leq
	\begin{cases}
	C_\epsilon L_n^{-3+\epsilon}  & \text{if }d=2 \, ,\\
	\grandO{L_n^{-\infty}} & \text{if }d=3 \, ,
	\end{cases} \\
	\label{eq: first a priori estimate_bis}
	\norme{u(\cdot + \mathbf{x}_n) + \chi_n^\ell u_n^-}_{H^1(\R^d)} = \norme{u(\cdot - \mathbf{x}_n) - \chi_n^ru_n^-}_{H^1(\R^d)}  \leq
	\begin{cases}
	C_\epsilon L_n^{-3+\epsilon}  & \text{if }d=2 \, ,\\
	\grandO{L_n^{-\infty}} & \text{if }d=3  \pt
	\end{cases} 
	\end{gather}
\end{prop}
\begin{proof}
	For $u_n^-$ we will carry out the same strategy as for $u_n^+$ (see proof of Propositions \ref{prop: strong_convergence_H1} and \ref{prop: rate_convergence}). The arguments being similar, we will only sketch the proof.
	Recall that, by the min-max principle, $u_n^-$ is the unique (up to a phase) minimizer of the problem
	\begin{align}
	\label{eq: min_problem_u_n_-}
	\min \enstq{\pdtsc{v}{h_n v}_{L^2(\R^d)}}{ v\in \mathcal{D}(h_n),\quad v \perp u_n^+,\quad \int_{\R^d} \abs{v}^2=2} \pt
	\end{align}
	First, we give an \emph{a priori }estimate on $u_n^-$.
	We introduce
	\[
	u_\mathrm{trial} = \frac{\sqrt{2}(u_n^r - u_n^\ell)}{\norme{u_n^r - u_n^\ell}_{L^2(\R^d)}} \, ,
	\]
	which is a valid trial state for the minimization problem \eqref{eq: min_problem_u_n_-} (using the reflection symmetry, one can check that $u_\mathrm{trial} \perp u_n^+$). As in the proof of Proposition \ref{prop: rate_convergence_ground_state_energy}, we obtain
	\begin{align}
	\label{md1}
	\frac{1}{2}\pdtsc{u_\mathrm{trial}}{h_n u_\mathrm{trial}}_{L^2(\R^d)}
	&\leq \mu + \int_{\R^d} \abs{u}^2 \left(\tau_{\mathbf{x}_n} V_n^\mathrm{MF} - V^\mathrm{MF}\right) + \grandO{L_n^{-\infty}} \\
	\label{md2}
	&\leq \mu_n^- + \begin{cases}
	\grandO{L_n^{-3+\epsilon}}  & \text{if }d=2 \, ,\\
	\grandO{L_n^{-\infty}} & \text{if }d=3 \, ,
	\end{cases}
	\end{align}
	for any $\epsilon \in \intoo{0}{3}$. To get the second bound \eqref{md2}, we have used Corollary \ref{lemma: lemme_utile_mais_relou} and Proposition~\ref{prop: rate_convergence_ground_state_energy}. Hence, $u_\mathrm{trial}$ is a minimizing sequence for the minimization problem \eqref{eq: min_problem_u_n_-}. Proceeding as in the proof of Proposition \ref{prop: strong_convergence_H1} and Proposition \ref{prop: stronger_convergence}, we can show that
	\begin{align*}
	\label{eq: 12012}
	\norme{u_n^- - u_\mathrm{trial}}_{H^{2 - \alpha} (\R^d)} \limit{n \to \infty} 0  \et \sup_n \norme{u_n^-}_{H^{2-\alpha}(\R^d)} <  \infty \, ,
	\end{align*}
	for any $\alpha >0$ if $d=2$ and $\alpha=0$ if $d=3$.
	We denote $v_n \coloneqq \tau_{\mathbf{x}_n} \chi_n^\ell u_n^-$ which, by $\mathcal{R}$ symmetry, satisfies $\norme{v_n}_{L^2(\R^d)} = 1$. Using the IMS formula  and the $\mathcal{R}$ symmetry, we get
	\begin{align*}
	\frac{1}{2}\pdtsc{u_n^-}{h_n u_n^-}_{L^2(\R^d)}
	&= \pdtsc{\chi_n^\ell u_n^-}{h_n \chi_n^\ell u_n^-}_{L^2(\R^d)} - \int_{\R^d} \abs{\nabla \chi_n^\ell}^2 \lvert u_n^- \rvert^2  \\
	&= \pdtsc{v_n}{h v_n}_{L^2(\R^d)} + \int_{\R^d} \abs{v_n}^2 \left(\tau_{\mathbf{x}_n} V_n^\mathrm{MF} - V^\mathrm{MF}\right) - \int_{\R^d} \abs{\nabla \chi_n^\ell}^2 \lvert u_n^- \rvert^2 \pt
	\end{align*}
	We bound from below the first term using Proposition \ref{prop: stability} and the the last term is a $\grandO{L_n^{-\infty}}$ by Proposition \ref{cor: exponential decay u^-_n}. Thus, we obtain
	\[
	\frac{1}{2}\pdtsc{u_n^-}{h_n u_n^-}_{L^2(\R^d)} \geq \mu + \int_{\R^d} \abs{v_n}^2 \left(\tau_{\mathbf{x}_n} V_n^\mathrm{MF} - V^\mathrm{MF}\right) + C \min_{\theta \in \intff{0}{2\pi}} \norme{e^{i\theta} v_n - u}^2_{H^1(\R^d)} + \grandO{L_n^{-\infty}} \pt
	\]
	Recalling \eqref{md1}, using Proposition \ref{lemma: lemme_utile_mais_relou}, the fact that $\sup_{n\in \N} \norme{v_n}_{H^1(\R^d)} < \infty$
	and the Cauchy-Schwarz inequality, we get
	\begin{align*}
	\min_{\theta \in \intff{0}{2\pi}} \norme{e^{i\theta} v_n - u}^2_{H^1(\R^d)} \leq  \begin{cases}
	\grandO{L_n^{-3+\epsilon}} \left(\norme{v_n + u}_{L^1(\R^d)} + \norme{v_n + u}_{L^2(\R^d)}\right) + \grandO{L_n^{-\infty}}  & \text{if }d=2 \, ,\\
	\grandO{L_n^{-\infty}}\left(1+\norme{v_n + u}_{L^1(\R^d)} + \norme{v_n + u}_{L^2(\R^d)}\right) & \text{if }d=3 \, ,
	\end{cases} 
	\end{align*}	
	since $u_n^- <0$ on $\{x_1 <0\}$.
	Because $v_n$ and $u$ are exponentially decaying, we have $\norme{v_n + u}_{L^1(\R^d)} \leq \grandO{L^{\epsilon}\norme{v_n + u}_{L^2(\R^d)}} + \grandO{L_n^{-\infty}}$ for any $\epsilon>0$.  Hence, we have obtained
	\begin{align*}
	\min_{\theta \in \intff{0}{2\pi}} \norme{e^{i\theta} v_n - u}^2_{H^1(\R^d)} \leq  \begin{cases}
	C_\epsilon L_n^{-3+\epsilon} \norme{v_n + u}_{L^2(\R^d)} + \grandO{L_n^{-\infty}}  & \text{if }d=2 \, ,\\
	\grandO{L_n^{-\infty}}\left(1+ \norme{v_n + u}_{L^2(\R^d)}\right) & \text{if }d=3 \, ,
	\end{cases} 
	\end{align*}
	for all $\epsilon >0$. Now, the minimization problem
	\[
	\min_{\theta \in \intff{0}{2\pi}} \norme{e^{i\theta} v_n - u}^2_{L^2(\R^d)} \, ,
	\]
	is solved for $\theta  = \pi$. To see this, one can expand $\norme{e^{i\theta} v_n - u}^2_{L^2(\R^d)}$ and notice that $\pdtsc{v_n}{u}_{L^2(\R^d)} \leq 0$. Thus, we get 
	\[
	\norme{ v_n + u}_{H^1(\R^d)} \leq 
	\begin{cases}
	C_\epsilon L_n^{-3+\epsilon}  & \text{if }d=2 \, ,\\
	\grandO{L_n^{-\infty}} & \text{if }d=3  \pt
	\end{cases} 
	\]
	To obtain \eqref{eq: first a priori estimate_bisbis}, we just have to recall the estimate $\norme{v_n + u}_{L^1(\R^d)} \leq \grandO{L_n^{\epsilon}\norme{v_n + u}_{L^2(\R^d)}} + \grandO{L_n^{-\infty}}$. This concludes the proof of Proposition \ref{prop: rate_convergence_bis}.
\end{proof}

\subsection{Lower bound on the second gap}
\begin{prop}[Lower bound on the second gap]
	\label{prop: the_remaining_spectrum_is_far_away}
	There exists a constant $C>0$ such that
	\begin{align}
	\label{eq: the_remaining_spectrum_is_far_away}
	\dist (\mu_n^-, \sigma(h_n) \setminus\{\mu_n^+,\mu_n^-\}) \geq C \pt		
	\end{align}
\end{prop}
\begin{proof}
	We denote by $q_{h_n^r}$ (resp. $q_{h_n^r}$, resp. $q_{h_n}$) the quadratic form associated with $h_n^r$ (resp. $h_n^\ell$, resp. $h_n$) and defined on $H^1(\R^d)$.
	We consider $v \in H^1(\R^d)$ such that $\norme{v}_{L^2(\R^2)}^2=1$, $v \perp u_n^+$ and $v \perp u_n^-$. We denote $G_n \coloneqq \dist (\mu_n^-, \sigma(h_n) \setminus\{\mu_n^+,\mu_n^-\}) \geq 0$. Using the IMS formula, Proposition \ref{prop: strong_convergence_H1} and Corollary~\ref{cor: boundedness_u}, we have
	\begin{align}
	\label{eq: balklava}
	\mu + G_n + \petito{1} = \mu_n^- + G_n \geq q_{h_n}(v) = q_{h_n^\ell}(\chi_n^\ell v)+ q_{h_n^r}(\chi_n^r v) + \petito{1} \pt
	\end{align}
	Then, by Proposition \ref{prop: stability}, we have for some $C>0$
	\begin{align*}
	q_{h_n^\ell}(\chi_n^\ell v)
	&\geq \norme{\chi_n^\ell v}^2_{L^2(\R^d)} \left(\mu + 2C \right) - 2C\norme{\chi_n^\ell v}_{L^2(\R^d)} \abs{\pdtsc{\chi_n^\ell v}{u_n^\ell}_{L^2(\R^d)}} \\
	&\geq \norme{\chi_n^\ell v}^2_{L^2(\R^d)} \left(\mu + 2C \right) -  2C\abs{\pdtsc{\chi_n^\ell v}{u_n^\ell}_{L^2(\R^d)}} \pt
	\end{align*}
	A similar lower bound also holds for $q_{h_n^r}(\chi_n^r v) $. Inserting these estimates into \eqref{eq: balklava}, we obtain
	\begin{align}
	\label{eq: baklava_3}
	G_n \geq 2C \left( 1 - \abs{\pdtsc{\chi_n^\ell v}{u_n^\ell}_{L^2(\R^d)}} -\abs{\pdtsc{\chi_n^r v}{u_n^r}_{L^2(\R^d)}}\right) + \petito{1} \pt
	\end{align}
	By noticing that Proposition \ref{prop: strong_convergence_H1} and the orthogonality conditions $v \perp u_n^+$ imply
	\begin{align}
	\label{eq: baklava_1}
	0 = \pdtsc{\chi_n^\ell v}{\chi_n^\ell u_n^+}_{L^2(\R^d)}+ \pdtsc{\chi_n^r v}{\chi_n^r u_n^+}_{L^2(\R^d)}
	= \pdtsc{\chi_n^\ell v}{u_n^\ell}_{L^2(\R^d)} + \pdtsc{\chi_n^r v}{u_n^r}_{L^2(\R^d)} + \petito{1} \pt
	\end{align}
	In a same manner, using Proposition \ref{prop: rate_convergence_bis}, we also have 
	\begin{align}
	\label{eq: baklava_2}
	0 = \pdtsc{v}{u_n^-}_{L^2(\R^d)}  = -\pdtsc{\chi_n^\ell v}{u_n^\ell}_{L^2(\R^d)} + \pdtsc{\chi_n^r v}{u_n^r}_{L^2(\R^d)} + \petito{1} \pt
	\end{align}
	Suitable linear combinations of \eqref{eq: baklava_1} and \eqref{eq: baklava_2} give
	\begin{align*}
	\pdtsc{\chi_n^r v}{u_n^r}_{L^2(\R^d)}  = \petito{1} \et \pdtsc{\chi_n^\ell v}{u_n^\ell}_{L^2(\R^d)} = \petito{1} \pt
	\end{align*}
	Inserting this into \eqref{eq: baklava_3} shows \eqref{eq: the_remaining_spectrum_is_far_away}.
\end{proof}

\subsection{Convergence rates in higher Sobolev spaces}

\begin{prop}[Convergence rate in higher Sobolev spaces]
	\label{prop: convergence_rate_higher_sobolev_spaces}
	Let $\alpha \in \intoo{0}{1}$ if $d=2$ and $\alpha = 0$ if $d=3$. For any $\epsilon \in \intoo{0}{3}$, there exists a constant $C_\epsilon>0$ such that the following estimates hold
	\begin{align*}
	\norme{u_n^\pm - (u(\cdot-\mathbf{x}_n) \pm u(\cdot + \mathbf{x}_n))}_{H^{2-\alpha}(\R^d)}\leq \begin{cases}
	C_\epsilon L_n^{-3 + \epsilon} & \text{if }d=2 \, , \\
	\grandO{L_n^{-\infty}} & \text{if }d=3 \pt
	\end{cases}
	\end{align*}
\end{prop}
\begin{proof}
	We do the proof for $u_n^+$, the arguments being similar for $u_n^-$. Let $\alpha \in \intoo{0}{\frac{1}{2}}$.
	We write
	\begin{align}
	\label{eq: ligne1}
	(-\Delta + \abs{\mu})^{1-\alpha}\left(u_n^+ - (u_n^r + u_n^\ell)\right)
	=&~ \frac{1}{(-\Delta+\abs{\mu})^\alpha} \Big[(\mu_n^+ - \mu)u_n^+  \\
	\label{eq: ligne2}
	& +\left((\tau_n^r + \tau_n^\ell) V^\mathrm{MF} - V^\mathrm{MF}_n\right)u_n^+ \\
	\label{eq: ligne3}
	& +(1+\mathcal{R}) \left[(\tau_n^r V^\mathrm{MF}) (u_n^r - \chi_n^r u_n^+)\right] \\
	\label{eq: ligne4}
	& + (1+\mathcal{R}) \left[(\tau_n^r V^\mathrm{MF}) (\chi_n^r - 1)u_n^+\right] \Big] \pt
	\end{align}
	The right side of \eqref{eq: ligne1} is bounded in $L^2$ using Proposition \ref{prop: rate_convergence_ground_state_energy} and the boundedness of $(-\Delta + \abs{\mu})^{-\alpha}$. The term \eqref{eq: ligne2} is bounded using Hölder's inequality, estimate \eqref{eq: preliminary_estimate} from the proof of Proposition~\ref{lemma: lemme_utile_mais_relou} and Corollary \ref{cor: boundedness_u_better}. To bounds \eqref{eq: ligne3}, we write
	\begin{align*}
	&\norme{\frac{1}{(-\Delta + \abs{\mu})^{\alpha}} (\tau_n^r V^\mathrm{MF}) (u_n^r - \chi_n^r u_n^+)}_{L^2(\R^d)}\\
	&\lesssim \norme{\frac{1}{(-\Delta + \abs{\mu})^{\alpha}} \left(1+\frac{1}{\abs{x}}\right) \frac{1}{(-\Delta + \abs{\mu})^{-\frac{1}{2} + \alpha}}} \norme{u_n^r - \chi_n^r u_n^+}_{H^{1 - 2 \alpha}(\R^d)} \, ,
	\end{align*}
	then we use Lemma \ref{lemma: fractional_hardy_bis} and Proposition \ref{prop: rate_convergence}. Finally, to bound \eqref{eq: ligne4}, we recall that $\supp (1 - \chi_n^r) \subset \{x_1 \leq \delta \sqrt{L_n}\}$ for some $\delta >0$. Then, using Lemma \ref{lemma: long range behavior_MF_potential} and $\sup_n \norme{u_n^+}_{L^\infty(\R^d)} < \infty$, we obtain 
	\[
	\norme{(\tau_n^r V^\mathrm{MF}) (\chi_n^r - 1)u_n^+}_{L^2(\R^d)} \leq \begin{cases}
	\grandO{L_n^{-3}} & \text{if }d=2 \, ,\\
	\grandO{L_n^{-\infty}} & \text{if }d=3 \pt
	\end{cases}
	\]
	This ends the proof of Proposition \ref{prop: convergence_rate_higher_sobolev_spaces}.
\end{proof}

\subsection{Sharper exponential bounds for \texorpdfstring{$u_n^\pm$}{un+-}}

\begin{prop}[Sharper exponential pointwise bounds for $u_n^+$]
	\label{prop: sharper_exponential_bounds}
	There exists $C >0$ such that for $n$ large enough we have for all $x\in \R^d$
	\begin{align}
	\label{eq: sharper_bound_1}
	u_n^+(x) &\geq \frac{1}{C} \left(\frac{e^{-\left(\abs{\mu} + L_n^{-1}\right)^{\frac{1}{2}} \abs{x - \mathbf{x}_n}}}{1+\abs{x - \mathbf{x}_n}^{\frac{d-1}{2}}} + \frac{e^{-\left(\abs{\mu} + L_n^{-1}\right)^{\frac{1}{2}} \abs{x+\mathbf{x}_n}}}{1+\abs{x+\mathbf{x}_n}^{\frac{d-1}{2}}}\right) \, ,\\
	\label{eq: sharper_bound_2}
	u^+_n(x) &\leq C \left(\frac{e^{-\left(\abs{\mu} - L_n^{-1}\right)^{\frac{1}{2}} \abs{x - \mathbf{x}_n}}}{1+\abs{x-\mathbf{x}_n}^\frac{d-1}{2}} + \frac{e^{-\left(\abs{\mu} - L_n^{-1}\right)^{\frac{1}{2}} \abs{x+\mathbf{x}_n}}}{1+\abs{x+\mathbf{x}_n}^\frac{d-1}{2}}\right)  \pt
	\end{align}
\end{prop}
\begin{proof}
	We write the proof in details in the case $d=2$ and mention the modifications when $d=3$. By Lemma \ref{lemma: technical2}, Corollary \ref{cor: boundedness_u} and Proposition \ref{prop: rate_convergence}, we can write
	\begin{align}
	\label{eq: decomposition_MF_potential}
	V_n^\mathrm{MF} = \left(\tau_n^\ell + \tau_n^r\right) V^\mathrm{MF} + \grandO{L_n^{-2}} \, ,
	\end{align}
	where the $O$ makes sense in $L^\infty(\R^2)$. By Proposition \ref{prop: rate_convergence_ground_state_energy}, we have
	\begin{align}
	\label{eq: decomposition_lagrange_multiplier}
	\mu_n^+ =  \mu +   \grandO{L_n^{-2} } \pt
	\end{align}	
	By Lemma \ref{lemma: long range behavior_MF_potential}, there exists $R >0$ such that
	\begin{align}
	\label{eq: behavior_MF}
	0 \leq V^\mathrm{MF}(x) \leq \abs{x}^{-2} \, ,
	\end{align}
	for all $\abs{x} \geq R$.
	Let $\Omega = B(-\mathbf{x}_n,R)^c \cup B(\mathbf{x}_n,R)^c$.
	For $\alpha \geq 0$, we introduce
	\[
	Y_{n,\alpha,\pm} = K_\alpha\left(\left(\abs{\mu} \pm L_n^{-1}\right)^\frac{1}{2} \abs{x}\right) \, ,
	\]
	where $K_\alpha$ denotes the modified Bessel function of the second kind with parameter $\alpha$. The function $Y_{n,\alpha,\pm} $ is positive and satisfies the equation
	\[
	\left(-\Delta + \frac{\alpha}{\abs{x}} - \mu \pm L_n^{-1}\right)Y_{n,\alpha,\pm} = 0\, ,
	\]
	in the region $\{\abs{x} \geq R\}$. Fix $A>0$. Using the following asymptotics \cite[pp. 266–267]{olver1997asymptotics}
	\begin{align*}
	\forall x >0,\quad K_\alpha(x) =\sqrt{ \frac{\pi}{2x} } e^{-x} \left(1 + R(\alpha,x)\right) \quad \text{where} \quad \abs{R(\alpha,x)} \leq \abs{\alpha^2 - 1/4}\frac{e^{\frac{\lvert\alpha^2 - 1/4\rvert}{x}}}{x} \, ,
	\end{align*}
	we show there exists $C>0$ such that
	\begin{align}
	\label{eq: bessel_asymptotics}
	\frac{1}{C \sqrt{\abs{x}}} e^{-\left(\abs{\mu} \pm L_n^{-1}\right)^{\frac{1}{2}} \abs{x}} \leq Y_{n,\alpha,\pm}(x) \leq \frac{C}{ \sqrt{\abs{x}}} e^{-\left(\abs{\mu} \pm L_n^{-1}\right)^{\frac{1}{2}} \abs{x}} \, ,
	\end{align}
	for $n$ large enough (depending only on $A$), for all $\abs{x} \geq R$ and for all $\alpha \in \intff{0}{A}$. In the following, we denote by $Y^{\ell/r}_{n,\alpha,\pm} =  \tau_n^{\ell/r} Y_{n,\alpha,\pm}$ the translations of $Y_{n,\alpha,\pm}$ by $\mathbf{x}_n$ or $-\mathbf{x}_n$.
	
	First, we show the upper bound $u_n^+ \leq C \left(Y^r_{n,0,-} + Y^\ell_{n,0,-}\right)$ on $\R^2$ for some constant $C>0$. By Proposition \ref{prop: stronger_convergence}, there exists $C>0$ such that $u_n^+ \leq C \left(Y^r_{n,0,-} + Y^\ell_{n,0,-}\right)$ on $\Omega^c$. Using \eqref{eq: decomposition_MF_potential}, \eqref{eq: decomposition_lagrange_multiplier} and the first inequality in \eqref{eq: behavior_MF}, we have for $n$ large enough
	\begin{align*}
	(-\Delta + V_n^\mathrm{MF} - \mu_n^+)\left(Y^r_{n,0,-} + Y^\ell_{n,0,-}\right) 
	&= \left(V_n^\mathrm{MF} - \mu_n^+ + \mu + L_n^{-1}\right)\left(Y^r_{n,0,-} + Y^\ell_{n,0,-}\right) \\
	&\geq \left( (\tau_n^\ell + \tau_n^r)V^\mathrm{MF} + L_n^{-1} + \grandO{L_n^{-2}}\right)\left(Y^r_{n,0,-} + Y^\ell_{n,0,-}\right) \\
	&\geq 0 \pt
	\end{align*}	
	By the second comparison Lemma \ref{lemma: comparison_hoffmann_ostenhoff}, we deduce
	\begin{align}
	\label{eq: upper_bound_sharper_1}
	u_n^+ \leq C \left(Y^r_{n,0,-} + Y^\ell_{n,0,-}\right) \quad \text{on} \quad \R^2 \pt
	\end{align}
	
	Now, we show the lower bound $u_n^+ \geq C^{-1} \left(Y^r_{n,\alpha,+} + Y^\ell_{n,\alpha,+}\right)$ for some $\alpha >2$. By Proposition~\ref{prop: stronger_convergence}, there exists $C>0$ such that $u_n^+ \geq C^{-1} \left(Y^r_{n,\alpha,+} + Y^\ell_{n,\alpha,+}\right)$ on $\Omega^c$. Using \eqref{eq: decomposition_MF_potential}, \eqref{eq: decomposition_lagrange_multiplier} and the second inequality in \eqref{eq: behavior_MF}, we have for $n$ large enough
	\begin{align*}
	(-\Delta + V_n^\mathrm{MF} - \mu_n^+) \left(Y^r_{n,\alpha,+} + Y^\ell_{n,\alpha,+}\right)
	\leq& ~\left( \frac{1}{\abs{x+\mathbf{x}_n}^2} - \frac{\alpha-1}{\abs{x-\mathbf{x}_n}^2}\right)Y^r_{n,\alpha,+} \\
	&+ \left( \frac{1}{\abs{x-\mathbf{x}_n}^2} - \frac{\alpha-1}{\abs{x+\mathbf{x}_n}^2}\right)Y^\ell_{n,\alpha,+} \\
	\leq &~ \left(1 + \mathcal{R}\right)\left[f_n Y^r_{n,\alpha,+}\right] \, ,
	\end{align*}
	where we have denoted $f_n(x) = \frac{1}{\abs{x+\mathbf{x}_n}^2} - \frac{\beta}{\abs{x-\mathbf{x}_n}^2}$ and $\beta = \alpha-1>1$. An elementary computation shows that $f_n(x) \geq 0$ if and only if $x \in \mathcal{C}_\beta$ where $\mathcal{C}_\beta$ is the disk defined by $\{\abs{x+\gamma_\beta \mathbf{x}_n} \leq s_\beta \frac{L_n}{2}\}$ with $\gamma_\beta = \frac{\beta+1}{\beta-1}$ and $s_\beta = \frac{2\sqrt{\beta}}{\beta-1}$. We also have
	\[
	d_\mathrm{min} \coloneqq \dist(\mathbf{x}_n,\mathcal{C}_\beta) = \frac{2\sqrt{\beta}}{\sqrt{\beta}+1}\frac{L_n}{2} \et d_\mathrm{max} \coloneqq \max_{x \notin \mathcal{C}_\beta} \abs{x+\mathbf{x}_n} = \frac{2}{\sqrt{\beta}-1} \frac{L_n}{2}\pt
	\]
	Let $x \in \mathcal{C}_\beta \setminus B(0,R)$. Using the estimates \eqref{eq: bessel_asymptotics}, we see that
	\begin{gather*}
	f_n(x)Y^r_{n,\alpha,+}(x) \leq \frac{C}{R^2 \sqrt{d_\mathrm{min}} } e^{-\left(\abs{\mu} + L_n^{-1}\right)^{\frac{1}{2}} d_\mathrm{min}} \, ,\\
	\mathcal{R}[f_nY^r_{n,\alpha,+}](x) \leq - \frac{\beta -1}{\abs{x+\mathbf{x}_n}^2} \frac{C}{\sqrt{d_\mathrm{max}} } e^{-\left(\abs{\mu} + L_n^{-1}\right)^{\frac{1}{2}} d_\mathrm{max}} \pt
	\end{gather*}
	Comparing the exponents, we see the amplitude of the second term is larger if we choose $\alpha >2$ such that $\frac{2\sqrt{\beta}}{\sqrt{\beta}+1} \geq \frac{2}{\sqrt{\beta}-1} + 1$ that is if $\alpha \geq 10 + 4\sqrt{5}$. As a consequence, for $n$ large enough, we have 
	\[
	(-\Delta + V_n^\mathrm{MF} - \mu_n^+) \left(Y^r_{n,\alpha,+} + Y^\ell_{n,\alpha,+}\right) \leq 0 \, ,
	\]
	on $\Omega$ and, by the second comparison Lemma \ref{lemma: comparison_hoffmann_ostenhoff}, we have
	\begin{align}
	\label{eq: upper_bound_sharper_2}
	u_n^+ \geq C^{-1} \left(Y^r_{n,\alpha,+} + Y^\ell_{n,\alpha,+}\right) \quad \text{on} \quad \R^2 \pt
	\end{align}
	We obtain the exponential bounds \eqref{eq: sharper_bound_1} and \eqref{eq: sharper_bound_2} for $u_n^+$ by using \eqref{eq: upper_bound_sharper_1} and \eqref{eq: upper_bound_sharper_2} together with \eqref{eq: bessel_asymptotics}.
	
	For the case $d=3$, we use modified spherical Bessel functions instead of $K_\alpha$. The rest of the argument is the same.	
\end{proof}
\begin{rem}
	In the statement of Proposition \ref{prop: sharper_exponential_bounds}, we can replace $L_n^{-1}$ by $L_n^{-\delta}$ for any $\delta >0$ if $d=2$ and by $L_n^{-k}$ for any $k\in \N$ if $d=3$, the constant $C$ remaining independent from $\delta$ or $k$.
\end{rem}

\subsection{Estimate on the spectral gap \texorpdfstring{$\mu_L^- - \mu_L^+$}{}}

In this section, we estimate the spectral gap $\mu_n^- - \mu_n^+$. From Proposition \ref{prop: rate_convergence_ground_state_energy}, we already know that $\abs{\mu_n^- - \mu_n^+} = \grandO{L_n^{-3+\epsilon}}$ for any $\epsilon >0$. The next theorem says this gap is in fact exponentially small.
\begin{theo}[Spectral gap estimation]
	\label{th: spectral_gap_estimation}
	There exists $C>0$ such that we have for $n$ large enough
	\begin{align}
	\label{eq: spectral_gap_estimation_bis}
	\boxed{\frac{1}{C}\frac{T_n}{L_n^d} \leq \mu_n^- - \mu_n^+ \leq C T_n \pt}
	\end{align}
\end{theo}
As in \cite[Proposition 2.2]{simon1984semiclassicalII} or \cite{olgiati2020hartree}, we use the ground state substitution formula to improve the convergence rate.
\begin{lemma}[Ground state substitution formula \cite{simon1984semiclassicalII}]
	\label{lemma: ground_state_resolution}
	Let $\Omega \subset \R^d$ be an open set.
	Let $H = -\Delta + V$ be a bounded from below self-adjoint operator on $L^2(\Omega)$. Let $\psi$ (resp. $\lambda$) denotes the ground state (resp. ground state energy) associated with $H$. Let $g$ be any $\mathcal{C}^1(\Omega)$ uniformly bounded function. Then, we have
	\[
	\pdtsc{g\psi}{\left(H - \lambda\right)g\psi}_{L^2(\Omega)} = \frac{1}{2} \norme{(\nabla g) \psi}_{L^2(\Omega)}^2 \pt
	\]
\end{lemma}
\begin{proof}[Proof of Theorem \ref{th: spectral_gap_estimation}]
	Let $\epsilon >0$. We introduce the set
	\[
	\Omega_n \coloneqq \enstq{x \in \R^2}{\abs{x - \mathbf{x}_n} \leq L_n \et \abs{x + \mathbf{x}_n} \leq L_n} \pt
	\]
	To get the upper bound in \eqref{eq: spectral_gap_estimation_bis}, we use Lemma \ref{lemma: ground_state_resolution} with a function $g_n \in \mathcal{C}^1(\R^d)$ satisfying
	\begin{gather*}
	\mathcal{R}[g_n] = - g_n ~\, ,\quad -1 \leq g_n \leq 1~\, , \quad \sup_n \norme{\nabla g_n }_{L^\infty(\R^d)} < \infty\, , \\
	g_n \equiv 1 \text{ on } B\left(\mathbf{x}_n,\frac{L_n - \delta}{2}\right) \et g_n \equiv 0 \text{ on } \Omega_n^c \, ,
	\end{gather*}
	for some constant $\delta>0$.
	Using the reflection symmetry $\mathcal{R}$, we have $g_nu_n^+ \perp u_n^+$. Hence, $g_n u_n^+$ is a trial state for the minimization problem \eqref{eq: min_problem_u_n_-} and we have
	\begin{align*}
	\mu_n^- - \mu_n^+ \leq \frac{\pdtsc{g_nu_n^+}{\left(h_n - \mu_n^+\right) g_nu_n^+}_{L^2(\R^d)}}{\norme{g_n u_n^+}^2_{L^2(\R^d)}} = \frac{1}{2}\left(\frac{\norme{(\nabla g_n) u_n^+}_{L^2(\R^d)}}{\norme{g_n u_n^+}_{L^2(\R^d)}}\right)^2 \pt
	\end{align*}
	Using the exponential bounds for $u_n^+$ (see for instance Proposition \ref{cor: exponential decay u^+_n}) and 
	\[
	\supp \nabla g_n\subset \Omega_n \setminus \left(B\left(\mathbf{x}_n,\frac{L_n - \delta}{2}\right) \cup B\left(-\mathbf{x}_n,\frac{L_n - \delta}{2}\right)\right) \, ,
	\]
	one can show that for $n$ large enough $\norme{g_n u_n^+}^2_{L^2(\R^d)} \geq 1$. Moreover, using the sharper exponential bounds on $u_n^+$ from Proposition \ref{prop: sharper_exponential_bounds}, we also have
	\begin{align*}
	\norme{\left(\nabla g_n\right) u_n^+}^2_{L^2(\R^d)}
	\lesssim \int_{\frac{L_n - \delta}{2} \leq \abs{x} \leq 2L_n} \frac{e^{-2(\abs{\mu} - L_n^{-1})^\frac{1}{2}\abs{x}}}{1+\abs{x}^{d-1}} \diff x \lesssim e^{-(\abs{\mu} - L_n^{-1})^\frac{1}{2}(L_n - \delta)} \lesssim e^{-\abs{\mu}^\frac{1}{2}L_n} = T_n \pt
	\end{align*}
	for some $C_\epsilon>0$. This gives the upper bound of \eqref{eq: spectral_gap_estimation_bis}.
	
	To get the lower bound in \eqref{eq: spectral_gap_estimation_bis}, we introduce a function $f_n \in \mathcal{C}^1(\R^d)$ satisfying the following conditions
	\begin{gather}
	\mathcal{R}[f_n] = f_n~\, ,\quad 0 \leq f_n \leq 1~\, , \quad \sup_n \norme{\nabla f_n }_{L^\infty(\R^d)} < \infty\, , \\
	f_n \equiv 1 \text{ on } \Omega_n \et f_n \equiv 0 \text{ on } \left(\Omega_n + B(0,1)\right)^c \pt
	\end{gather}
	We notice that by Proposition \ref{cor: exponential decay u^-_n} we have
	\begin{align*}
	\norme{f_n u_n^- - u_n^-}^2_{H^1(\R^d)} \lesssim \int_{\Omega_n^c} \left(\abs{\nabla u_n^-}^2 + \abs{u_n^-}^2\right) = \grandO{L_n^{-\infty}T_n} \pt
	\end{align*}
	This implies $\norme{f_n u_n^-}^2_{L^2(\R^2)} = 2 + \grandO{L_n^{-\infty}T_n}$.
	Then, recalling that for all $\epsilon >0$ there exists $C_\epsilon>0$ such that $0 \leq h_n - \mu_n^+ \leq -(1+\epsilon)\Delta + C_\epsilon$ in the sense of quadratic forms, we obtain
	\begin{align*}
	\mu_n^- - \mu_n^+ = \frac{1}{2} \pdtsc{u_n^-}{(h_n - \mu_n^+)u_n^-}_{L^2(\R^2)} 
	&= \frac{1}{2} \pdtsc{f_n u_n^-}{(h_n - \mu_n^+) f_n u_n^-}_{L^2(\R^d)} + \grandO{\norme{f_n u_n^- - u_n^-}^2_{H^1(\R^d)}} \\
	&= \frac{1}{2} \pdtsc{f_n u_n^-}{(h_n - \mu_n^+) f_n u_n^-}_{L^2(\R^d)} + \grandO{L_n^{-\infty}T_n} \pt
	\end{align*}
	We want to apply Lemma \ref{lemma: ground_state_resolution} with $\widetilde{f}_n = f_nu_n^- / u_n^+$ which is well-defined since $u_n^+ >0$. First, by Proposition~\ref{prop: regularity}, we have $\widetilde{f}_n \in \mathcal{C}^1(\R^d \setminus\{\pm \mathbf{x}_n\})$. Using Lemma \ref{lemma: ground_state_resolution}, we get
	\begin{align*}
	\mu_n^- - \mu_n^+ = \frac{1}{2} \pdtsc{\tilde{f}_n u_n^+}{(h_n - \mu_n^+) \tilde{f}_n u_n^+}_{L^2(\R^d)} + \grandO{L_n^{-\infty}T_n} = \frac{1}{4} \norme{(\nabla \widetilde{f}_n) u_n^+}_{L^2(\R^d)}^2 + \grandO{L_n^{-\infty}T_n} \pt
	\end{align*}
	It remains to bound from below the right side. We write the argument only for $d=2$, the other case being identical. From Proposition \ref{prop: sharper_exponential_bounds}, Young's inequality and triangle inequality, we have
	\begin{align}
	\label{eq: lower_bound}
	\forall x\in\R^d,\quad u_n^+(x) \gtrsim  \frac{e^{- \left(\abs{\mu} + L_n^{-1}\right)^\frac{1}{2} \frac{\abs{x-\mathbf{x}_n} + \abs{x+\mathbf{x}_n}}{2}}}{(1+\abs{x-\mathbf{x}_n}+\abs{x+\mathbf{x}_n})^{\frac{1}{2}}}  \pt
	\end{align}
	We consider the set $\Gamma_n = \{0 \leq x_1 \leq L_n/2\} \cap \{\abs{x_2} \leq 1\}$ where we have $\abs{x-\mathbf{x}_n}+\abs{x+\mathbf{x}_n} \leq L_n + 1$. Then, from \eqref{eq: lower_bound}, we have
	\begin{align*}
	\norme{(\nabla \widetilde{f}_n) u_n^+}_{L^2(\R^2)}^2 
	\geq \int_{\Gamma_n} \abs{(\nabla \widetilde{f}_n) u_n^+ }^2  \gtrsim \frac{T_n}{L_n}  \int_{\Gamma_n} \lvert \nabla \widetilde{f}_n\rvert^2  \gtrsim \frac{T_n}{L_n}\int_{-1}^1 \diff x_2  \int_0^{\frac{L_n}{2}} \diff x_1 \abs{\partial_{x_1} \widetilde{f}_n(x_1,x_2)}^2 \pt
	\end{align*}
	Because $u_n^- \equiv 0$ on $\{x_1 = 0\}$, we have
	\begin{align*}
	\abs{\widetilde{f}_n\left(\frac{L_n}{2},x_2\right)}^2 
	&= \abs{ \widetilde{f}_n\left(\frac{L_n}{2},x_2\right) - \widetilde{f}_n(0,x_2)}^2 = \abs{\int_{0}^{\frac{L_n}{2}} \diff x_1 \partial_{x_1} \widetilde{f}_n(x_1,x_2)}^2 \\
	&\leq \frac{L_n}{2} \int_{0}^{\frac{L_n}{2}} \diff x_1  \abs{\partial_{x_1}\widetilde{f}_n(x_1,x_2)}^2 \pt
	\end{align*}
	Moreover, by Proposition \ref{prop: convergence_rate_higher_sobolev_spaces}, $u_n^+$ and $u_n^-$ converge toward $u$ in the vicinity of $\mathbf{x}_n$ in $L^\infty(\R^2)$. We deduce that for all $\abs{x_2} \leq 1$, we have $\tilde{f}_n\left(\frac{L_n}{2},x_2\right) \geq \frac{1}{2}$ for $n$ large enough.
	Therefore, we have shown
	\[
	\norme{(\nabla \widetilde{f}_n) u_n^+}_{L^2(\R^2)}^2  \gtrsim \frac{T_n}{L_n^2} \, ,
	\]
	which is the lower bound in \eqref{eq: spectral_gap_estimation_bis}.
\end{proof}

\appendix

\section{Two-dimensional multipole expansion}\label{sec: expansion_formula}

In this section, we give a precise expansion formula for the quasi-coulombic potential created by a exponentially decreasing charge distribution $\rho$. When the potential is the Green function of the Laplace operator, this question is resolved by the Newton's theorem \cite{newton1833philosophiae} and all the orders except the first one disappear. The situation we consider here is more complicated as the potential is not harmonic but derives from the Green function of the non-local operator $\sqrt{-\Delta}$.
\begin{lemma}
	\label{lemma: potential_expansion}
	Let $\rho \in L^1(\R^2)$ be exponentially decaying from the origin, that is $\abs{\rho(x)} \leq C \exp(-\alpha\abs{x})$ for some $C,\alpha >0$. For $x\in\R^2 \setminus \{0\}$, we denote $\hat{x} = x/\abs{x}$. Then, for any $N\in\N$, we have the expansion formula up to order $N$
	\begin{align}
	\label{eq: expansion_formula_general_order_N}
	\left(\rho \ast |\cdot|^{-1}\right)(x)
	= \sum_{n =0}^{N-1}\frac{1}{\abs{x}^{n+1}} \left(\int_{\R^2} \rho(y)P_n(\widehat{x}\cdot\widehat{y})\abs{y}^n \diff y\right) + \grandO{\frac{1}{\abs{x}^{N+1}}}  \, ,
	\end{align}
	where the $O$ depends on $\alpha$, $C$ and $N$ and where $P_n$ denotes the \emph{Legendre} \emph{polynomial} of order $n$.
	If $\rho$ is radial then
	\begin{align}
	\label{potential_expansion2}
	\int_{\R^2} \frac{\rho(y)}{\abs{x - y}} \diff y 
	=\sum_{n =0}^{N-1} \frac{1}{4^{2n}} \binom{2n}{n}^2 \frac{1}{\abs{x}^{2n+1}}  \left(\int_{\R^2} \rho(y) \abs{y}^{2n} \diff y\right) +\grandO{\frac{1}{\abs{x}^{2N+1}}} \, ,
	\end{align}
	where the $O$ depends on $\alpha$, $C$ and $N$.
\end{lemma}
\begin{proof}
	Let $x,y \in \R^2$ such that $\abs{x} > \abs{y}$. We recall the identity \cite[Eq. 22.9.3]{abramowitz1964handbook}
	\begin{align}
	\label{eq: gegenbauer_expansion}
	\frac{1}{\abs{x - y}} 
	= \frac{1}{(\abs{x}^2 - 2 x \cdot y + \abs{y}^2)^{1/2}}
	=\frac{1}{\abs{x}} \sum_{n=0}^\infty P_n(\widehat{x}\cdot\widehat{y}) \left(\frac{\abs{y}}{\abs{x}}\right)^{n} \, ,
	\end{align}
	where $P_n$ denotes the Legendre polynomial of order $n$. Let $R : \R_+ \to \R^*_+ $ be a function such that :
	\[
	\exists \gamma \in \intoo{0}{1} ,~ R(r)/r^\gamma \to \infty \, , \quad \forall r \geq 0,~ R(r) \leq r/2 \et  \lim\limits_{r \to \infty}R(r)/r = 0 \pt
	\]
	Then, we have for all $x \in \R^2 \setminus \{0\}$
	\begin{align}
	\label{eq: decomposition_chasles}
	\left(\rho \ast |\cdot|^{-1}\right)(x)
	= \int_{\abs{y} \leq R(\abs{x})} \frac{\rho(y)}{\abs{x - y}}\diff y 
	+ \int_{\abs{x - y} \leq R(\abs{x})} \frac{\rho(y)}{\abs{x - y}}\diff y
	+ \int_{\min(\abs{y},\abs{x-y})> R(\abs{x})} \frac{\rho(y)}{\abs{x - y}}\diff y \pt
	\end{align}
	Using \eqref{eq: gegenbauer_expansion} and the estimates (see \cite[Ineq. 22.14.7]{abramowitz1964handbook} for the first one)
	\begin{gather}
	\label{eq: 10_02_2021}
	\forall \abs{z}<1\, ,\abs{P_n(x)} \leq 1 \et
	\forall n\in \R,~ \forall R \geq 0 \, , \int_{\abs{y} \leq R} e^{-\alpha \abs{y}} \abs{y}^n \diff y \leq 2\pi \frac{R^{n+1}}{n+1} \,  ,
	\end{gather}
	we can apply Fubini's theorem for
	the first term in \eqref{eq: decomposition_chasles}
	\begin{align*}
	\int_{\abs{y} \leq R(\abs{x})} \frac{\rho(y)}{\abs{x - y}}\diff y  = \frac{1}{\abs{x}} \sum_{n=0}^\infty  \frac{\widehat{P}_n(x)}{\abs{x}^n} \quad \text{with} \quad \widehat{P}_n(x) = \int_{\abs{y} \leq R(\abs{x})} \rho(y)P_n(\widehat{x}\cdot\widehat{y})\abs{y}^n \diff y \pt
	\end{align*}
	To bound the second term in \eqref{eq: decomposition_chasles}, we use the triangle inequality $\abs{y} \geq \abs{x} - \abs{x-y}$ and the properties of the function $R$. We get
	\begin{align*}
	\int_{\abs{x - y} \leq R(\abs{x})} \frac{\rho(y)}{\abs{x - y}}\diff y \leq 2\pi C e^{-\alpha \abs{x}} \int_{0}^{R(\abs{x})} e^{\alpha r} \diff r \leq \frac{2\pi}{\alpha} e^{-\alpha(\abs{x} - R(\abs{x}))} \leq \frac{2\pi}{\alpha} e^{-\alpha R(\abs{x})} \pt
	\end{align*}
	The last term in \eqref{eq: decomposition_chasles} is easily estimated by $ \frac{2\pi C}{\alpha}\left(1+\frac{1}{\alpha R(\abs{x})}\right)e^{-\alpha R(\abs{x})}$. Then, we have shown 
	\begin{align}
	\label{eq: 10_02_2021_2bis}
	\left(\rho \ast |\cdot|^{-1}\right)(x)
	= \sum_{n=0}^\infty  \frac{\widehat{P}_n(x)}{\abs{x}^{n+1}} + \grandO{e^{-\alpha R(\abs{x})}} \, ,
	\end{align}
	where the $O$ depends only on $\alpha$ and $C$. We truncate this series expansion up to the order $N$ for some $N\in\N$. Then we use the estimate
	\begin{align*}
	\abs{\int_{\abs{y} > R(\abs{x})} \rho(y)P_n(\widehat{x}\cdot\widehat{y})\abs{y}^{n-1} \diff y}
	\leq 2 \pi C \int_{R(\abs{x})}^\infty r^{n} e^{-\alpha r} \diff r = \frac{2\pi C n!}{\alpha^{n+1}} e^{-\alpha R(\abs{x})} \sum_{k \leq n} \frac{\left(\alpha R(\abs{x})\right)^k}{k!} \, ,
	\end{align*}
	to get
	\begin{align}
	\label{eq: 10_02_2021_3bis}
	\sum_{n\leq N-1} \frac{\widehat{P}_n(x)}{\abs{x}^{n+1}}
	= \sum_{n \leq N-1} \frac{1}{\abs{x}^{n+1}} \left(\int_{\R^2} \rho(y)P_n(\widehat{x}\cdot\widehat{y})\abs{y}^n \diff y\right) + \grandO{\left(\frac{R(\abs{x})}{\abs{x}}\right)^Ne^{-\alpha R(\abs{x})}} \, ,
	\end{align}
	where the $O$ depends only on $C$, $\alpha$ and $N$. Recalling estimates from \eqref{eq: 10_02_2021}, we have
	\begin{align*}
	\abs{\sum_{n\geq N}  \frac{\widehat{P}_n(x)}{\abs{x}^{n+1}}} 
	&\leq 2\pi C \sum_{n\geq N} \frac{1}{n+1} \left(\frac{R(\abs{x})}{\abs{x}}\right)^{n+1} 
	\leq 2\pi C \sqrt{\sum_{n\geq N} \frac{1}{(n+1)^2}} \sqrt{\sum_{n\geq N} \left(\frac{R(\abs{x})}{\abs{x}}\right)^{2(n+1)}} \\
	&= \grandO{\frac{1}{\sqrt{N}} \left(\frac{R(\abs{x})}{\abs{x}}\right)^{N+1}} \, ,
	\end{align*}
	where the $O$ depend only on $C$ and $N$. Of course, we have used the Cauchy-Schwarz inequality to get the second line. We insert this previous estimate and \eqref{eq: 10_02_2021_3bis} into \eqref{eq: 10_02_2021_2bis}. By choosing $R(\abs{x})$ small enough compared to $\abs{x}$, we get \eqref{eq: expansion_formula_general_order_N} with $\grandO{\abs{x}^{-N-1+\epsilon}}$ as remaining term for some $\epsilon \in \intoo{0}{1}$. To get rid of the $\epsilon$, we write the next order expansion then truncate the expansion to the last order term.
	
	If we assume that $\rho$ is radial, we can explicitly compute the integral $\int_{\R^2} \rho(y)P_n(\widehat{x}\cdot\widehat{y})\abs{y}^n \diff y$ for $n\in\N$. By radial symmetry, we can assume $\widehat{x} = (1,0)$ and switch to polar coordinates. Identity \cite[Eq. 22.13.6]{abramowitz1964handbook} show that $\int_0^{2\pi} P_{2n}(\cos \theta) \diff \theta = \frac{2\pi}{4^{2n}} \binom{2n}{n}^2$. Recall that $P_{2n+1}$ involves only odd degree monomials. As a consequence: $\int_0^{2\pi} P_n(\cos \theta) \diff \theta = 0$. Then, we have shown \eqref{potential_expansion2}.
\end{proof}
\begin{rem}
	\label{rem: expansion_formula}
	\begin{enumerate}[noitemsep, label=(\roman*)]
		\item In the case where $\rho$ is radial, the two first orders are given by
		\begin{align*}
		\int_{\R^2} \frac{\rho(y)}{\abs{x - y}} \diff y 
		= \frac{1}{\abs{x}} \int_{\R^2} \rho(y) \diff y + \frac{1}{4\abs{x}^3} \int_{\R^2} \rho(y) \abs{y}^2 \diff y + \frac{9}{64\abs{x}^5} \int_{\R^2} \rho(y) \abs{y}^4 \diff y + \grandO{\frac{1}{\abs{x}^{7}}} \pt
		\end{align*}
		\item More generally, with the same assumptions as in Lemma \ref{lemma: potential_expansion} and following a similar proof, we can show that for all $a,\delta>0$ and all $N\geq 0$, we have
		\begin{align}
		\label{eq: expansion_formula_general_order_N_bis}
		\int_{\abs{y} \leq (1-\delta)\abs{x}} \frac{\rho(y)}{\abs{x-y}^a} \diff y
		= \sum_{n =0}^{N-1} \frac{1}{\abs{x}^{n+a}} \left(\int_{\R^2} \rho(y)C_n^{(a/2)}(\widehat{x}\cdot\widehat{y})\abs{y}^n \diff y\right) + \grandO{\frac{1}{\abs{x}^{N+a}}}  \, ,
		\end{align}
		where the $O$ depends on $C$, $\alpha$, $N$ and $a$. Here, $C_n^{(a/2)}$ denotes the \emph{ultraspherical} (or \emph{Gegenbauer}) polynomials with parameter $a/2$. The cutoff in \eqref{eq: expansion_formula_general_order_N_bis} is mandatory only in the case $a\geq 2$. Otherwise, the $\rho \ast |\cdot|^{-a}$ would not be well-defined. If $\rho$ is radial, then we can explicit the coefficients
		\begin{align}
		\label{potential_expansion2_bis}
		\int_{\abs{y} \leq (1-\delta)\abs{x}} \frac{\rho(y)}{\abs{x - y}^a} \diff y 
		=\sum_{n =0}^{N-1} \binom{-a/2}{n}^2\frac{1}{\abs{x}^{2n+a}}  \left(\int_{\R^2} \rho(y) \abs{y}^{2n} \diff y\right) +\grandO{\frac{1}{\abs{x}^{2N+a}}} \pt
		\end{align}
		If we assume that $\rho$ has compact support then we can get rid of the cutoff in \eqref{eq: expansion_formula_general_order_N_bis} and \eqref{potential_expansion2_bis}.
	\end{enumerate}	
\end{rem}

\section{Proof of Lemma \ref{lemma: lemma_technical_exponential_convolution}}\label{sec:technical-results}

	The proof of the first part of \eqref{eq: lemma_technical_exponential_convolution} is given in \cite[Lemma 21]{gontier2020nonlinear} where the authors give an efficient proof of in the case $k = \frac{d-1}{2}$ and in dimensions $d \geq 1$. Our proof is more computational, not optimal when $k = \frac{d-1}{2}$ and dimension dependent but we are able to handle estimate the second part of \eqref{eq: lemma_technical_exponential_convolution}.
We recall the following basic convexity/concavity inequalities
\begin{gather}
\label{eq: convexity_inequality_1}
\forall k \geq 1,~\forall a \geq 0,~ 1+a^k \leq (1+a)^k \leq 2^{k-1}(1+a^k) \, ,\\
\label{eq: convexity_inequality_2}
\forall k \in \intff{0}{1},~\forall a\geq 0,~ (1+a)^k \leq 1+a^k \leq 2^{1-k}(1+a)^k \pt
\end{gather}
First, we treat the $d=2$ case.	
By radial symmetry, we only have to show the first part of \eqref{eq: lemma_technical_exponential_convolution} for $x = (L,0)$ for any $L\geq 0$. Using \eqref{eq: convexity_inequality_1} if $k\geq 1$ or \eqref{eq: convexity_inequality_2} if $k\in\intof{0}{1}$ and after an affine change of variables, we obtain
\[
\abs{(v\ast v)(x)} \lesssim \int_{\R^2} \frac{e^{-\nu  (\abs{y-\frac{x}{2}}+\abs{y+\frac{x}{2}})}}{\left(1 + \abs{y-\frac{x}{2}}\right)^k \left(1+ \abs{y+\frac{x}{2}}\right)^k } \diff y \lesssim  
\int_{\R^2} \frac{e^{-\nu  (\abs{y-\frac{x}{2}}+\abs{y+\frac{x}{2}})} }{(1 + \abs{y-\frac{x}{2}} + \abs{y+\frac{x}{2}})^k}\diff y  \pt
\]
We denote by $I$ the second integral.
For all $a \geq L/2$, the level set $\abs{y-\frac{x}{2}} + \abs{y+\frac{x}{2}} = 2a$ is the ellipse $\mathcal{E}_a$ with linear excentricity $c = L/2$ and semi-major axis $a$. Recalling the circumference of $\mathcal{E}_a$ is equal to $4aE(\frac{L}{2a})$ where $E$ denoted the complete elliptic integral of second kind, we have
\[
I = \int_{L/2}^\infty \frac{e^{-2\nu a}}{(1+2a)^k} \left(\int_{\mathcal{E}_a} \diff y\right)  \diff a = \int_{L/2}^\infty \frac{4aE(\frac{L}{2a})}{(1+2a)^k} e^{-2\nu a} \diff a = \int_{L}^\infty \frac{aE(\frac{L}{a})}{(1+a)^k} e^{-\nu a} \diff a \pt
\]
The map $e \mapsto E(e)$ being bounded on $\intff{0}{1}$ by $\frac{\pi}{2}$, we deduce
\[
I \leq \frac{\pi}{2} \int_{L}^\infty \frac{e^{-\nu a}}{(1+a)^{k-1}} \diff a \pt
\]
When $k \geq 1$, we bound $(1+a)^{-(k-1)}$ by $(1+L)^{-(k-1)}$. When $0 \leq k < 1$, an integration by parts leads to
\[
I \leq \frac{\pi}{2 \nu} \frac{e^{-\nu L}}{(1+L)^{k-1}} + \frac{\pi (1-k)}{2 \nu} \int_{L}^\infty \frac{e^{-\nu a}}{(1+a)^{k}} \diff a \leq \frac{\pi e^{-\nu L}}{2\nu} \left( \frac{1}{(1+L)^{k-1}} + \frac{1-k}{\nu(1+L)^k}\right) \pt
\]
Using again the convexity/concavity inequalities \eqref{eq: convexity_inequality_1} and \eqref{eq: convexity_inequality_2} depending on $k$ shows the first part of \eqref{eq: lemma_technical_exponential_convolution}. 

To prove the second part of \eqref{eq: lemma_technical_exponential_convolution}, we use the same strategy which amounts to bound
\[
J \coloneqq \int_{L/2}^\infty  \frac{e^{-2\nu a}}{(1+2a)^k} \left(\int_{\mathcal{E}_a} \frac{\diff y}{\abs{y+\frac{x}{2}}} \right)\diff a \pt
\]
First, we estimate $\int_{\mathcal{E}_a} \frac{\diff y}{\abs{y+\frac{x}{2}}}$.
Let $b = \sqrt{a^2-c^2}$ (recall that $c = L/2$) be the semi-minor axis of $\mathcal{E}_a$. Then using the standard parametric representation $y(s) = (a\cos(s),b\sin(s))$, we have
\begin{gather*}
\abs{y(s) + \frac{x}{2}} = \sqrt{(c+a\cos(s))^2+b^2\sin^2(s)} = a + c \cos(s) \, , \\
\abs{y'(s)} = \sqrt{a^2 \sin^2(s)+b^2\cos^2(s)}= \sqrt{a^2 - c^2 \cos^2(s)} \pt
\end{gather*}
This leads to
\[
\int_{\mathcal{E}_a} \frac{\diff y}{\abs{y+\frac{x}{2}}} = 2 \int_0^\pi \frac{\abs{y'(s)}}{\abs{y(s) + \frac{x}{2}}} \diff s = 2 \int_{0}^{\pi} \sqrt{\frac{a - c \cos(s)}{a+ c \cos(s)}} \diff s \leq 2 \pi \sqrt{\frac{a+c}{a-c}} \leq \frac{2\sqrt{2a} \pi}{\sqrt{a-c}} \pt
\]
Going back to $J$, we have
\begin{align}
\label{eq: nom_a_la_con}
J \leq 2\sqrt{2}\pi \int_{L/2}^\infty  \frac{\sqrt{a}e^{-2\nu a}}{(1+2a)^k\sqrt{a-c}}  \diff a \leq \frac{\sqrt{2}\pi}{(1+L)^k} \int_L^\infty \frac{\sqrt{a} e^{-\nu a}}{\sqrt{a- L}} \diff a \pt
\end{align}
The map $a \mapsto \frac{\sqrt{a}}{\sqrt{a-L}}$ being decreasing on $\intoo{L}{\infty}$, we get
\[
\int_{L+1}^\infty \frac{\sqrt{a} e^{-\nu a}}{\sqrt{a- L}} \diff a \leq \frac{\sqrt{L+1} }{\nu} e^{-\nu (L+1)} \pt
\]
Then, by the Hölder's inequality, we get
\[
\int_L^{L+1}\frac{\sqrt{a} e^{-\nu a}}{\sqrt{a- L}} \diff a 
\leq \sqrt{L+1} \left( \int_{0}^{1} a^{-\frac{p}{2}} \diff a\right)^{1/p} \left(\int_L^{L+1} e^{-\nu q a} \diff a\right)^{1/q} \, ,
\]
where we have choosen $p \in \intoo{1}{2}$ and $q \in \intoo{2}{\infty}$ such that $\frac{1}{p} + \frac{1}{q} = 1$. After some computations, we found
\[
\int_L^{L+1}\frac{\sqrt{a} e^{-\nu a}}{\sqrt{a- L}} \diff a \leq \frac{\left(1 - e^{-\nu q}\right)^{1/q}}{(\nu q)^{1/q}(1-p/2)^{1/p}} \sqrt{L+1} e^{-\nu L} \pt
\]
We could optimize with respect to $p$ and $q$ but it is not necessary. We inject this two previous estimate into \eqref{eq: nom_a_la_con} and we get
\[
J \leq \frac{C'}{(1+L)^{k - \frac{1}{2}}} e^{-\nu L} \, ,
\]
for some constant $C'$. Then we use the convexity/concavity inequalities \eqref{eq: convexity_inequality_1} or \eqref{eq: convexity_inequality_2} depending on $k$ to conclude the proof of the second part of \eqref{eq: lemma_technical_exponential_convolution}.

Now, we treat the $d=3$ case. The arguments are similar so we only give a sketch of the proof. To show left estimate of \eqref{eq: lemma_technical_exponential_convolution}, we have to bound the integral
\begin{align*}
I = \int_{\R^2} \frac{e^{-\nu  (\abs{y-\frac{x}{2}}+\abs{y+\frac{x}{2}})} }{(1 + \abs{y-\frac{x}{2}} + \abs{y+\frac{x}{2}})^k}\diff y  \pt
\end{align*}
For all $a\geq L/2$, the level set $\abs{y-\frac{x}{2}}+\abs{y+\frac{x}{2}}= 2a$ is a prolate spheroid (or ellipsoid of revolution) denotes $\mathcal{E}_{a,2}$. Its parameters are the distance to the poles $a$, the equatorial semi-axis $b = \sqrt{a^2-c^2}$ where $c = L/2$ and the eccentricity $e=c/a$. The surface of $\mathcal{E}_{a,2}$ is equal to
\begin{align*}
\int_{\mathcal{E}_{a,2}} \diff y = 2\pi \left(b^2 + \frac{ba}{e} \arcsin(e)\right) \pt
\end{align*}
Using the coarse upper bounds $\arcsin(x)/x \leq \pi/2$ and $b\leq a$, we have
\begin{align*}
\int_{\mathcal{E}_{a,2}} \diff y \leq 2\pi \left(1+\frac{\pi}{2}\right)a^2\pt
\end{align*}
We conclude by performing the same computations as in the two-dimensional case.

It remains to show the right estimate of \eqref{eq: lemma_technical_exponential_convolution} when the $d=3$. For this aim, we need to evaluate the integral $\int_{\mathcal{E}_{a,2}} \frac{\diff y}{\abs{y + \frac{x}{2}}}$. We introduce the parametrization
\[
\forall s \in \intff{0}{2\pi},~\forall t \in \intff{0}{\pi},~\mathbf{x}(s,t) = \begin{pmatrix}
a \cos (s)\cos (t) \\
a \cos (s) \sin (t) \\
b\sin (s)
\end{pmatrix}   \pt
\]
Then, we have
\begin{gather*}
\norme{\left(\frac{\partial \mathbf{x}}{\partial s} \wedge \frac{\partial \mathbf{x}}{\partial t}\right)(s,t)} = a \abs{\cos(s)} \sqrt{a^2 - c^2 \cos^2(s)} \, ,\\
\abs{\mathbf{x}(s,t) + \frac{x}{2}} = \sqrt{a^2 + c^2\cos^2(s) + 2 a c \cos(s) \cos(t)} \pt
\end{gather*}
We denote by $K$ the complete elliptic integral of the first kind. Then, we have
\begin{align*}
\int_{\mathcal{E}_{a,2}} \frac{\diff y}{\abs{y + \frac{x}{2}}}
= 4a \int_0^{\pi/2} \cos(s) \sqrt{\frac{a - c \cos(s)}{a+c\cos(s)}} K\left(\frac{2\sqrt{e\cos(s)}}{1+e \cos(s)} \right)\diff s \pt
\end{align*}
The maps $e \mapsto K\left(\frac{2\sqrt{e}}{1+e}\right)$ being increasing on $\intff{0}{1}$, we have
\begin{align*}
\int_{\mathcal{E}_{a,2}} \frac{\diff y}{\abs{y + \frac{x}{2}}}
\leq 4a K\left(\frac{2\sqrt{e}}{1+e}\right) \int_0^{\pi/2}  \sqrt{\frac{a - c \cos(s)}{a+c\cos(s)}} \diff s \leq \frac{4\pi a\sqrt{2a}}{\sqrt{a-c}} K\left(\frac{2\sqrt{e}}{1+e}\right) \pt
\end{align*}
Then, using the inequality $K(k) \leq \pi/2 - \ln(\sqrt{1-k^2})$ which is valid for any $k\in \intoo{0}{1}$, we obtain
\begin{align*}
\int_{\mathcal{E}_{a,2}} \frac{\diff y}{\abs{y + \frac{x}{2}}} \leq \frac{4\pi a\sqrt{2a}}{\sqrt{a-c}} \left(\frac{\pi}{2} + \ln \left(1+\frac{2c}{a-c}\right)\right) \pt
\end{align*}
The remaining computations are similar to the two-dimensional case. We just need to notice that $t \mapsto \frac{1}{\sqrt{t} \ln(t)}$ is integrable in the vicinity of $0$.
	

\end{document}